\documentclass[twocolumn]{IEEEtran}

\usepackage[utf8]{inputenc} 
\usepackage[T1]{fontenc}
\usepackage{url}
\usepackage{ifthen}
\usepackage{cite}
\usepackage[cmex10]{amsmath} 

\usepackage{graphicx} 
\usepackage{amssymb}
\usepackage{amsthm}
\usepackage{mathtools}
\usepackage{hyperref}
\usepackage{bbm}
\usepackage{tikz,lipsum,lmodern}
\usepackage[most]{tcolorbox}
\newtheorem{theorem}{Theorem}
\theoremstyle{definition}
\newtheorem{remark}{Remark}
\newtheorem{corollary}{Corollary}
\newtheorem{lemma}{Lemma}
\newtheorem{example}{Example}
\newtheorem{claim}{Claim}
\newtheorem{definition}{Definition}
\usepackage{algorithm}
\usepackage{algpseudocode}

\newcommand{\reg}{\mathsf{Reg}}

\newcommand{\pare}{\eta}

\newcommand{\g}{g(\underline{z})}
\newcommand{\gdot}{g(.)}
\newcommand{\acc}{\mathcal{A}_{\reg}}

\newcommand{\acce}{\mathcal{A}_{\pare}}
\newcommand{\E}{\mathbb{E}}
\newcommand{\bu}{\mathbf{u}}
\newcommand{\by}{\mathbf{y}}
\newcommand{\bn}{\mathbf{n}}
\newcommand{\mT}{\mathcal{T}}
\newcommand{\mH}{\mathcal{H}}
\newcommand{\mA}{\mathcal{A}}
\newcommand{\est}{\mathsf{est}}
\newcommand{\DC}{\mathsf{DC}}
\newcommand{\AD}{\mathsf{AD}}

\newcommand{\MMSE}{\mathsf{MMSE}}
\newcommand{\PA}{\mathsf{PA}}

\begin{document}
\title{Game of Coding: Beyond Honest-Majority Assumptions
\thanks{
  The work of Mohammad Ali Maddah-Ali has been partially supported by the National Science Foundation under Grant CIF-1908291.
  The work of Mohammad Ali Maddah-Ali and Hanzaleh Akbari Nodehi has been partially supported by the National Science Foundation under Grant CCF-2348638. The work of Viveck R. Cadambe has been  supported by the National Science Foundation under Grants 2506573 and 2211045 .
  }
  \thanks{
  This work has been partially presented in the 2024 IEEE International Symposium on Information Theory (ISIT 2024), Athens, Greece, July 7 to July 12, 2024.  
  }
} 


\author{%
  \IEEEauthorblockN{Hanzaleh Akbari Nodehi$^*$,}
  \and
  \IEEEauthorblockN{Viveck R. Cadambe$^\dagger$,}
  \and
  \IEEEauthorblockN{Mohammad Ali Maddah-Ali$^*$,\\}
$^*$University of Minnesota Twin                    Cities, $^\dagger$The Georgia Institute of Technology}


\maketitle
\begin{abstract}
Coding theory revolves around the incorporation of redundancy into transmitted symbols, computation tasks, and stored data to guard against adversarial manipulation. However, error correction in coding theory is contingent upon a strict trust assumption.  In the context of computation and storage, it is required that honest nodes outnumber adversarial ones by a certain margin. However, in several emerging real-world cases, particularly, in decentralized blockchain-oriented applications, such assumptions are often unrealistic. Consequently, despite the important role of coding in addressing significant challenges within decentralized systems, its applications become constrained.  Still, in decentralized platforms, a distinctive characteristic emerges, offering new avenues for secure coding beyond the constraints of conventional methods. In these scenarios,  the adversary benefits when the legitimate decoder recovers the data, and preferably with a high estimation error. This incentive motivates them to act rationally, trying to maximize their gains. In this paper, we propose a game theoretic formulation for coding, called the game of coding, that captures this unique dynamic where each of the adversaries and the data collector (decoder) have respective utility functions to optimize. The utility functions reflect the fact that both the data collector and the adversary are interested in increasing the chance of data being recoverable by the data collector. Moreover, the utility functions express the interest of the data collector to estimate the input with lower estimation error, but the opposite interest of the adversary. As a first, still highly non-trivial step, we characterize the equilibrium of the game for the repetition code with a repetition factor of 2 for a wide class of utility functions with minimal assumptions.  
\end{abstract}
\section{Introduction}
Coding theory is a fundamental approach to ensuring the integrity and reliability of communication, computing, and storage systems, safeguarding them against errors, failures, faults, and adversarial behaviors. 
While historically focused on discrete data (e.g., over finite fields) and exact recovery~\cite{SudanBook}, coding theory has also been expanded to the analog domain, allowing for approximate recovery~\cite{ZamirCoded, roth2020analog, jahani2018codedsketch,BACC}.

However, there exist some restrictive trust assumptions required for coding to protect the data against the error. For example, consider the ubiquitous Reed-Solomon (RS) code which is maximum distance separable, meaning that it meets the Singleton bound \cite{SudanBook}.
 Assume that we are using RS codes to encode $K$ input symbols to $N$ coded symbols, while $f$ of the coded symbols are delivered by an adversary\footnote{Errors in coding can come from many sources, e.g., hardware errors, etc. In this paper, we are focused on applications to decentralized systems like blockchains,  where the most important source of aberrant behavior is malicious adversarial nodes. } and  $h = N - f$ are delivered by honest parties. To recover data without any error, the condition $f \leq \frac{N - K}{2}$, or equivalently $f \leq h - K$, must hold. 
This requires an honest majority,
which is a significant trust assumption. 
If otherwise $K \leq h < f+K$, the decoder, as the data collector (DC), can detect the error, but not correct it.  Similar constraints hold in analog coding~\cite{roth2020analog}. 

In this paper, we target new applications of coding, such as decentralized (blockchain-oriented) machine learning (DeML) and oracles (See Section~\ref{sec:app} for more details). For such systems, a major limitation of existing coding frameworks is that the aforementioned trust limitation is hard to guarantee  in practical scenarios. Thus, error correction is impossible in regimes of interest. In other words, under conventional coding-theoretic models, the adversary can manipulate the data collector into rejecting the inputs and producing no output (denial of service attack).
Yet, such applications have incentive-centered design and thereby have unique dynamics that offer new avenues. Specifically, in such systems accepted contributions are rewarded, therefore adversaries are often incentivized to ensure that data is recoverable by the DC, rather than performing a denial of service attack. They would prefer that the system remains \emph{live} for the legitimate decoder, but the decoded value has a large error. In other words, the classical worst-case attacks studied in coding theory are too conservative for DeML applications and oracles, and miss the fact that the adversary also pays a penalty for executing such denial of service attacks. The main contribution of this paper is a new game-theoretic framework for coding called \emph{the game of coding} which aims to expand applications of coding techniques to blockchain-oriented applications including DeML and oracles. In this framework, each of the DC and the adversary, as the players, has a utility function to optimize. Importantly, rather than a worst-case adversary, we assume a rational adversary that maximizes its utility. In the context of analog coding, the utility functions of the players are functions of two metrics: (1) the probability that the DC accepts the input, (2) the error of estimating the original data  if inputs are accepted. Focusing on repetition coding with two nodes as the first step\footnote{We admit that a system with 2 nodes appears relatively simple, but even this simple-looking system offers much richness and technical complexity when studied from the game theoretic perspective that we develop.},  we characterize the equilibrium of the game.  Our results show that the game theoretic view opens new avenues of cooperation between the adversary and the data collector that enables liveness of the system even if the conventional trust assumptions (e.g., honest majority in the case of repetition codes) are not satisfied. We determine the optimum strategy of the data collector in accepting and rejecting the inputs, and the optimum noise distribution for the adversary that achieves the equilibrium. The result is broadly applicable as we make very minimal assumptions on the utility functions of the players.



\subsection{Applications and Motivation: Decentralized Platforms}
\label{sec:app}
Decentralized computing platforms are crafted to function without dependence on a central trusted entity. 
By operating as a form of state machine replication (SMR), blockchain consensus empowers decentralized computing platforms (e.g. Ethereum). Such computing platforms now support various applications beyond cryptocurrencies~\cite{bitcoin2008bitcoin, buterin2013ethereum}, including decentralized finance (DeFi), supply chain management, tokenization of assets, and identity management \cite{ruoti2019sok}. However, expanding the idea of decentralization to a broader set of powerful applications faces fundamental challenges, that the game of coding has the potential to resolve. Here we explain two major applications and their corresponding challenges, and how the game of coding can help to resolve those challenges. 

{\bf Application 1:  Decentralized Machine Learning (DeML) - The Challenge of Verifiable Computing:}
Machine learning is one of the major areas that researchers are actively aspiring to decentralize~\cite{shafay2023blockchain, ding2022survey, kayikci2024blockchain, taherdoost2023blockchain, taherdoost2022blockchain, bhat2023sakshi, tian2022blockchain, salah2019blockchain}. At the same time, many start-ups, e.g., \textsc{FedML}, \textsc{gensyn}, \textsc{KOSEN LABS}, \textsc{Ritual}, \textsc{EZKL}, \textsc{TOGHETHER},  are competing to develop the first realization of DeML. The challenge is that the underlying blockchain-oriented computation platforms are considered very weak. For example, Ethereum 2.0, recently unveiled following a substantial upgrade, is weaker than a standard desktop computer, albeit at a considerably higher cost.

A promising solution to address this issue involves moving the execution layer off-chain ~\cite{zhao2021veriml,thaler2022proofs}, while utilizing the blockchain for consensus (as an SMR), data availability (as a database), and settlement (for verification). In other words, resource-intensive computations are delegated to external entities. These external agents are tasked not only with \emph{executing the computations} but also with \emph{presenting cryptographic validity proofs} for verifying these computations. Because these validity proofs are verified over the blockchains, they must incur minimal computation overhead.  
However, solutions based on verifiable computing face two major challenges:

\begin{enumerate}
    \item \textbf{Resource Intensive Proof Generation:}  In verifiable computing schemes,  the process of generating proof is resource-intensive.  Specifically, in machine learning and artificial intelligence applications, generating proofs can demand more resources than the computation itself, making the approach non-scalable \cite{liu2021zkcnn, xing2023zero, weng2021mystique, mohassel2017secureml, lee2024vcnn}.
    The other issue is that in most of those schemes, in the proof generation process, the memory requirements tend to be exhaustive, which is an additional bottleneck~\cite{garg2023experimenting}.
    
    \item \textbf{Limited to Finite Field Operation and Exact Computation:} 
    Some of the most efficient verifiable computing schemes are typically devised by representing computations through arithmetic circuits defined over a finite field. When it comes to approximate computing over real numbers, especially for certain non-linear computations such as non-linear activation functions in deep neural networks or truncation for quantization, these processes don't naturally and efficiently translate into an arithmetic circuit. Studies like~\cite{weng2021mystique, chen2022interactive, garg2022succinct, setty2012taking} reveal the significant overhead incurred by
generating verifiable proof for computations involving fixed-point and floating-point arithmetic. \hspace*{\fill} $\blacksquare$
\end{enumerate} 

{\bf Application 2: Oracles -- The Problem of Unverifiable Inputs:}
One of the limitations of the blockchain is that smart contracts running on blockchains cannot directly access external information or data -- such as weather conditions or market prices.  Accessing external data is done through entities called \emph{oracles} ~\cite{eskandari2021sok, breidenbach2021chainlink, benligiray2020decentralized}. An oracle receives such information from some external agents, some of which are possibly controlled by the adversary. Moreover, even the data reported by honest agents are noisy, and thus not perfectly consistent. The problem is that honest nodes cannot generate any verifiable proof for the accuracy of their inputs.  \hspace*{\fill} $\blacksquare$

A solution to resolve these challenges is to leverage redundancy using error-correcting codes. For example, a smart contract representing a DeML application deployed on the blockchain can assign a computationally intensive task, such as multiple iterations of training, redundantly to $N$ external nodes. Subsequently, the contract collects the results and estimates the final outcome using the decoding algorithm of analog repetition codes. This approach can be further optimized by substituting repetition coding with more advanced coded computing techniques \cite{yu2017polynomial, jahani2018codedsketch, yu2019lagrange}. We note that in this approach, the external nodes do not need to provide any validity proof for their computation task. Rather, the decoding algorithm can detect and correct the errors. Moreover, the computation does not have to be exact and within a finite field. 
However, as mentioned, the downside of coding is that error correction is feasible only under very strong \emph{trust assumption}, which is often hard to provide at the application level of decentralized platforms~\cite{sliwinski2019blockchains, han2021fact, gans2023zero}. 

As another example, consider a scenario where some agents report a specific market value to an oracle (also referred to as the data collector (DC)). Adversaries can submit values severely contradicting those of honest agents. Under an honest-majority assumption, the oracle could use the \emph{median} to estimate the true value. However, in a trust-minimized scenario, where we don't have an honest majority, the oracle has no alternative but to reject all inputs.  In this case, adversaries lose potential \emph{rewards}, which are provided to encourage the agents to collaborate with the oracle. On the other hand, if the inputs are rejected, the oracle loses its \emph{liveness}. Thus, maintaining liveness is important for \emph{both} the oracle and adversaries. In addition, a live yet distorted oracle creates other opportunities, such as arbitrage opportunities, for adversarial agents. Consequently, the adversaries' utility functions are contingent on both the probability of the oracle's liveness and the estimation error of the oracle. On the other hand, the DC (here the oracle) has its own utility function that prioritizes both the probability of liveness and accuracy of the estimation. In the game of coding setup, the DC and adversarial nodes operate at an equilibrium with - possibly, depending on the utility functions -  a non-zero probability of liveness.

\subsection{Main Contributions of This Paper}
In summary, the main contributions of this paper are as follows:
\begin{enumerate}
    \item  We develop a new paradigm in coding theory by developing a novel game-theoretic framework to the study of code constructions\footnote{To the best of our knowledge, ours is the first to introduce a game-theoretic framework for coding theory.}. This framework leverages the mutual interest of both players—the adversarial nodes and the DC—in maintaining the system's liveness, despite their opposite objectives regarding estimation accuracy.   By incorporating the concepts of liveness and adopting a game-theoretic viewpoint,
    this framework offers a viable solution for trust minimized scenarios, where conventional coding theoretic analysis fails.  This framework expands the scope of coding theory, making it suitable for emerging applications such as decentralized machine learning.
    
    \item We ensure the full alignment and consistency with practical instantiations, like  decentralized machine learning by adopting a leader-follower setting known as the Stackelberg model. Here, DC is the leader and first commits its strategy by deploying a smart contract on the chain, while the adversary is the follower, selecting its best strategy in response to the DC's committed strategy. We characterize the equilibrium of the game   and determine the corresponding optimal strategies of the players, including acceptance rule for the DC and the noise distribution for the adversarial nodes. This is achieved using computationally feasible and straightforward algorithms. Importantly, our approach imposes minimal assumptions on the utility functions and the noise of the honest node, making it applicable to a wide range of practical applications.

    \item We demonstrate that at the equilibrium,  the system operates with an improved probability of liveness in a trust-minimized settings.  Recall that in those scenarios where there is only one honest node and multiple adversarial nodes, conventional coding theory yields a zero likelihood of liveness.  
\end{enumerate}

\subsection{Organization of The Paper}
The subsequent sections of this paper are structured as follows. In Section \ref{Informal Model}, we initiate by providing an informal definition of the problem formulation. This is further detailed formally in Section \ref{Formal Problem Settin}. Moving forward, Section \ref{Main results} delves into the main results, while the proof can be found in Sections \ref{proof:theorem: equivalence_two_problem}, \ref{proof:lemma:mean_is_near_optimal} and \ref{proof:lemma:best_noise}. Finally, in Section \ref{sec:conclusion}, we present our conclusion.

\subsection{Notation}
We use bold notation, for example, $\mathbf{y}$, to represent a random variable, while we don't use any bold symbol for a specific value of this random variable. For a case that a random variable $\mathbf{y}$ follows a uniform distribution on the interval $[-u, u]$, we use the notation $\mathbf{y} \sim \text{unif}[-u,u]$, where $u \in \mathbb{R}$. 
For any real numbers $x, a, b $, where $b > a$, the function $\underset{x}{\text{unif}}[a, b]$ is defined as $\frac{1}{b - a}$ when $a \leq x \leq b$, and it is equal to $0$ otherwise.
 The notation $[a]$ is the set of $\{1,\dots,a\}$, for $a \in \mathbb{N}$. 
 For any countable set $\mathcal{S}$, we denote its size as $|\mathcal{S}|$. 
 Each variable marked with an underline, such as $\underline{z}$, represents a vector. 
For any set $\mathcal{S}$ and an arbitrary function $f: \mathbb{R}^* \to \mathbb{R}$, the output of $\underset{x \in \mathcal{S}}{\arg\max} ~f (x)$ is a set comprising all elements $x$ in $\mathcal{S}$ that maximize $f (x)$. Similarly we define $\underset{x \in \mathcal{S}}{\arg\min} ~f (x)$.

\section{Problem formulation}
In this section, we first present a motivation for our problem set up and describe it informally. Subsequently, we introduce the formal problem formulation of the \emph{game of coding}, focusing on the scenario of a repetition code.  For simplicity of explanation, we present the informal set up for the scenario where there are just two nodes, whereas the formal set up considers the general case of arbitrary number of nodes.

\subsection{Informal Problem Setting}\label{Informal Model}
We consider a system (shown in Fig. \ref{fig:Two_node_model}), with two nodes, where one is honest, and the other is adversarial, and a DC. There is a random variable $\mathbf{u} \sim \text{unif}[-M,M]$, where $M \in \mathbb{R}$. The DC does not have any direct access to $\mathbf{u}$, rather it relies on the set of nodes to obtain an estimation of $\mathbf{u}$.

\begin{figure}[t]
    \centering
    \includegraphics[width=0.65\linewidth]{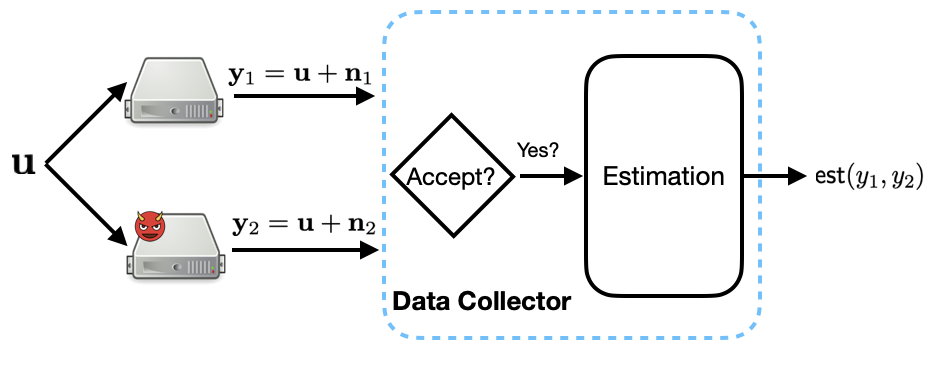}
\caption{Game of Coding for Repetition Coding with $N=2$. Node one is honest and the other one is adversarial. Each of these nodes reports a noisy version of $\mathbf{u}$ to DC. For the honest node, $\bn_1 \sim \text{unif}[-\Delta,\Delta]$, while for the adversarial node, $\mathbf{n}_2$ can have an arbitrary distribution $g(.)$, chosen by the adversary.  Upon receiving the data, the DC decides whether to accept or reject the inputs. If it is accepted, the DC outputs an estimation of $\mathbf{u}$. In this game, the data collector chooses the acceptance rule and the adversary chooses $g(.)$.
}

    \label{fig:Two_node_model}
\end{figure}

As an application, the variable $\mathbf{u}$ can be the result of some computation that the DC cannot run by itself. In this case, the nodes represent two servers, to which the computation has been offloaded. In another application, the variable $\mathbf{u}$ can also be some market value that the DC does not have access to.  The major challenge here is that one of the nodes is an adversary. We show the index of the honest node with $h \in \{1,2\}$ and the adversarial one with $a \in \{1,2\}$. 

The honest node sends $\by_h$ to the DC, where $\by_h = \bu + \bn_h$, $\bn_h \sim \text{unif}[-\Delta,\Delta]$, for some $\Delta \in \mathbb{R}$. In computing use cases, the noise $\bn_h$ represents the noise of rounding in fixed-point and floating-point computing or the noise of approximate computing (e.g. sketching, randomized quantization, random projection, random sampling, etc.), where exact computation is expensive or infeasible (see~\cite{han2023hyperattention} as an example on how approximate computing can reduce the computation load of large language model). In oracles, it represents the fact that honest nodes themselves have access to a noisy version of $\mathbf{u}$.

On the other hand, the adversarial nodes send $\by_a = \bu + \bn_a$ to the DC, where $\bn_a$ is some arbitrary noise independent of $\bu$. Indeed, the adversary has the option to choose the distribution of the noise $\bn_a$ as it wishes. We assume that the values of $M, \Delta$ are known by all players, and we have $\Delta \ll M$.

Upon receiving $(\by_1,\by_2)$, the DC decides to accept or reject the inputs. Then, the fundamental question is, under what conditions should the DC accept the inputs? We note that even when both of the nodes are honest, still $\by_1$ and $\by_2$ would be different, however, due to the distribution of $\bn_h$, $|\by_1 - \by_2| \leq 2\Delta$. This observation suggests an initial rule of acceptance in which the inputs are accepted if  $|\by_1 - \by_2| \leq 2\Delta$.  Although this narrow acceptance region seems reasonable at first sight, it may not well align with the interests of the DC. The reason is that it rejects any pair of inputs, even if that pair yields an \emph{accurate-enough} estimation of $\mathbf{u}$. 
Why does the DC kill the liveness of the system in this case?  
There can be a situation in which the adversary is unwilling to provide $\mathbf{u}$ within $\pm \Delta$ accuracy, and yet, its input can help the DC to estimate $\bu$ with enough accuracy. Specifically, the DC may be willing to trade the accuracy of the estimation to increase the liveness of the system. Liveness is also very important for the adversary for at least two reasons. First, in systems that incentivize good behavior, the adversary will be rewarded if the inputs are accepted. Second, the adversary can influence the estimated value only if the inputs are accepted and the system is live. If the system is not live, the adversary may lose the opportunity to influence the accepted value in the system.  This makes the adversary to act \emph{rationally} in choosing the distribution of $\bn_a$, considering its overall benefit.
Likewise, the DC is motivated to modify the acceptance region to $|\by_1 - \by_2| \leq \eta \Delta$, for some $\eta \in \mathbb{R}$, larger than 2. This motivation raises two important questions: what is the optimal choice for $\eta$? How the adversary should choose the distribution of $\bn_a$? To comprehensively address these questions, we introduce a game-theoretic formulation, called \emph{game of coding}. This framework can be used for general coding, while here, as the first step, we focus on two-node repetition coding.

In this formulation, the DC and the adversary have utility functions, respectively denoted by $\mathsf{U}_{\DC}(g(.), \eta)$ and $\mathsf{U}_{\AD}(g(.), \eta)$, to optimize, where $g(.)$ is a noise distribution that the adversary is utilizing. Let  $\mathcal{A}_{\eta}$  be the event that inputs are accepted, i.e.,  $\mA_{\eta}=\{ |\by_1-\by_2|\leq \eta \Delta\}$. Thus $\PA(g(.), \eta) \triangleq \Pr(\mathcal{A}_{\eta})$ is the probability of acceptance. 
In addition, let  $\mathsf{MMSE}(g(.), \eta)= \underset{\est}{\min}\E[(\hat{\bu}-\bu)^2 | \mathcal{A}_{\eta}]$, where $\hat{\bu}=\est(\by_1,\by_2)$ is the estimation of $\bu$, using some estimation function $\est(.,.)$, if $(\by_1,\by_2)$ is accepted. The utility function $\mathsf{U}_{\DC}(g(.), \eta)$ of the DC is an increasing function of the probability of acceptance and a decreasing function of $\mathsf{MMSE}$. In addition,  the utility function $\mathsf{U}_{\AD}(g(.), \eta)$ of the adversary is an increasing function of the probability of acceptance and $\mathsf{MMSE}$. For example, $\mathsf{U}_{\DC}(g(.), \eta)=-\log(\MMSE(g(.), \eta))+ \log(\PA(g(.), \eta))$ and  $\mathsf{U}_{\mathsf{AD}}(g(.), \eta)=\log(\MMSE(g(.), \eta))+ 5\log(\PA(g(.), \eta))$.

In this arrangement, players would opt to play at the (mixed) Nash equilibrium,  where each player's strategy is optimal given the strategies of others. However, in a typical application of the game of coding,  the DC represents a smart contract, and thus it commits to a strategy first, as the code is deployed on the blockchain. The adversary observes the committed strategy of the DC, and chooses its strategy based on what the DC has committed to. This scenario falls to a special category of games, known \emph{Stackelberg games} \cite{von2010market}, where one of the players as leader, here the DC,  takes the initiative in committing to a strategy, and then the other player as the follower, here the adversary, responds based on the leader's actions.  

Then, the objectives of the game of coding are (i) to characterize the Stackelberg equilibrium of the game, (ii) determine the optimum choice of $\eta$ for the data collector and (ii) determine the optimum choice of noise distribution $g(.)$ for the adversary. 
\subsection{Formal Problem Setting}\label{Formal Problem Settin}
In this section, we describe the problem formulation formally.
We emphasize that the focus of the main results is on the specific case of two nodes and the region of acceptance is chosen from a particular but important class of regions $\{ |\by_1-\by_2|\leq \eta \Delta\}$. However, since the general problem formulation of the \emph{game of coding} is new, in this section,  we provide a general and formal problem formulation as a reference for further discussion.

We consider a system including $N\in \mathbb{N}$ nodes and a DC. There is a random variable $\mathbf{u} \sim \text{unif}[-M,M]$, where $M \in \mathbb{R}$. The DC aims to have an estimation of $\mathbf{u}$, however, it does not have any direct access to it. Rather it relies on the set of nodes to receive some information about $\bu$.  The random variable $\bu$ can originate, for example, from outsourced computations involving $N$ nodes (in the setup of distributed computing) or from external information inaccessible to the DC (in an oracle).
 However, not all the nodes are honest.  Indeed, a reason that the DC relies on $N$ nodes to report identical computation or information is to protect the accuracy of estimating $\bu$ in the presence of the adversarial nodes. 
 
The set of nodes $[N]$ is partitioned into two sets, honest nodes $\mH$ and adversarial nodes, $\mT$, where $\mH, \mT \subseteq[N]$, and $\mathcal{H} \cap \mT= \emptyset$. 
We assume that $|\mT| =f$, for some integer $f <N$.  
The subset $\mathcal{T}$ is selected uniformly at random from the set including all subsets of size $f$ of $[N]$. 
Neither the DC nor the set of nodes in $\mathcal{H}$ is aware of $\mT$.

An honest node $h \in \mathcal{H}$  sends $\by_h$ to the DC, where
\begin{align}
    \by_h=\bu+\bn_h, \ \ h \in \mathcal{H},
\end{align}
where the noise $\bn_h$ has a symmetric probability density function (PDF) 
$f_{\bn_h}$ over the bounded range $[-\Delta, \Delta]$, for some $\Delta \in \mathbb{R}$.   This noise represents the noise of approximate computing (e.g. sketching, randomized quantization, random projection, random sampling, etc.), where exact computation is expensive. In oracles, it represents the fact that honest nodes themselves have access to a noisy version of $\bu$. 
For simplicity, we assume that $\{\bn_h \}_{h\in \mathcal{H}}$ are independent and identically distributed (i.i.d.).


On the other hand, each adversarial node $a\in \mT$ sends some function, possibly randomized, of $\bu$ to the DC, 
\begin{align}
    \by_a=\bu+\bn_a, \ \ a \in \mT,
\end{align}
where $\{\bn_a\}_{a \in \mT} \sim g(\{n_a\}_{a \in \mT })$, for some joint PDF $g(.)$. The adversary chooses $g(.)$ and   
 the DC is unaware of it.  Here we assume that the adversary has access to the exact value of $\bu$. We assume that the PDFs of $\bu$ and $\bn_h$  are known by all of the players, and $\Delta \ll M$.

The DC receives $\underline{\by} \triangleq (\by_1, \dots, \by_N)$, and evaluates the results according to the following steps (See in Fig.~\ref{fig:General Model}):

\begin{figure}[t]
    \centering
    \includegraphics[width=0.65\linewidth]{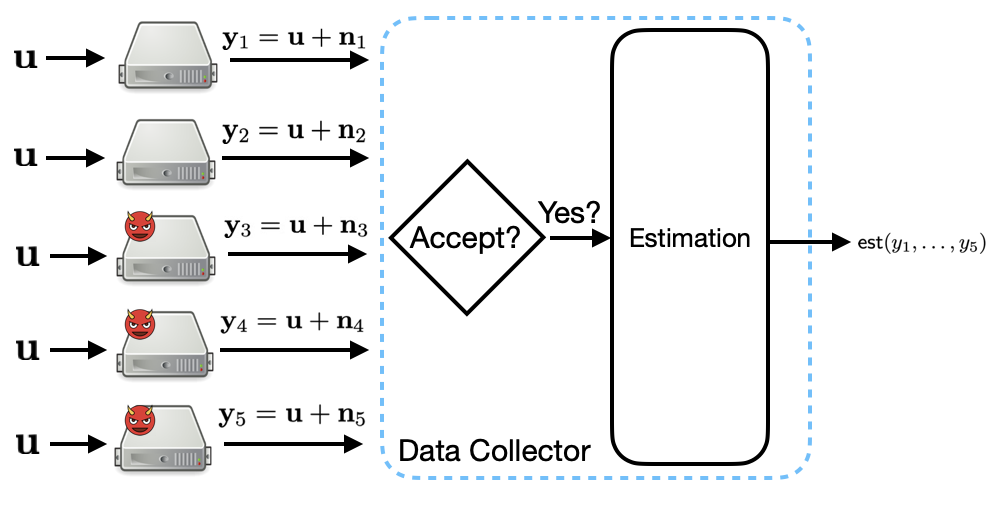}
\caption{
In this figure, there are a total of $N=5$ nodes, where $f=3$ of them are adversaries, depicted by red ones in the figure. Each node's task is to output $\mathbf{u}$. However, this process is affected by noise. The noise for honest nodes is 
given by $\mathbf{n}_h \sim \text{unif}[-\Delta, \Delta]$ for $h=1,2$. Conversely, for adversarial nodes, we have $[\mathbf{n}_a]_{a \in \mathcal{T}} \sim \gdot$, where $\gdot$ is an arbitrary distribution assumed to be 
independent of $\mathbf{u}$, and $\mathcal{T} = \{3,4,5\}$. Upon receiving the data, the DC checks whether $\underline{y}$ falls within the region $\reg$ or not. If not, it rejects it; otherwise, it accepts it and outputs 
$\est(\underline{y})$ as its estimation.}
    \label{fig:General Model}
\end{figure}

\begin{enumerate}
\item \textbf{Accept or Reject}: The DC considers an acceptance region $\reg \subseteq \mathbb{R}^N$, and accepts $\underline{\by}$, if and only if $\underline{\by} \in \reg$.
We denote by $\mA_{\reg}$, the event of the inputs being accepted, that is, the event of  $\underline{\by} \in \reg$. In addition, we denote 
\begin{align}
\mathsf{PA} \left( \gdot, \reg \right)=\Pr(\mA_{\reg}).   
\end{align}

\item \textbf{Estimation}: If the inputs are accepted the DC outputs $\est(\underline{\by})$  
 as its estimate of $\bu$,  where $\est: \mathbb{R}^N \to \mathbb{R}$. We define the cost of estimation as
\begin{align*}
\mathsf{MSE}\left(\est(.,.), \gdot, \reg\right) \triangleq \mathbb{E}\left[\big(\mathbf{u} - \est(\mathbf{\underline{y}})\big)^2 | \acc\right],
\end{align*}
which is the mean square error (MSE) for the estimation function $\est$, given that observed data has been accepted.
The DC chooses $\est$ to minimize the mean square error. We define the 
 minimum mean square estimator  as
 \begin{align}\label{definition_best_estimator}
 \est^*_{\reg,g} \triangleq \underset{\est: \mathbb{R}^N \to \mathbb{R}}{\arg\min} ~\mathsf{MSE}\big(\est(.,.), \gdot, \reg\big),
 \end{align}
 and also the 
 minimum mean square estimation error (MMSE)  as
 \begin{align}\label{MMSE_definition}
\mathsf{MMSE}\left(\gdot, \reg \right) \triangleq \mathsf{MSE}\left( \est^*_{\reg,g}(.,.), \gdot, \reg \right).
 \end{align}
 
\end{enumerate}
 We consider a two-player game with DC and the adversary as the players.   Each player has a utility function that seeks to maximize. The utility function of the DC is denoted by 
\begin{align*}
\mathsf{U}_{\mathsf{DC}}\left( \gdot, \reg \right) \triangleq Q_{\mathsf{DC}} \left( \mathsf{MMSE}\big(\gdot, \reg \big), \mathsf{PA} \left( \gdot, \reg \right)\right),
\end{align*}
where $Q_{\mathsf{DC}}: \mathbb{R}^2 \to \mathbb{R}$, is a non-increasing function with respect to its first argument and a non-decreasing function with respect to the second one.  More precisely, for a fixed value for the second argument, $Q_{\mathsf{DC}}(.,.)$ is monotonically non-increasing in the first argument. Similarly, for a fixed value for the first argument, $Q_{\mathsf{DC}}(.,.)$ is monotonically non-decreasing in the second argument.
The DC chooses its strategy from the action set  $\Lambda_{\mathsf{DC}} = \left\{ \reg | ~ \reg \subseteq \mathbb{R}^N \right\}$. 

The utility function of the adversary is denoted by 
\begin{align*}
\mathsf{U}_{\mathsf{AD}}\left( \gdot, \reg \right) \triangleq Q_{\mathsf{AD}} \left( \mathsf{MMSE}\big(\gdot, \reg \big), \mathsf{PA} \left( \gdot, \reg \right)\right),
\end{align*}
where  $Q_{\mathsf{AD}}: \mathbb{R}^2 \to \mathbb{R}$ is a strictly increasing function with respect to both of its arguments. More precisely, for a fixed value for the second argument, $Q_{\AD}(.,.)$ is strictly increasing in the first argument. Similarly, for a fixed value for the first argument, $Q_{\AD}(.,.)$ is strictly decreasing in the second argument.
The adversary chooses its strategy from its action set,  encompassing all possible noise distributions, i.e., $\Lambda_{\mathsf{AD}} = \left\{ \g |  ~\text{where} ~g:\mathbb{R}^f \to \mathbb{R} ~\text{is a valid PDF} \right\}$.

For this game, we aim to determine a \textbf{Stackelberg equilibrium}. Specifically, we assume that the DC plays the role of the leader, while the adversary acts as the follower. For each acceptance region $\reg \in \Lambda_{\mathsf{DC}}$ to which  the DC is committing, we define the set of best responses of the adversary as
\begin{align}
    \mathcal{B}^{\reg}_{\mathsf{AD}} \triangleq  \underset{\gdot \in \Lambda_{\mathsf{AD}}}{\arg\max} ~ {\mathsf{U}}_\mathsf{AD}\big(\gdot, \reg \big) .
\end{align}
 For every $\reg \in \Lambda_{\mathsf{DC}}$ that DC commits to, the adversary may choose any of $g^*(.) \in \mathcal{B}^{\reg}_{\mathsf{AD}}$ because each element of $\mathcal{B}^{\reg}_{\mathsf{AD}}$  has the same utility for the adversary. However, different elements of  may have different utility for the DC.
 For every $\reg \in \Lambda_{\mathsf{DC}}$ that DC commits to, the adversary may choose any of $g^*(.) \in \mathcal{B}^{\reg}_{\mathsf{AD}}$ because each element of $\mathcal{B}^{\reg}_{\mathsf{AD}}$  has the same utility for the adversary. However, different elements of $\mathcal{B}^{\reg}_{\mathsf{AD}}$  may have different utilities for the DC.
 We define
 \begin{align}\label{chaiveble_best_response}
    \Bar{\mathcal{B}}^{\reg}_{\mathsf{AD}} \triangleq  \underset{\gdot \in \mathcal{B}^{\reg}_{\mathsf{AD}}}{\arg \min} ~ {\mathsf{U}}_\mathsf{DC}\left(\gdot, \reg \right).
 \end{align}

 \begin{definition}
      For any $\gdot \in \Lambda_{\AD}$ and $\reg \in \Lambda_{\DC}$, a pair $(\gdot, \reg)$ is an achievable pair, if and only if we have  $\gdot \in \Bar{\mathcal{B}}^{\reg}_{\mathsf{AD}}$.  
 \end{definition}
Note that if we consider a $\reg \in \Lambda_{\DC}$ and fix it, then for all $\gdot \in \Bar{\mathcal{B}}^{\reg}_{\mathsf{AD}}$, the value of ${\mathsf{U}}_\mathsf{DC}\left(\gdot, \reg \right)$ remains the same. 
We define
\begin{align}\label{stackleberg-eqili}
    \reg^* = \underset{\reg \in \Lambda_{\mathsf{DC}}}{\arg\max}  ~ {\mathsf{U}}_\mathsf{DC}\left(\gdot, \reg \right),
\end{align}
where in the above $\gdot$ is an arbitrary element in $\Bar{\mathcal{B}}^{\reg}_{\mathsf{AD}}$.
We define the notion of  Stackelberg equilibrium as follows.
\begin{definition}\label{def:stack_equil}
    For any $g^*(.) \in \Bar{\mathcal{B}}^{\reg^*}_{\mathsf{AD}}$, we call the pair of $\left(
{\mathsf{U}}_\mathsf{DC}\left(g^*(.), \reg^* \right), {\mathsf{U}}_\mathsf{AD}\left(g^*(.), \reg^* \right)
    \right)$ as a Stackelberg equilibrium. Interchangeably,  we call  $\left( \mathsf{MMSE}\big(g^*(.), \reg^* \big), \mathsf{PA} \left( g^*(.), \reg^* \right)\right)$ or  $(\mathsf{MMSE}^*, \PA^*)$ as a Stackelberg equilibrium.
\end{definition}
\begin{remark}
        To support limiting cases, we extend Definition \ref{def:stack_equil} as follows.   Let  $\{(g_n(.), \reg_n)\}_{n=1}^\infty$ be a sequence of achievable pairs, where the corresponding 
         pairs $\left(
{\mathsf{U}}_\mathsf{DC}\left(g_n(.), \reg_n \right), {\mathsf{U}}_\mathsf{AD}\left(g_n(.), \reg_n \right)
    \right)$ converges to $(\mathsf{U}^*_\mathsf{DC}, \mathsf{U}^*_\mathsf{AD})$ as $n \rightarrow \infty$.  
        Then, $(\mathsf{U}^*_\mathsf{DC}, \mathsf{U}^*_\mathsf{AD})$ is called a Stackelberg equilibrium of the game, if  
$\mathsf{U}^*_{\mathsf{DC}} \geq \mathsf{U}_{\mathsf{DC}}(\gdot, \reg)$, for any achievable  pair  $(\gdot, \reg)$.

    \end{remark}
    
    The objective is to characterize  Stackelberg equilibrium of the game.

\begin{remark}
    Using game theory to design secure systems and model the interplay between different players has been widely adopted in other areas \cite{han2010physical, saraydar2002efficient, wang2008distributed, bonneau2008non}. An alternative approach is to define the game in a Nash framework, where all players simultaneously select their strategies at equilibrium, with each player’s strategy being the optimal response to the strategies chosen by the others. However, in the game of coding, the leader-follower framework—known as the Stackelberg game—emerges as the natural and more suitable choice. In our setting, the DC first commits to its strategy by deploying a smart contract on the blockchain, thereby revealing its decision rule. The adversary then chooses its strategy in response to the DC’s committed strategy, effectively maximizing its utility with respect to that choice.

\end{remark}

\begin{remark}
    Our problem formulation is based on repetition coding. It is straightforward to generalize this formulation to 
    accommodate general coding schemes, however, the characterization of their Stackelberg equilibrium is an area of future work. Also, in this paper we constrain the acceptance rules to the form of $|\by_1-\by_2| \leq \eta \Delta$.  In a forthcoming paper, we will demonstrate that general acceptance regions can be replaced by a mixed strategy involving  $|\by_1-\by_2| \leq \eta \Delta$, where $\eta$ is chosen according to a distribution, without losing the performance.
\end{remark}    

\begin{remark}
    In this paper, we consider the case with one honest node and one adversarial node. A natural question that arises is: what is the effect of having \(N>2\) nodes? In our subsequent work \cite{akbari2024game}, we generalize the game of coding framework to scenarios with \(N \geq 2\) nodes, exploring critical aspects of system behavior. In particular, we show that the adversary’s utility at equilibrium does not increase with the number of nodes under their control. This result implies that adversaries gain no strategic advantage by duplicating themselves (i.e., launching Sybil attacks), thereby demonstrating that the framework is inherently Sybil-resistant. This property is essential for maintaining robustness in larger decentralized systems where trust assumptions may not hold. Furthermore, note that in this paper, and also in \cite{akbari2024game}, we assume that the game is complete, meaning that both the DC and the adversary know each other’s utility functions. While it is reasonable to assume that the adversary, as a powerful player, is aware of the DC's utility function, the reverse is often unrealistic. In many practical scenarios, the DC typically lacks knowledge of the adversary's utility function. This lack of perfect information, known as the \emph{game of incomplete information}, introduces uncertainty for the DC in choosing the optimal strategy. In \cite{akbarinodehi2025game}, we address this limitation by developing an algorithm that enables the DC to commit to a strategy that approximates the equilibrium, even without full knowledge of the adversary’s utility function.

\end{remark}

\section{Main Results}\label{Main results}
In this section, we present the main results of the paper, alongside an illustrative example for clarity.  In this paper, we restrict ourselves to a scenario where $N=2$, $f=1$, and the action set of the data collector is given by $\Lambda_{\mathsf{DC}} = \big\{ \pare ~| ~ \eta \geq 2 \big\}$.  Define the region $\reg(\eta) \triangleq \{ (y_1,y_2)~| ~|y_1-y_2| \leq \eta \Delta\}$. For notational simplicity, we  use the parameter $\eta$ to represent entire region $\reg(\eta)$, whenever it is clear from the context.
 
   Based on \eqref{chaiveble_best_response} and \eqref{stackleberg-eqili}, we can write
\begin{align}
    \eta^* 
     =\underset{\pare \in \Lambda_{\mathsf{DC}}}{\arg\max} ~ \underset{\gdot \in \mathcal{B}^{\eta}}{\min} ~ Q_{\mathsf{DC}} \left( \mathsf{MMSE}\left(\gdot, \pare \right), \mathsf{PA} \left( \gdot, \pare \right)\right). \label{eq:etastar}
\end{align}
At the outset, (\ref{eq:etastar}) represents a complicated optimization problem because: (i) It depends on the utility functions of the adversary and the DC, and we make no assumptions on these utility functions apart from the fact that the former is strictly increasing in both arguments and the latter is non-increasing in the first and non-decreasing in the second, and (ii) it is an infinite dimensional optimization problem over the space of all possible noise distributions $g(.)$ of the adversary. 

We circumvent issue (i) through an intermediate optimization problem that is independent of the players' utility functions. For each $0 < \alpha \leq 1$, consider the following optimization problem referred to as Opt. 1.
\begin{align}\label{C_definition}
    \textrm{Opt. 1:}   \quad c_{\eta} (\alpha) \triangleq 
    \underset{\gdot \in \Lambda_{\mathsf{AD}}}{\max} ~ \underset{\mathsf{PA} \left( \gdot, \pare \right) \geq \alpha}{\mathsf{MMSE}\left(\gdot, \pare \right)}.
\end{align}

We show that the equilibrium of the Stackelberg game can be readily obtained - using knowledge of $c_{\eta}(\alpha)$ as defined in (\ref{C_definition}) - through a two-dimensional optimization problem outlined in Algorithm \ref{Alg:finding_eta}. Algorithm \ref{Alg:finding_eta} takes utility functions $Q_{\mathsf{AD}}(., .), Q_{\mathsf{DC}}(., .)$, and $c_{\eta}(.)$ - the result of OPT. 1 - as the inputs and outputs $\hat{\eta}$. The following theorem establishes the correctness of Algorithm \ref{Alg:finding_eta}.

\begin{theorem}\label{theorem: equivalence_two_problem}
    Let $\hat{\eta}$ be the output of Algorithm \ref{Alg:finding_eta}. We have
    $\eta^* = \hat{\eta}$.
\end{theorem}
    The theorem is proved in  Section \ref{proof:theorem: equivalence_two_problem}.

\begin{algorithm}[t]
\caption{Finding the optimal Decision Region}
\label{Alg:finding_eta}
\begin{algorithmic}[1]
\State \textbf{Input:} Functions $Q_{\mathsf{AD}}(., .), Q_{\mathsf{DC}}(., .)$, and $c_{\eta}(.)$
\State \textbf{Output:} $\hat{\eta}$

\vspace{1em} 

\State \textbf{Step 1:}
\State Calculate the set $\mathcal{L}_{\eta} = \underset{0 < \alpha \leq 1 }{\arg\max} ~Q_{\mathsf{AD}}(c_{\eta} (\alpha), \alpha)$

\State \textbf{Step 2:}
\State Calculate $\hat{\eta} = \underset{\pare \in \Lambda_{\mathsf{DC}}}{\arg\max} ~ \underset{\alpha \in \mathcal{L}_{\eta}}{\min} ~ Q_{\mathsf{DC}} \left(c_{\eta} (\alpha), \alpha\right)$
\end{algorithmic}
\end{algorithm}

Algorithm \ref{Alg:finding_eta} is a two-dimensional optimization problem which is computationally feasible for well-behaved utility functions\footnote{For example, if the utility functions are Lipschitz-smooth, then a simple two-dimensional grid search yields an $\epsilon$-approximation of the optimal in $O(1/\epsilon^2)$ complexity.}.  Thus,  much of the technical challenges of finding $\eta^*$ lies in obtaining a characterization of the function $c_{\eta}(.)$. 
Theorem \ref{theorem:Main_Bound_For_Max_Problem} resolves this task. We denote the cumulative distribution function (CDF) of $\bn_h$ by $F_{\bn_h}$ and the PDF of $\bn_h$ by $f_{\bn_h}$. Recall that $\Pr (|\bn_h| > \Delta) = 0$, for some $\Delta \in \mathbb{R}$. This implies that  $F_{\bn_h}(-\Delta) = 0$, and $F_{\bn_h}(\Delta) = 1$. We assume that $F_{\bn_h}(.)$ is a strictly increasing function in $[-\Delta, \Delta]$, i.e., for all $-\Delta \leq a < b \leq \Delta$, we have 
$F_{\bn_h}(a) < F_{\bn_h}(b)$.
\begin{theorem}\label{theorem:Main_Bound_For_Max_Problem}
    For any $\pare \in \Lambda_{DC}$, 
    and $0 < \alpha \leq 1$, we have 
    \begin{align}
      \frac{h^*_{\eta}(\alpha)}{4\alpha} -\frac{(\eta^2+4)(\eta+2)\Delta^3}{M}   \leq c_{\eta}(\alpha) 
       \leq \frac{h^*_{\eta}(\alpha)}{4\alpha},
    \end{align}
    where $h^*_{\eta}(q)$ is the concave envelop of the function\footnote{Note that since we assumed $F_{\bn_h}(.)$ is a strictly increasing function in $[-\Delta, \Delta]$, the inverse function of $k_{\eta}(.)$ exists.} $ h_{\eta}(q) \triangleq \nu_{\eta}(k_{\eta}^{-1} (q))$, $0 \leq q \leq 1$,  for $\nu_{\eta}(z) \triangleq \int_{z-\eta\Delta}^{\Delta} (x+z)^2f_{\bn_h}(x)\,dx$ and $k_{\eta}(z) \triangleq \int_{z-\eta\Delta}^{\Delta} f_{\bn_h}(x)  \,dx$, $(\eta-1)\Delta \leq z \leq  (\eta+1)\Delta$.
\end{theorem}

    The details of the proof is in Section \ref{proof:theorem:Main_Bound_For_Max_Problem}. Here, we make a couple of remarks.
    \begin{remark}  Notably, Theorem~\ref{theorem:Main_Bound_For_Max_Problem} makes minimal assumptions on the noise distribution of the honest node and therefore characterizes the  $c_{\eta}(\alpha)$ for a broad range of noise distributions. In particular, the theorem only assumes that the  honest node's noise distribution is symmetric and   bounded with a strictly increasing cumulative distribution function (CDF).  
    
    For the specific case of uniform distribution, i.e., $\mathbf{n}_h \sim \text{unif}[-\Delta, \Delta]$, we specialize Theorem \ref{theorem:Main_Bound_For_Max_Problem} in Appendix \ref{Characterizing_h_eta}, where we calculate $h_{\eta}(q)$ and  $h^*_{\eta}(q)$. Specifially, we show that for  $0 \leq q \leq 1$,  
    \begin{align}
    h_{\eta}(q)  = \Delta^2 \frac{(\eta+2-2q)^3 - (\eta+2-4q)^3}{6}.
\end{align}
    Also, if $\eta \geq \frac{8}{3}$ or $0 \leq q \leq \frac{9\eta+4}{28}$, we have $h^*(q) = h(q)$. Otherwise, we have 
    \begin{align}
        h^*_{\eta}(q) = \frac{h_{\eta}(1) - h_{\eta}(\frac{9\eta+4}{28})}{1 - \frac{9\eta+4}{28}} (q - \frac{9\eta+4}{28}) + h_{\eta}(\frac{9\eta+4}{28}).
    \end{align}
\end{remark}

\begin{remark}
    Notably, in Theorems \ref{theorem: equivalence_two_problem} and  \ref{theorem:Main_Bound_For_Max_Problem}, and also Algorithms \ref{Alg:finding_eta}, \ref{Alg:finding_noise}, we make no assumptions on the utility functions of the adversary and the DC, apart from the requirement that the adversary's utility is strictly increasing in both arguments and the DC's utility is non-increasing in the first and non-decreasing in the second. 
\end{remark}

\begin{remark}
    The approximation in Theorem~\ref{theorem:Main_Bound_For_Max_Problem} becomes tight as $\frac{(\eta\Delta)^3}{M} \rightarrow 0$. Therefore, the theorem statement is relevant for the case of $\Delta \ll M.$
\end{remark}

The objective of  Theorem~\ref{theorem:Main_Bound_For_Max_Problem} is to characterize Opt. 1. A highly non-trivial aspect of  Opt. 1 is in computing $\mathsf{MMSE}\left(\gdot, \pare \right)$ -- since we need the $\mathsf{MMSE}$ estimator which itself depends of the adversary's noise distribution $g(.)$, and in principle, we have to explore the space of all possible distributions $g(\cdot)$. Our approach towards resolving this is rooted in a powerful discovery: that, upto an approximation factor that is explained in the discussion that follows Lemma \ref{lemma:mean_is_near_optimal} - the \emph{mean} is universally optimal estimator for all possible symmetric adversarial distributions! 
Specifically towards proving Theorem \ref{theorem:Main_Bound_For_Max_Problem}, we  introduce an alternative optimization problem in which the estimator is simply $\mathsf{mean}(y_1,y_2)=\frac{y_1+y_2}{2}$, as follows:
\begin{align}\label{definition_beta_eta}
 \textrm{Opt. 2:}   \quad \beta_{\eta}(\alpha) = \underset{\gdot \in \Lambda_{\mathsf{AD}}}{\max} ~ \underset{\mathsf{PA} \left( \gdot, \pare \right) \geq \alpha}{\mathsf{MSE}\left(\mathsf{mean}(.,.), \gdot, \pare \right)}. 
\end{align}
Introducing Opt. 2, we prove Theorem~\ref{theorem:Main_Bound_For_Max_Problem} through a two-step process, establishing twin Lemmas~\ref{lemma:mean_is_near_optimal} and \ref{lemma:best_noise}.
In Lemma~\ref{lemma:mean_is_near_optimal}, as follows, we show that Opt. 2 is closely related to Opt. 1. More precisely, in Lemma~\ref{lemma:mean_is_near_optimal}, we show that the result of Opt. 2 sandwiches the result of Opt. 1 within a gap that vanishes as $\frac{\Delta^3}{M} \to 0$.

\begin{lemma}\label{lemma:mean_is_near_optimal}
    For any $0 < \alpha \leq 1$, we have 
    \begin{align}
       \beta_{\eta}(\alpha) - \frac{(\eta^2+4)(\eta+2)\Delta^3}{M} \leq  c_{\pare}(\alpha)  \leq \beta_{\eta}(\alpha)
    \end{align}
\end{lemma}
     The details of the proof is in Section \ref{proof:lemma:mean_is_near_optimal}. 
Then, in Lemma~\ref{lemma:best_noise}, as follows, we fully characterize $\beta_{\eta}(\alpha)$. 

 \begin{lemma}\label{lemma:best_noise}
    For any $0 < \alpha \leq 1$, we have
    \begin{align}\label{bound_of_mean}
        \beta_{\eta}(\alpha) =  \frac{h^*_{\eta}(\alpha)}{4\alpha},
\end{align}
where $h^*_{\eta}(.)$ is defined in Theorem \ref{theorem:Main_Bound_For_Max_Problem}.
\end{lemma}
     The detailed proof of this lemma is in Section \ref{proof:lemma:best_noise}. 

One of the main challenges in Lemma \ref{lemma:mean_is_near_optimal} is to find $\mathsf{MMSE}$ estimator. We establish that,  for any symmetric noise distribution of the adversary and honest node, and $|y_1+y_2| \leq 2M - (\eta+2)\Delta$, we have $\est^*_{\pare,g}\big( \underline{y} \big)  = \frac{y_1+y_2}{2}$.
This is surprising result, given that the noise of the honest and adversary nodes have arbitrary  symmetric  distribution.  The approximation in Lemma \ref{lemma:mean_is_near_optimal} comes from the boundary effect, where $|y_1+y_2| > 2M - (\eta+2)\Delta$. The proof of Lemma \ref{lemma:best_noise} is also technically involved, particularly to study $\beta_{\eta}(\alpha)$ for any general noise distribution for the honest node, as long as it is symmetric, bounded, and with a strictly increasing CDF. 
    

    

In the process of proving Theorem \ref{theorem:Main_Bound_For_Max_Problem}, we in fact characterize an optimal noise distribution $g^*(.)$ as stated in  Algorithm \ref{Alg:finding_noise}. This algorithm takes $Q_{\mathsf{AD}}(., .)$, $f_{\bn_h}(.)$, $\eta^*$, and $c_{\eta^*}(.)$ as inputs and outputs $g^*(z)$, which is the best noise distribution of the adversary.

 In Step $1$ of Algorithm \ref{Alg:finding_noise}, if $|\mathcal{L}_{\eta^*}|\neq 1$,  the adversary have multiple options to choose as the equilibrium, with different noise distributions. Still in all of those options, the utility of the adversary is the same. Thus adversary will choose one of those options arbitrarily. However, the data collector would not be able to guess the adversary's choice, while the DC requires to know the adversary's noise distribution to derive the optimum MMSE estimator accordingly. 
 Surprisingly, the estimator   $\frac{y_1+y_2}{2}$ is universally optimum, within a vanishing gap, for all of the choices of the adversary's noise distributions that achieve an equilibrium, as stated in the following lemma.
 
\begin{lemma}\label{lemma:more_on_remark:mean_is_good_for_all_best_response}
    For any $\eta \in \Lambda_{\DC}$ and $g^*(.) \in \mathcal{B}^{\eta}_{\mathsf{AD}}$ we have 
\begin{align*}
     \frac{-(\eta^2+4)(\eta+2)\Delta^3}{M}  &\leq \\
     \mathsf{MMSE}\big( g^*(.) , \pare\big) &- \mathsf{MSE}\big(\mathsf{mean}(.,.), g^*(.) , \pare\big)\leq 0.
    \end{align*}

\end{lemma}
The  detailed proof is in Appendix \ref{appendix:lemma:more_on_remark:mean_is_good_for_all_best_response}.

\begin{algorithm}[t]
\caption{Characterizing the Optimal Distribution for Adversary}
\label{Alg:finding_noise}
\begin{algorithmic}[1]
\State \textbf{Input:} The function $Q_{\mathsf{AD}}(., .)$, $f_{\bn_h}(n_h)$, $\eta^*$, and $c_{\eta^*}(.)$
\State \textbf{Output:} $g^*(z)$

\State Let  $k_{\eta^*}(z) \triangleq \int_{z-\eta^*\Delta}^{\Delta} f_{\bn_h}(x)  \,dx$ and $\nu_{\eta^*}(z) \triangleq \int_{z-\eta^*\Delta}^{\Delta} (x+z)^2f_{\bn_h}(x)\,dx$, for $z \in [(\eta^*-1)\Delta, (\eta^*+1)\Delta]$.

\State Let $h_{\eta^*}(q) \triangleq \nu_{\eta^*}(k_{\eta}^{-1} (q))$ and $h^*_{\eta^*}(q)$ be the concave envelop of $h_{\eta^*}(q)$, for  $q \in [0,1]$.

\vspace{1em} 
\State \textbf{Step 1:} 
\State Calculate $\mathcal{L}_{\eta^*} = \underset{0 < \alpha \leq 1 }{\arg\max} ~Q_{\mathsf{AD}}(c_{\eta^*} (\alpha), \alpha)$, and choose $\alpha$ as an arbitrary element of  $ \mathcal{L}_{\eta^*}$.

\vspace{1em} 
\State \textbf{Step 2:} 
\If {$h^*_{\eta^*}(\alpha) = h_{\eta^*}(\alpha)$}
    \State Let $z_1 \triangleq k^{-1}_{\eta^*}(\alpha)$
    \State Output $g^*(z) = \frac{1}{2}\delta(z+z_1) + \frac{1}{2}\delta(z-z_1)$
\Else
    \State Find $q_1 < \alpha < q_2$, such that $h^*_{\eta^*}(q_1) = h_{\eta^*}(q_1)$ and $h^*_{\eta^*}(q_2) = h_{\eta^*}(q_2)$,
    and for all $q_1 \leq q \leq q_2$, we have
    \begin{align*}
        h^*_{\eta^*}(q) = \frac{h_{\eta^*}(q_2) - h_{\eta^*}(q_1)}{q_2 - q_1} (q - q_1) + h_{\eta^*}(q_1).
    \end{align*}
    \State Let $z_1 \triangleq k^{-1}_{\eta^*}(q_1)$, $z_2 \triangleq  k^{-1}_{\eta^*}(q_2)$, $\beta_1 \triangleq  \frac{q_2 -\alpha }{2(q_2 - q_1)}$, and $\beta_2 \triangleq  \frac{\alpha - q_1}{2(q_2 - q_1)}$.
    \State Output $g^*(z) = \beta_1 \delta(z+z_1) +\beta_2 \delta(z+z_2) +\beta_1 \delta(z-z_1) +\beta_2 \delta(z-z_2)$
\EndIf
\end{algorithmic}
\end{algorithm}

To appreciate the applicability of our main results in a concrete scenario, consider the following examples. In the following two examples, we assume $\mathbf{n}_h \sim \text{unif}[-\Delta, \Delta]$.  

\begin{example}\label{first_example
_equilibrium}
Consider the problem setup with $\Delta = 1$, and  the utility functions
\begin{align}\label{version_1_utilities_example}
    \mathsf{U}_{\mathsf{AD}}\left( \gdot, \pare \right) &= \log \left( \mathsf{MMSE}\left(\gdot, \pare \right) \right) + \frac{3}{4} \log \left( \mathsf{PA} \left( \gdot, \pare \right) \right), \nonumber \\
    \mathsf{U}_{\mathsf{DC}}\left( \gdot, \eta \right) &= - \mathsf{MMSE}\left(\gdot, \pare \right)  + 25 \mathsf{PA} \left( \gdot, \pare \right).
\end{align}
Recall that the objectives are to find (i) the Stackelberg equilibrium,  (ii) $\eta^*,$ the optimum strategy for the DC, and (iii) $g^*(.)$, the optimum strategy for the adversary.
We characterize the function $c_{\eta} (.)$ using Theorem \ref{theorem:Main_Bound_For_Max_Problem}, as depicted in Fig.  \ref{fig:Finding_equilibrium}, for  $\eta \in \{2, 2.25, 2.5, \ldots, 8\}$. Note that this characterization does not depend on the utility functions $\mathsf{U}_{\mathsf{AD}}$ and $\mathsf{U}_{\mathsf{DC}}$. 
Then, we find the best response of the adversary using Algorithm \ref{Alg:finding_eta}. In particular,
for each $\eta$, we find the set $\mathcal{L}_{\eta} = \underset{0 < \alpha \leq 1 }{\arg\max} ~Q_{\mathsf{AD}}(c_{\eta} (\alpha), \alpha)$. 
Here for each $\eta$,  the set $\mathcal{L}_{\eta}$ has only one member, for which the corresponding point $(c_{\eta}(\alpha), \alpha)$ is shown with a green circle on the curve for $c_{\eta}(.)$.  
\begin{figure}
  \centering
\includegraphics[width=0.85\linewidth]{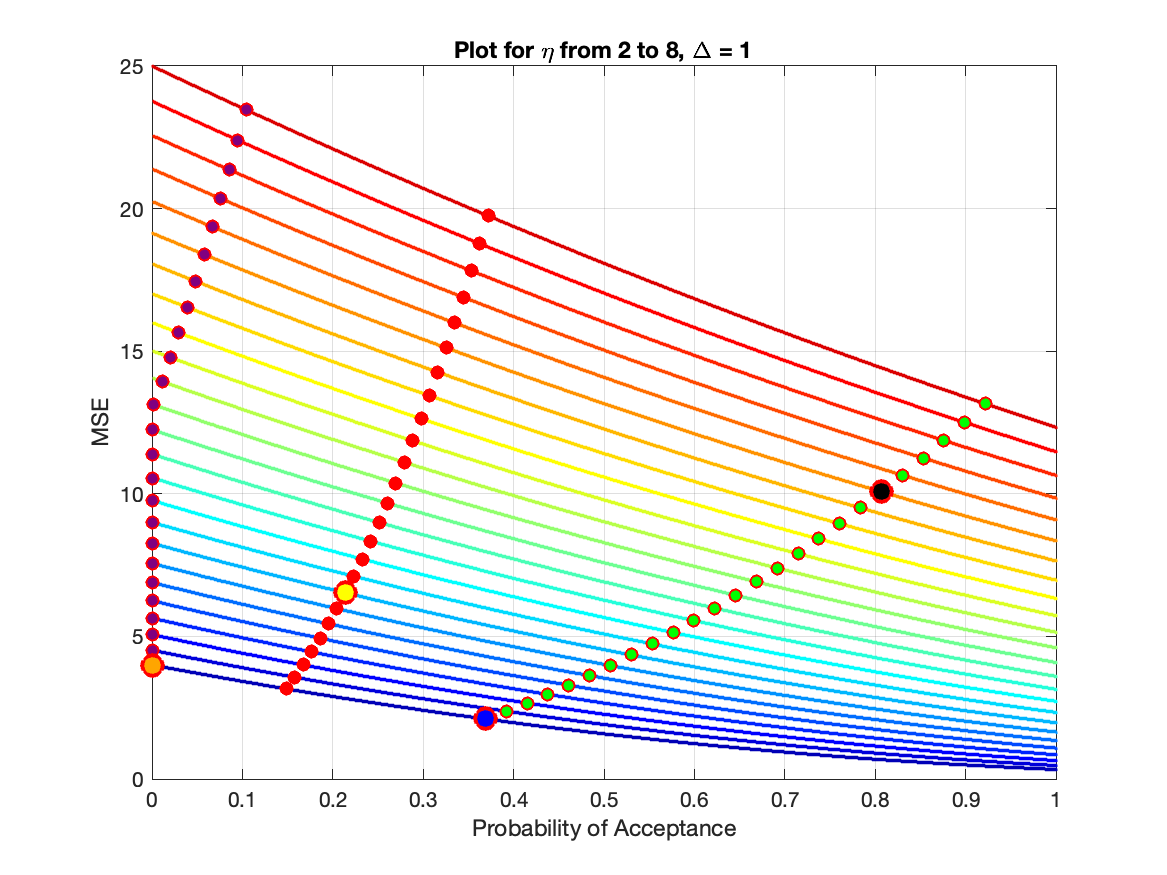}
  \caption{The curves of $c_{\eta}(.)$ for $\eta \in \{2, 2.25, 2.5, \ldots, 8\}$ ($c_{2}(.)$ is the lower-blue curve and $c_{8}(.)$ is the upper-red curve). For each of these $\eta$s,  and the utility functions of Example \ref{first_example
_equilibrium}, the green circles represent $\mathcal{L}_{\eta}$.  Here $|\mathcal{L}_{\eta}|=1$ for each $\eta$. The black circle represents the Stackelberg equilibrium for Example \ref{first_example
_equilibrium}. 
  The red circles represent $\mathcal{L}_{\eta}$, for the utility functions of Example \ref{second_example
_equilibrium}.  Here again $|\mathcal{L}_{\eta}|=1$ for each $\eta$. The yellow circle represents the  Stackelberg equilibrium.
The purple  circles represent $\mathcal{L}_{\eta}$, for the utility functions of Example \ref{example:non_cooperation}.  Here again $|\mathcal{L}_{\eta}|=1$ for each $\eta$. The orange circle represents the  Stackelberg equilibrium.}
  \label{fig:Finding_equilibrium}
\end{figure}
 Finally, having the set $\mathcal{L}_{\eta}$, based on Theorem \ref{theorem: equivalence_two_problem}, we  determine $\eta^*$. For this case, we have
 \begin{align}
     \eta^* = \underset{\pare \in \Lambda_{\mathsf{DC}}}{\arg\max} ~ \underset{\alpha \in \mathcal{L}_{\eta}}{\min}  \left ( -c_{\eta}(\alpha) + 25 \alpha \right).
 \end{align}
 
 The equilibrium is shown by the black circle in Fig. \ref{fig:Finding_equilibrium}. Therefore, for this case, we conclude that $\eta^* = 6.75$, with the probability of acceptance of $ 0.807$ and $\mathsf{MMSE}$ of  $10.07$. The optimum choice for the adversary's noise $g^*(.)$ is determined using Algorithm \ref{Alg:finding_noise}.
We note that in the absence of the framework for the game of coding, the DC commits to $\eta=2$. It is because if both nodes were honest, $|\by_1-\by_2|\leq 2 \Delta$. In response,  the rational adversary chooses its noise according to its utility function. The resulting $(\PA, \MMSE)=(0.37, 2.10)$ has been shown with the blue circle in Fig.     \ref{fig:Finding_equilibrium}. 
At this blue circle,  $\MMSE=2.10$, which is better than what we achieved at the black circle, i.e., $10.97$. However, here the probability of acceptance is just $0.37$, which is small compared to $0.81$. Thus at the black circle, compared to the blue one,  the DC trades accuracy with liveness. 
\end{example}

\begin{example}\label{second_example
_equilibrium}
    For another example, assume that the utility functions of the adversary and the DC are
$\mathsf{U}_{\mathsf{AD}}\left( \gdot, \pare \right) = \log  \mathsf{MMSE}\left(\gdot, \pare \right) + \frac{1}{4} \log  \mathsf{PA}\left( \gdot, \pare \right)$ and 
$\mathsf{U}_{\mathsf{DC}}\left( \gdot, \eta \right) = \frac{\mathsf{PA}\left( \gdot, \pare \right)}{\sqrt{\mathsf{MMSE}\left( \gdot, \pare \right)}}$.
The procedure, similar to  Example \ref{first_example
_equilibrium}, yields $\eta^* =3.75$, along with the Stackelberg equilibrium $(\PA, \MMSE)=(0.214, 6.52)$, as depicted with a yellow circle in Fig.~\ref{fig:Finding_equilibrium}.
For each $\eta$, the single element of $\mathcal{L}_{\eta}$ is shown with a red circle on the corresponding curve of $c_{\eta}(.)$.  
\end{example}

\begin{example}\label{example:non_cooperation}
For another example, assume that the utility functions of the adversary and the DC are
$\mathsf{U}_{\mathsf{AD}}\left( \gdot, \pare \right) = \log  \mathsf{MMSE}\left(\gdot, \pare \right) +  \frac{1}{4}\log (\mathsf{PA}\left( \gdot, \pare \right) + 0.3)$ and 
$\mathsf{U}_{\mathsf{DC}}\left( \gdot, \eta \right) = - \mathsf{MMSE}\left( \gdot, \pare \right) + \mathsf{PA}\left( \gdot, \pare \right)$. The procedure, similar to  Example \ref{first_example
_equilibrium}, yields $\eta^* =2$, along with the Stackelberg equilibrium $(\PA, \MMSE)=(0, 4)$, as depicted with an orange circle in Fig.~\ref{fig:Finding_equilibrium}. This example illustrates a scenario where the utility functions of the DC and adversary are so conflicting that they  effectively preclude any cooperation between the two parties.
As a result, the DC chooses a narrow region with $\eta = 2$, and in response, the adversary selects a noise distribution that leads to a zero chance of liveness at the equilibrium.
\end{example}

\section{Proof of Theorem \ref{theorem: equivalence_two_problem}}\label{proof:theorem: equivalence_two_problem}
Note that (\ref{eq:etastar}) can be reformulated as 
\begin{align}\label{seprating_utilities}
    \eta^* 
        =\underset{\pare \in \Lambda_{\mathsf{DC}}}{\arg\max} ~ \underset{(\beta, \alpha) \in \mathcal{J}_{\eta}}{\min} ~ Q_{\mathsf{DC}} \left(\beta, \alpha\right),
\end{align}
where
\begin{align}\label{J_definition}
    \mathcal{J}_{\eta} \overset{\Delta}{=} \bigl\{\bigl( 
    \mathsf{MMSE}\big(g^*(.),  \eta \big), 
    \mathsf{PA} \left( g^*(.), \pare \right)\big)
    \big) \bigl| ~g^*(.) \in \mathcal{B}^{\pare}_{\mathsf{AD}}\bigr\}.
\end{align} 
To prove Theorem \ref{theorem: equivalence_two_problem}, we  we first show the following lemma.
\begin{lemma}
\label{lemmaJSC}
Let 
\begin{align}\label{defining_the_Set_C_eta}
    \mathcal{C}_{\pare} \triangleq \{ \left(c_{\eta}(\alpha), \alpha\right) ~|~ 0 < \alpha \leq 1\},
\end{align}
where $c_{\eta}(\alpha)$ is defined in \eqref{C_definition}.
For any $\pare \in \Lambda_{\mathsf{AD}}$, $\mathcal{J}_{\pare} \subseteq \mathcal{C}_{\pare}$. 
\end{lemma}
\begin{proof}
    Consider an arbitrary element $(\beta, \alpha) \in \mathcal{J}_{\eta}$. We aim to show $(\beta, \alpha) \in \mathcal{C}_{\eta}$. Based on the definition of $\mathcal{C}_{\eta}$ in \eqref{defining_the_Set_C_eta},  we show that for any $g^*(.) \in \mathcal{B}^{\pare}_{\mathsf{AD}}$, and $\alpha \triangleq \mathsf{PA} \left( g^*(.), \pare \right)$, the distribution $g^*(.)$ is one of the optimal solutions of the following maximization problem
\begin{align}\label{maximization_problem_for_C}
    \underset{\gdot \in \Lambda_{\mathsf{AD}}}{\max} ~ \underset{\mathsf{PA} \left( \gdot, \pare \right) \geq \alpha}{\mathsf{MMSE}\left(\gdot, \pare \right)}.
\end{align}
We prove this claim by contradiction. Assume that this is not the case. Let\footnote{In the case that the optimization problem \eqref{definition_of_g_1} has more than one solution, just pick one of them randomly as the noise function $g_1(.)$.}
    \begin{align}\label{definition_of_g_1}
        g_1(.) = \underset{\gdot \in \Lambda_{\mathsf{AD}}}{\arg\max} ~ \underset{\mathsf{PA} \left( \gdot, \pare \right) \geq \alpha}{\mathsf{MMSE}\left(\gdot, \pare \right)}.
    \end{align}
    Based on our contradictory hypothesis that $g^*(.)$ is not an optimal solution of \eqref{maximization_problem_for_C}, one can verify that 
    \begin{align}\label{comparing_g_1_and_star}
        \mathsf{MMSE}\big(g_1(.), \pare \big) > \mathsf{MMSE}\big( g^*(.),  \pare \big)
    \end{align}
    Therefore, we have
    \begin{align}\label{JSC_contradiction}
        \mathsf{U}_{\mathsf{AD}}\left( g_1(.), \pare \right) &= Q_{\mathsf{AD}} \big(\mathsf{MMSE}\big( g_1(.), \pare\big), \PA\big(g_1(.), \pare\big)\big)\nonumber \\
        &\overset{(a)}{\geq}
        Q_{\mathsf{AD}} \big(\mathsf{MMSE}\big( g_1(.), \pare\big), \alpha \big)\nonumber \\
        &\overset{(b)}{>} 
        Q_{\mathsf{AD}} \big(\mathsf{MMSE}\big( g^*(.),  \pare \big), \alpha \big) \nonumber \\
        &\overset{(c)}{=} Q_{\mathsf{AD}} \big(\mathsf{MMSE}\big( g^*(.),  \pare \big), \mathsf{PA} \left( g^*(.), \pare \right)\big)\nonumber \\
        &=\mathsf{U}_{\mathsf{AD}}\left( g^*(.),  \pare \right),
    \end{align}
    where (a) follows from the facts that based on \eqref{definition_of_g_1}, $\PA\big(g_1(.), \pare\big) \geq \alpha$ and also that $Q_{\mathsf{AD}}(., .)$ is a strictly increasing function with respect to its second argument, (b) follows from \eqref{comparing_g_1_and_star} and also that $Q_{\mathsf{AD}}(., .)$ is a strictly increasing function with respect to its first argument, and (c) follows from the definition of $\alpha$.

    However, \eqref{JSC_contradiction} implies that 
    \begin{align*}
        \mathsf{U}_{\mathsf{AD}}\left( g_1(.), \pare \right) > \mathsf{U}_{\mathsf{AD}}\left( g^*(.),  \pare \right),
    \end{align*}
     which is a contradiction, since $g^*(.) \in \mathcal{B}^{\pare}_{\mathsf{AD}}$. Thus, our first contradictory hypothesis is wrong, and the statement of this lemma is valid.
\end{proof}
Now we prove Theorem \ref{theorem: equivalence_two_problem}.
    Recall that 
\begin{align}
    \mathcal{L}_{\eta} =\underset{0 < \alpha \leq 1 }{\arg\max} ~Q_{\mathsf{AD}}(c_{\eta} (\alpha), \alpha).
\end{align}
Let 
\begin{align}\label{L_definition}
    \mathcal{K}_{\eta} = \{ (c_{\eta}(\alpha), \alpha)~|~ \alpha \in \mathcal{L}_{\eta} \}
\end{align}
Based in \eqref{seprating_utilities}, to complete the proof we show 
$\mathcal{J}_{\eta} = \mathcal{K}_{\eta}$.
 we first show that $\mathcal{J}_{\eta} \subseteq \mathcal{K}_{\eta}$, and then prove $\mathcal{K}_{\eta} \subseteq \mathcal{J}_{\eta}$.

\subsection{\texorpdfstring{Proof of $\mathcal{J}_{\eta} \subseteq \mathcal{K}_{\eta}$}{x}} Consider $(\beta, \alpha) \in \mathcal{J}_{\eta}$.  We claim that $(\beta, \alpha) \in \mathcal{K}_{\eta}$, and show this by contradiction. Assume, as a contradictory hypothesis, that it is not the case. According to Lemma \ref{lemmaJSC} $(\beta, \alpha) \in \mathcal{C}_{\eta}$. Since we assumed that $(\beta, \alpha) \notin \mathcal{K}_{\eta}$, based on the definition in \eqref{L_definition},  there exists $(b,a) \in \mathcal{K}_{\eta}$, such that we have \begin{align}\label{JnsL_contradiction_assuption}
        Q_{\mathsf{AD}}(b,a) > Q_{\mathsf{AD}}(\beta, \alpha).
    \end{align}

    Since $(\beta, \alpha) \in \mathcal{J}_{\eta}$, based on the definition in \eqref{J_definition} there exists a noise distribution $g_{\alpha}(.) \in \mathcal{B}^{\pare}_{\mathsf{AD}}$, where 
    \begin{align}\label{definition of beta alpha}
    \left(\beta, \alpha\right) &= \bigl( 
    \mathsf{MMSE}\big(g_{\alpha}(.), \pare \big), \PA \big(g_{\alpha}(.), \pare
    \big)
    \big).
    \end{align}
    
    Similarly for $(b,a) \in \mathcal{K}_{\eta}$,
    based on the definition in \eqref{L_definition}, 
    we have $(b,a) \in \mathcal{C}_{\eta}$. Based on the definition in \eqref{defining_the_Set_C_eta},
    there exists a noise distribution $g_b(.) \in \Lambda_{\mathsf{AD}}$, where 
    \begin{align}\label{a_to_PA_relationship}
        \mathsf{PA} \left( g_b(.), \pare \right) \geq a,
    \end{align}
    and $g_b(.)$ is one of the optimal solutions of the following
    maximization problem
    \begin{align}\label{gb_to_opt}
     \underset{\gdot \in \Lambda_{\mathsf{AD}}}{\max} ~ \underset{\mathsf{PA} \left( \gdot, \pare \right) \geq a}{\mathsf{MMSE}\left(\gdot, \pare \right)},
    \end{align}
    and we have
    \begin{align}\label{b_to_MMSE_relation}
        b = \mathsf{MMSE}\left(g_{b}(.), \pare \right).
    \end{align}
    One can verify that we have
    \begin{align}\label{JnsL_contradiction_point}
        \mathsf{U}_{\mathsf{AD}}\left(g_b(.), \pare \right) &= Q_{\mathsf{AD}}\big( 
    \mathsf{MMSE}\left(g_{b}(.), \pare \right), \mathsf{PA} \left( g_b(.), \pare \right)
    \big) \nonumber \\
    &\overset{(a)}{=} Q_{\mathsf{AD}}\big( b, \mathsf{PA} \left( g_b(.), \pare \right)
    \big) \nonumber \\
    &\overset{(b)}{\geq} Q_{\mathsf{AD}}\left( b, a \right) \nonumber \\
    &\overset{(c)}{>} Q_{\mathsf{AD}}\left( \beta, \alpha \right) \nonumber \\
    &\overset{(d)}{=} Q_{\mathsf{AD}}\big( 
    \mathsf{MMSE}\left(g_{\alpha}(.), \pare \right), \mathsf{PA} \left( g_{\alpha}(.), \pare \right)
    \big) \nonumber \\
    &=\mathsf{U}_{\mathsf{AD}}\left(g_{\alpha}(.), \pare \right)
    \end{align}
    where (a) follows from \eqref{b_to_MMSE_relation}, (b) follows from \eqref{a_to_PA_relationship} and  $Q_{\mathsf{AD}}(.,.)$ is a strictly increasing function with respect to its second argument, (c) follows from \eqref{JnsL_contradiction_assuption}, (d) follows from \eqref{definition of beta alpha}. However, \eqref{JnsL_contradiction_point} implies that $\mathsf{U}_{\mathsf{AD}}\left(g_b(.), \pare \right) > \mathsf{U}_{\mathsf{AD}}\left(g_{\alpha}(.), \pare \right)$, which is a contradiction, since $g_{\alpha}(.) \in \mathcal{B}^{\pare}_{\mathsf{AD}}$. Therefore, our first assumption was wrong, and for each $(\beta, \alpha) \in \mathcal{J}_{\eta}$ we have $(\beta, \alpha) \in \mathcal{K}_{\eta}$.
    
    \subsection{\texorpdfstring{Proof of $\mathcal{K}_{\eta} \subseteq \mathcal{J}_{\eta}$}{x}}
    We prove this by contradiction.  Assume that this is not the case. Consider $(b,a) \in \mathcal{K}_{\eta}$, where $(b,a) \notin \mathcal{J}_{\eta}$, and the noise distribution $g_b(.) \in \Lambda_{\mathsf{AD}}$, where \eqref{a_to_PA_relationship}, \eqref{gb_to_opt}, and \eqref{b_to_MMSE_relation} holds.
    Now consider a $(\beta, \alpha) \in \mathcal{J}_{\eta}$, and the noise distribution $g_{\alpha}(.) \in \mathcal{B}^{\pare}_{\mathsf{AD}}$, where \eqref{definition of beta alpha} holds.
    We note that based on Lemma \ref{lemmaJSC}, $(\beta, \alpha) \in \mathcal{C}_{\eta}$. Since $(b,a) \in \mathcal{K}_{\eta}$, based on the definition \eqref{J_definition} we have
    \begin{align}\label{ba_tobetaalpha_relation}
       Q_{\mathsf{AD}}\left( b,a \right) \geq Q_{\mathsf{AD}}\left( \beta, \alpha \right).
    \end{align}
    Consider the following chain of inequalities: 
    \begin{align}\label{jset_to_lset}
        \mathsf{U}_{\mathsf{AD}}\left(g_b(.), \pare \right) &= Q_{\mathsf{AD}}\bigl( 
    \mathsf{MMSE}\left(g_{b}(.), \pare \right), \PA \left(g_{b}(.), \pare
    \right)
    \big) \nonumber \\
    &\overset{(a)}{=} Q_{\mathsf{AD}}\big( b, \PA \left(g_{b}(.), \pare
    \right)
    \big) \nonumber \\
    &\overset{(b)}{\geq} Q_{\mathsf{AD}}\left( b, a \right) \nonumber \\
    &\overset{(c)}{\geq} Q_{\mathsf{AD}}\left( \beta, \alpha \right) \nonumber \\
    &\overset{(d)}{=} Q_{\mathsf{AD}}\big( 
    \mathsf{MMSE}\left(g_{\alpha}(.), \pare \right), \mathsf{PA} \left( g_{\alpha}(.), \pare \right)
    \big) \nonumber \\
    &=\mathsf{U}_{\mathsf{AD}}\left(g_{\alpha}(.), \pare \right),
    \end{align}
    where (a) follows from \eqref{b_to_MMSE_relation}, (b) follows from \eqref{a_to_PA_relationship} and  $Q_{\mathsf{AD}}(.,.)$ is a strictly increasing function with respect to its second argument, (c) follows from \eqref{ba_tobetaalpha_relation}, (d) follows from \eqref{definition of beta alpha}.

    However, \eqref{jset_to_lset} implies that $\mathsf{U}_{\mathsf{AD}}\left(g_b(.), \pare \right) \geq \mathsf{U}_{\mathsf{AD}}\left(g_{\alpha}(.), \pare \right)$. On the other hand, since  $g_{\alpha}(.) \in \mathcal{B}^{\pare}_{\mathsf{AD}}$, we have $\mathsf{U}_{\mathsf{AD}}\left(g_{\alpha}(.), \pare \right) \geq \mathsf{U}_{\mathsf{AD}}\left(g_b(.), \pare \right)$. Thus $\mathsf{U}_{\mathsf{AD}}\left(g_b(.), \pare \right) = \mathsf{U}_{\mathsf{AD}}\left(g_{\alpha}(.), \pare \right)$.
    Consequently, based on \eqref{jset_to_lset}, we have $Q_{\mathsf{AD}}\left( b, \PA \left(g_{b}(.), \pare
    \right)
    \right) = Q_{\mathsf{AD}}\left( b, a \right)$. Since $Q_{\mathsf{AD}}(.,.)$ is a strictly increasing function with respect to its second argument, it implies that $\PA \left(g_{b}(.), \pare \right) = a$. Thus, we have
    \begin{align}
        (b,a) = \bigl( 
    \mathsf{MMSE}\left(g_{b}(.), \pare \right), \PA \left(g_{b}(.), \pare
    \right)
    \big),
    \end{align}
    On the other hand, $g_{\alpha}(.) \in \mathcal{B}^{\pare}_{\mathsf{AD}}$ and $\mathsf{U}_{\mathsf{AD}}\left(g_b(.), \pare \right) = \mathsf{U}_{\mathsf{AD}}\left(g_{\alpha}(.), \pare \right)$, which implies that $g_b(.) \in \mathcal{B}^{\pare}_{\mathsf{AD}}$. Therefore, by definition it implies that $(b,a) \in \mathcal{J}_{\eta}$. This is against our first assumption that $(b,a) \notin \mathcal{J}_{\eta}$, which is a contradiction. Therefore, we have $\mathcal{K}_{\eta} \subseteq \mathcal{J}_{\eta}$. This completes the proof of Theorem \ref{theorem: equivalence_two_problem}.

\section{Proof of Theorem \ref{theorem:Main_Bound_For_Max_Problem}}\label{proof:theorem:Main_Bound_For_Max_Problem}    
As mentioned before, Theorem \ref{theorem:Main_Bound_For_Max_Problem} is a direct consequence of Lemma \ref{lemma:mean_is_near_optimal} and \ref{lemma:best_noise}, which are proved in the following sections.

\subsection{Proof of Lemma \ref{lemma:mean_is_near_optimal}}\label{proof:lemma:mean_is_near_optimal}
To prove Lemma \ref{lemma:mean_is_near_optimal}, we first show the following lemma which establishes that in Opt. 2, we can limit the search space for $\gdot$ to the set of symmetric probability density functions. 
\begin{lemma}\label{lemma:making_symmetry}
    For any $\eta \in \Lambda_{\mathsf{DC}}$ and $\gdot \in \Lambda_{\mathsf{AD}}$, we have $\mathsf{MSE}\left(\mathsf{mean}(.,.), g_{\textrm{sym}}(.), \pare\right) = \mathsf{MSE}\left(\mathsf{mean}(.,.), \gdot, \pare\right)$ and  $\PA (g_{\textrm{sym}}(.), \pare) = \PA (g(.), \pare)$,  where
      $g_{\textrm{sym}}(z) \triangleq \frac{g(z) + g(-z)}{2}$. 
\end{lemma}
\begin{proof}
    The formal proof is in Appendix \ref{proof:lemma:making_symmetry}. Informally the proof is as follows. Let $g_{\textrm{ref}}(z) \triangleq g(-z)$. Note that the acceptance rule is of the form $|\by_1-\by_2| \leq \eta\Delta$, or equivalently $|\bn_h-\bn_a| \leq \eta\Delta$. The statistics of $|\bn_h-\bn_a|$ are the same for both $\gdot, g_{\textrm{ref}}(.)$ and $g_{\textrm{sym}}(.)$. Thus, the probability of acceptance is the same for all these noise distributions. In addition, the MSE for the $\mathsf{mean}$, given acceptance, is $|\bu - \frac{\by_1+\by_2}{2}|^2 = \frac{|\bn_h+\bn_a|^2}{4}$, whose statistics also remain the same for all these noise distributions.
\end{proof}
\begin{corollary}
    Let  $\Bar{\Lambda}_{\mathsf{AD}} = \big\{ g(z) ~|  ~g(z) \in \Lambda_{\mathsf{AD}} ~\& ~g(z) = g(-z) \big\}$, which is simply the set of all symmetric noise distributions. A direct conclusion of Lemma \ref{lemma:making_symmetry} and this definition is that, for any $0<\alpha \leq 1$ and $\eta \in \Lambda_{\DC}$
\begin{align}\label{relation_between_best_seymmetric_best_general}
    &\underset{\gdot \in \Bar{\Lambda}_{\mathsf{AD}}}{\max} ~ \underset{\mathsf{PA} \left( \gdot, \pare \right) \geq \alpha}{\mathsf{MSE}\big(\mathsf{mean}(.,.), \gdot, \pare \big)} \nonumber \\
    &= \underset{\gdot \in \Lambda_{\mathsf{AD}}}{\max} ~ \underset{\mathsf{PA} \left( \gdot, \pare \right) \geq \alpha}{\mathsf{MSE}\big(\mathsf{mean}(.,.), \gdot, \pare \big)}.
\end{align}
\end{corollary}

Recall that based on \eqref{definition_best_estimator}  we have
\begin{align}
 \est^*_{\pare,g} = \underset{\est: \mathbb{R}^2 \to \mathbb{R}}{\arg\min} ~\mathsf{MSE}\big(\est(.,.), \gdot, \pare\big),
 \end{align}
 and
 \begin{align}
     \mathsf{MMSE}\big( \gdot, \pare\big) = \mathsf{MSE}\big(\est^*_{\pare,g}(.,.), \gdot, \pare\big).
 \end{align}
 
To complete the proof of Lemma \ref{lemma:mean_is_near_optimal}, we use the following lemmas.
 
\begin{lemma}\label{MMSE_for_Internal_Part}
    For any $\pare \in \Lambda_{\mathsf{DC}}$,  $\gdot \in \Bar{\Lambda}_{\mathsf{AD}}$, and $\underline{y} = (y_1,y_2)$, where $|y_1+y_2| \leq 2M - (\eta+2)\Delta$, we have
\begin{align}
    \est^*_{\pare,g}\big( \underline{y} \big) = \mathbb{E}\big[\mathbf{u} | \mathcal{A}_{\pare}, \underline{y}\big] = \frac{y_1+y_2}{2}.
\end{align}
\end{lemma}
\begin{proof}
    The proof is in Appendix \ref{proof:MMSE_for_Internal_Part}
\end{proof}
\begin{lemma}\label{bound_mse_symmetric_mean}
    For any $\pare \in \Lambda_{\mathsf{DC}}$, and  $\gdot \in \Bar{\Lambda}_{\mathsf{AD}}$, we have
    \begin{align}\label{statement_of_lemma_for_mmse_and_mean}
      &\mathsf{MSE}\big(\mathsf{mean}(.,.), \gdot, \pare\big) + \frac{-(\eta^2+4)(\eta+2)\Delta^3}{M}\nonumber\\  &\leq \mathsf{MMSE}\big( \gdot, \pare\big)  
      \leq \mathsf{MSE}\big(\mathsf{mean}(.,.), \gdot, \pare\big).
    \end{align}
\end{lemma}
\begin{proof}
    The proof is in Appendix \ref{proof:bound_mse_symmetric_mean}.
\end{proof}

\begin{corollary}
    A direct consequence of Lemma \ref{bound_mse_symmetric_mean} is that for any $\pare \in \Lambda_{\mathsf{DC}}$, and $0 < \alpha \leq 1$, we have

    \begin{align}\label{equation:bound_mse_symmetric_mean}
        &\underset{\gdot \in \Bar{\Lambda}_{\mathsf{AD}}}{\max} ~ \underset{\mathsf{PA} \left( \gdot, \pare \right) \geq \alpha}{\mathsf{MSE}\big(\mathsf{mean}(.,.), \gdot, \pare \big)} + \frac{-(\eta^2+4)(\eta+2)\Delta^3}{M} \nonumber \\
        &\quad \leq \underset{\gdot \in \Bar{\Lambda}_{\mathsf{AD}}}{\max} ~ \underset{\mathsf{PA} \left( \gdot, \pare \right) \geq \alpha}{\mathsf{MMSE}\big( \gdot, \pare \big)} \nonumber \\
        &\quad\leq \underset{\gdot \in \Bar{\Lambda}_{\mathsf{AD}}}{\max} ~ \underset{\mathsf{PA} \left( \gdot, \pare \right) \geq \alpha}{\mathsf{MSE}\big(\mathsf{mean}(.,.), \gdot, \pare \big)}.
    \end{align}
\end{corollary}

Now we prove Lemma \ref{lemma:mean_is_near_optimal}. 
    For any $0 < \alpha \leq 1$, let\footnote{For the case that the optimization problem in \eqref{best_general_noise_definition} and \eqref{best_symmetric_noise_definition} has more than one solution, we just pick one them randomly.}
\begin{align}
    g^*(.) &= \underset{\gdot \in \Lambda_{\mathsf{AD}}}{\arg \max} ~ \underset{\mathsf{PA} \left( \gdot, \pare \right) \geq \alpha}{\mathsf{MMSE}\left(\gdot, \pare \right)}, \label{best_general_noise_definition}\\
    g^*_s(.) &= \underset{\gdot \in \Bar{\Lambda}_{\mathsf{AD}}}{\arg \max} ~ \underset{\mathsf{PA} \left( \gdot, \pare \right) \geq \alpha}{\mathsf{MMSE}\big( \gdot, \pare \big)}. \label{best_symmetric_noise_definition}
\end{align}

One can verify that 
\begin{align}
   c_{\pare}(\alpha) &\overset{(a)}{=}\mathsf{MMSE}\big(g^*(.),  \pare \big) \nonumber \\
   &\geq \mathsf{MMSE}\big(g^*_s(.), \pare \big) \nonumber \\
   &\overset{(b)}{\geq}  \bigg(\underset{\gdot \in \Bar{\Lambda}_{\mathsf{AD}}}{\max} ~ \underset{\mathsf{PA} \left( \gdot, \pare \right) \geq \alpha}{\mathsf{MSE}\big(\mathsf{mean}(.,.), \gdot, \pare \big)} \bigg) \nonumber \\
   &\quad - \frac{(\eta^2+4)(\eta+2)\Delta^3}{M}\nonumber \\
   &\overset{(c)}{=} \bigg(\underset{\gdot \in \Lambda_{\mathsf{AD}}}{\max} ~ \underset{\mathsf{PA} \left( \gdot, \pare \right) \geq \alpha}{\mathsf{MSE}\big(\mathsf{mean}(.,.), \gdot, \pare \big)} \bigg)\nonumber \\
   &\quad - \frac{(\eta^2+4)(\eta+2)\Delta^3}{M} \nonumber \\
   & \overset{(d)}{=} \beta_{\eta}(\alpha) - \frac{(\eta^2+4)(\eta+2)\Delta^3}{M},
\end{align}
where (a) follows from \eqref{C_definition}, (b) is based on \ref{equation:bound_mse_symmetric_mean}, (c) is a consequence of \eqref{relation_between_best_seymmetric_best_general}, and (d) follows from \eqref{definition_beta_eta}. Let 
\begin{align}
    g_{\textrm{sym}}(z) \triangleq \frac{g(z) + g(-z)}{2}.
\end{align}
To complete the proof of Lemma \ref{lemma:mean_is_near_optimal}, note that we have
\begin{align}
    c_{\pare}(\alpha) & = \mathsf{MMSE}\big(g^*(.),  \pare \big) \nonumber \\
    &\leq \mathsf{MSE}\big(\mathsf{mean}(.,.), g^*(.),  \pare \big) \nonumber \\
    &\overset{(a)}{=}\mathsf{MSE}\big(\mathsf{mean}(.,.), g^*_{\textrm{sym}}(.), \pare \big) \nonumber \\
    &\overset{(b)}{\leq}\underset{\gdot \in \Bar{\Lambda}_{\mathsf{AD}}}{\max} ~ \underset{\mathsf{PA} \left( \gdot, \pare \right) \geq \alpha}{\mathsf{MSE}\big(\mathsf{mean}(.,.), \gdot, \pare \big)} \nonumber \\
    &\overset{(c)}{=} \underset{\gdot \in \Lambda_{\mathsf{AD}}}{\max} ~ \underset{\mathsf{PA} \left( \gdot, \pare \right) \geq \alpha}{\mathsf{MSE}\big(\mathsf{mean}(.,.), \gdot, \pare \big)} \nonumber \\
    & = \beta_{\pare}(\alpha),
\end{align}
where (a) follows from Lemma \ref{lemma:making_symmetry}, and (b) follows from the fact that $g^*_{\textrm{sym}}(.) \in \Bar{\Lambda}_{\mathsf{AD}}$, (c) follows from \eqref{relation_between_best_seymmetric_best_general}.

\begin{remark}
    Lemma \ref{lemma:mean_is_near_optimal} presents a surprising result. Specifically, in Optimization 1, the estimation function is the minimum mean square estimator, which is dependent on the noise distribution of the adversary, whereas in Optimization 2, the estimation function is fixed. However, Lemma \ref{lemma:mean_is_near_optimal} reveals that the results of these two optimization problems are approximately equal, with a gap that vanishes asymptotically. 
\end{remark}

\subsection{Proof of Lemma \ref{lemma:best_noise}}\label{proof:lemma:best_noise} 
 To prove this lemma, firstly in Claim \ref{lemma:upper_bound_of_mean} we demonstrate that $\beta_{\eta}(\alpha) \leq  \frac{h^*_{\eta}(\alpha)}{4\alpha}$ (converse). To show that, we first show that by limiting the search space of $g(.)$ to symmetric distributions, we do not lose optimality. Then,
we take the following steps: (i) We show that if we limit the domain of the noise distributions to $[-(\eta+1)\Delta, -(\eta-1)\Delta] \cup [(\eta-1)\Delta, (\eta+1)\Delta] $, and assume that the probability of acceptance is exactly equal to $\alpha$, we again don't lose the optimality. (ii)  We prove that for those noise distributions $\mathsf{MMSE}\left(\mathsf{mean}(.,.), \gdot, \pare \right)
        = \frac{1}{2\alpha}\int_{(\eta-1)\Delta}^{(\eta+1)\Delta} \nu_{\eta}(z)g(z) \,dz$, where $\nu_{\eta}(z)$ is defined in Theorem \eqref{theorem:Main_Bound_For_Max_Problem}. (iii) By changing the parameters of the integral and applying Jensen's inequality, we derive the upper bound of $\frac{h^*_{\eta}(\alpha)}{4\alpha}$. 

        Then in Claim \ref{lemma:lower_bound_of_mean} we prove that $\beta_{\eta}(\alpha) \geq  \frac{h^*_{\eta}(\alpha)}{4\alpha}$ (achievability):  We use 
         Algorithm \ref{Alg:finding_noise} 
         to find a noise distribution and show that for that distribution $\mathsf{MSE}(\mathsf{mean}(.,.), \gdot, \pare) = \frac{h^*_{\eta}(\alpha)}{4\alpha}$, and $\PA(\gdot, \eta) = \alpha$.
         
         The details of the proof of Lemma \ref{lemma:best_noise} is as follows. Note that based on Lemma~\ref{lemma:making_symmetry}, we limit the search space of the noise distributions to symmetric ones.  Therefore, hence after we assume that the noise distribution of the adversary is symmetric.

  \begin{claim}\label{lemma:upper_bound_of_mean}
        For any $\pare \in \Lambda_{\mathsf{DC}} $ and $0 < \alpha \leq 1$, we have
    \begin{align}\label{H_eta_is_upperbound}
        \beta_{\eta}(\alpha) \leq \frac{h^*_{\eta}(\alpha)}{4\alpha}.
    \end{align}
    \end{claim}

  To prove Claim \ref{lemma:upper_bound_of_mean}, we prove Lemmas~\ref{general_format_symmetric}, \ref{lemma:bounded_noise_existence}, and \ref{lemma:exact_acc_noise_existence}, stated as follows: 
    \begin{lemma}\label{general_format_symmetric}
    For any symmetric $\gdot \in \Lambda_{\mathsf{AD}}$, we have
    \begin{align}
    \mathsf{PA} \left( \gdot, \pare \right) &= 2 \int_{0}^{(\eta-1)\Delta} g(z) \,dz \nonumber \\
    &\quad + 2\int_{(\eta-1)\Delta}^{(\eta+1)\Delta} 
  \int_{z-\eta\Delta}^{\Delta} f_{\bn_h}(x) g(z) \,dx \,dz,\label{general_symmetric_acc}
 \end{align}
 and
  \begin{align}
    &\mathsf{MSE}\big(\mathsf{mean}(.,.), \gdot, \pare \big)  \nonumber \\
        & =\frac{1}{2\mathsf{PA} \left( \gdot, \pare \right)}\int_{0}^{(\eta-1)\Delta} (z^2 + \sigma^2_{\bn_h})g(z) \,dz \nonumber \\
        &+\frac{1}{2\mathsf{PA} \left( \gdot, \pare \right)}\int_{(\eta-1)\Delta}^{(\eta+1)\Delta} \int_{z-\eta\Delta}^{\Delta} (x+z)^2f_{\bn_h}(x)g(z)\,dx \,dz ,\label{general_symmetric_mse}
\end{align}
where $f_{\bn_h}$ is the PDF of  the noise of the honest node.
    \end{lemma}
    \begin{proof}
        Proof is in Appendix \ref{proof:general_format_symmetric}
    \end{proof}
    

\begin{lemma}\label{lemma:bounded_noise_existence}
    For any symmetric noise distribution $g_1(.)$, where $\PA(g_1(.), \eta) > 0$, there exist a symmetric noise distribution $g_2(.)$, such that for the case of 
     $|z| < (\eta-1)\Delta$,   
    $|z| > (\eta+1)\Delta$, $g_2(.)=0$ and
    \begin{align}
        \PA(g_2(.), \eta) &\geq \PA(g_1(.), \eta) \label{betternoise_acc} \\
        \mathsf{MSE}(\mathsf{mean}(.,.), g_2(.), \eta) &\geq \mathsf{MSE}(\mathsf{mean}(.,.), g_1(.), \eta) \label{betternoise_mse}.
    \end{align}
\end{lemma}
    
\begin{proof}
    Proof is in Appendix \ref{proof:lemma:bounded_noise_existence}
\end{proof}

For any symmetric noise distribution $g(.)$, such that for the case of $|z| < (\eta-1)\Delta$ and $|z| > (\eta-1)\Delta$, we have $g(z) = 0$, we call $g(.)$ a satisfying noise of Lemma \ref{lemma:bounded_noise_existence}.

\begin{lemma}\label{lemma:exact_acc_noise_existence}
    Let $g_1(.)$ be a satisfying noise of Lemma \ref{lemma:bounded_noise_existence}. For each $0< \alpha < \PA(g_1(.), \eta)$, there exists a satisfying noise of Lemma \ref{lemma:bounded_noise_existence} , $g_2(.)$, such that  $\mathsf{PA} \left( g_2(.), \pare \right) = \alpha$, and $\mathsf{MSE}(\mathsf{mean}(.,.), g_2(.), \eta) = \mathsf{MSE}(\mathsf{mean}(.,.), g_1(.), \eta)$.
\end{lemma}

\begin{proof}
    Proof is in Appendix \ref{proof:lemma:exact_acc_noise_existence}.
\end{proof}

    Using the above lemmas, we prove Claim \ref{lemma:upper_bound_of_mean}, as follows. 
    \begin{proof}
       Without loss of generality, henceforth, we assume that all of the adversarial noise meets the condition of Lemmas \ref{lemma:bounded_noise_existence} and \ref{lemma:exact_acc_noise_existence}. More precisely, $\mathsf{PA} \left( g(.), \pare \right) = \alpha$, and also  $g(z) = 0$ for $|z| < (\eta-1)\Delta$ and
    $|z| > (\eta+1)\Delta$. Then from Lemma  \ref{general_format_symmetric}, we have
\begin{align}\label{acc_neww_general}
    \mathsf{PA} \left( \gdot, \pare \right) &= 2\int_{(\eta-1)\Delta}^{(\eta+1)\Delta}  
  \int_{z-\eta\Delta}^{\Delta} f_{\bn_h}(x) g(z)\,dx  \,dz, 
    \end{align}
    and
  \begin{align}
    &\mathsf{MSE}\big(\mathsf{mean}(.,.), \gdot, \pare \big) \nonumber \\
    &=\frac{1}{2\alpha}\int_{(\eta-1)\Delta}^{(\eta+1)\Delta} \int_{z-\eta\Delta}^{\Delta} (x+z)^2f_{\bn_h}(x)g(z)\,dx  \,dz. \label{mse_neww_general}
\end{align}
Let $k_{\eta}(z) \triangleq \int_{z-\eta\Delta}^{\Delta} f_{\bn_h}(x)\,dx$, for $(\eta-1)\Delta \leq z \leq  (\eta+1)\Delta$.  Based on \eqref{acc_neww_general} we have
\begin{align}\label{prob_acc_to_noise_general}
    \mathsf{PA} \left( \gdot, \pare \right) &= 2\int_{(\eta-1)\Delta}^{(\eta+1)\Delta} \left( 
  \int_{z-\eta\Delta}^{\Delta} f_{\bn_h}(x) \,dx\right) g(z) \,dz  \nonumber \\ &=2\int_{(\eta-1)\Delta}^{(\eta+1)\Delta}  k_{\eta}(z) g(z) \,dz.
\end{align}
    Since we assume that $g(z)$ is symmetric, we have
\begin{align}\label{pdf_positive_general}
    \int_{(\eta-1)\Delta}^{(\eta+1)\Delta} g(z) \,dz = \frac{1}{2}.
\end{align}
 By changing the parameters of in the integrals for $q = k_{\eta}(z)$, we can rewrite \eqref{pdf_positive_general} as
\begin{align}\label{rewrite:pdf_positive_general}
\int_{0}^{1} w_{\eta}(q) \,dq = \frac{1}{2},
\end{align}
where in the above equation, $w_{\eta}(q) \triangleq \frac{-g(k_{\eta}^{-1}(q))}{k_{\eta}^{\prime}( k_{\eta}^{-1} (q))}$, for $0 \leq q \leq 1$. Note that since $k_{\eta}(z)$ is a strictly decreasing function, then $w_{\eta}(q)$ is a non-negative function.  Similarly, we can rewrite \eqref{prob_acc_to_noise_general} as
\begin{align}\label{rewrite:prob_acc_to_noise_general}
  \int_{0}^{1}  q w_{\eta}(q) \,dq = \frac{\alpha}{2}
\end{align}
Let $\nu_{\eta}(z) \triangleq \int_{z-\eta\Delta}^{\Delta} (x+z)^2f_{\bn_h}(x)\,dx$. Based on \eqref{mse_neww_general}, we have
\begin{align}\label{mse_general_for_bounded_noise}
    &\mathsf{MSE}\big(\mathsf{mean}(.,.), \gdot, \pare \big) =\nonumber \\
    &=  \frac{1}{2\alpha}\int_{(\eta-1)\Delta}^{(\eta+1)\Delta} \left(\int_{z-\eta\Delta}^{\Delta} (x+z)^2f_{\bn_h}(x)\,dx \right)g(z) \,dz \nonumber \\
     & =  \frac{1}{2\alpha}\int_{(\eta-1)\Delta}^{(\eta+1)\Delta} \nu_{\eta}(z)g(z) \,dz \nonumber \\
     &\overset{(a)}{=} \frac{1}{2\alpha}\int_{0}^{1}  h_{\eta}(q) w_{\eta}(q) \,dq,
\end{align}
where in (a) $q = k_{\eta}(z)$ and  $ h_{\eta}(q) \triangleq \nu_{\eta}(k_{\eta}^{-1} (q))$. Let $h^*_{\eta}(q)$ be the upper concave envelop of $h_{\eta}(q)$, for $0 \leq q \leq 1$. Then
\begin{align}
    \frac{1}{2\alpha}\int_{0}^{1}  h_{\eta}(q) w_{\eta}(q) \,dq &\leq \frac{1}{2\alpha}\int_{0}^{1}  h^*_{\eta}(q) w_{\eta}(q) \,dq \nonumber \\
    & =  \frac{1}{4\alpha}\int_{0}^{1}  h^*_{\eta}(q) 2w_{\eta}(q) \,dq \nonumber \\
    &\overset{(a)}{\leq} \frac{1}{4\alpha} h^*_{\eta} \left( \int_{0}^{1}  2q w_{\eta}(q) \,dq \right) \nonumber \\
    & \overset{(b)}{=} \frac{h^*_{\eta}(\alpha)}{4\alpha},
\end{align}
where (a) follows from  Jensen's inequality and the fact that based on \eqref{rewrite:pdf_positive_general}, $\int_{0}^{1} 2w_{\eta}(q) \,dq = 1$, and (b) follows from \eqref{rewrite:prob_acc_to_noise_general}. This completes the proof of Claim \ref{lemma:upper_bound_of_mean}.

    \end{proof}
    \begin{claim}\label{lemma:lower_bound_of_mean}
        For any $\pare \in \Lambda_{\mathsf{DC}} $ and $0 < \alpha \leq 1$, we have
    \begin{align}
        \beta_{\eta}(\alpha) \geq \frac{h^*_{\eta}(\alpha)}{4\alpha}.
    \end{align}
    \end{claim}
\begin{proof}
To prove this claim, we introduce a noise distribution satisfying 
$\mathsf{PA} \left( \gdot, \pare \right) \geq \alpha$ and  $\mathsf{MSE}\big(\mathsf{mean}(.,.), \gdot, \pare \big) = \frac{h^*_{\eta}(\alpha)}{4\alpha}$. Note that for $(\eta-1)\Delta \leq z \leq  (\eta+1)\Delta$, we have
\begin{align}
    k_{\eta}(z) = \int_{z-\eta\Delta}^{\Delta} f_{\bn_h}(x)\,dx \nonumber
    &=F_{\bn_h}(\Delta) - F_{\bn_h}(z-\eta\Delta) \nonumber \\
    &\overset{(a)}{=}1 - F_{\bn_h}(z-\eta\Delta),
\end{align}
where (a) follows from the fact 
that $\Pr (|\bn_h| > \Delta) = 0$, which implies that $F_{\bn_h}(\Delta) = 1$.
Recall that we assume that $F_{\bn_h}(.)$ is a strictly increasing function in $[-\Delta, \Delta]$, i.e., for all $-\Delta \leq a < b \leq \Delta$, we have 
$F_{\bn_h}(a) < F_{\bn_h}(b)$. This implies that $k_{\eta}(z)$ is an invertible function for $(\eta-1)\Delta \leq z \leq  (\eta+1)\Delta$. Therefore, we have
\begin{align}\label{inverse_of_CDF}
    k^{-1}_{\eta}(q) = \eta \Delta + F^{-1}_{\bn_h}(1-q),
\end{align}
for any $0 \leq q \leq 1$. Consider following cases:
\begin{enumerate}
    \item $h^*_{\eta}(\alpha) = h_{\eta}(\alpha)$: In this case, let $z_1 = k^{-1}_{\eta}(\alpha)$ and $g(z) = \frac{1}{2}\delta(z+z_1) + \frac{1}{2}\delta(z-z_1)$. Note that we can calculate $z_1$ based on \eqref{inverse_of_CDF}.
    Based on \eqref{prob_acc_to_noise_general}, we have
    \begin{align}
       \mathsf{PA} \left( g(.), \pare \right) &=  2\int_{(\eta-1)\Delta}^{(\eta+1)\Delta}  k_{\eta}(z) g(z) \,dz \nonumber \\
       & = k_{\eta}(z_1) \nonumber \\
       & = k_{\eta}(k^{-1}_{\eta}(\alpha)) \nonumber \\
       & = \alpha.
    \end{align}
    Additionally, based on \eqref{mse_general_for_bounded_noise}
    \begin{align}
        \mathsf{MSE}\big(\mathsf{mean}(.,.), \gdot, \pare \big) & =  \frac{1}{2\alpha}\int_{(\eta-1)\Delta}^{(\eta+1)\Delta} \nu_{\eta}(z)g(z) \,dz  \nonumber \\
        & = \frac{1}{4\alpha} \nu_{\eta}(z_1) \nonumber \\
        & \overset{(a)}{=} \frac{1}{4\alpha} \nu_{\eta}(k^{-1}_{\eta}(\alpha)) \nonumber \\
        & \overset{(b)}{=} \frac{h_{\eta}(\alpha)}{4\alpha} \nonumber \\
        & \overset{(c)}{=} \frac{h^*_{\eta}(\alpha)}{4\alpha},
    \end{align}
    where (a) follows from the definition of $z_1$ and (b) follows from the definition of $h_{\eta}(.)$, and (c) follows from the fact that in this case we assume that we have $h^*_{\eta}(\alpha) = h_{\eta}(\alpha)$.
    \item $h^*_{\eta}(\alpha) \neq h_{\eta}(\alpha)$: In this case, since $h^*_{\eta}(.)$ is the upper-concave envelop of $h_{\eta}(.)$, there exist $q_1$ and $q_2$, such that $0 \leq q_1 < \alpha < q_2 \leq 1$ and 
    \begin{align}\label{concave_envelop_boundries}
        h^*_{\eta}(q_1) = h_{\eta}(q_1), \nonumber \\
        h^*_{\eta}(q_2) = h_{\eta}(q_2).
    \end{align}
    In addition,  for all $q_1 \leq q \leq q_2$, we have
    \begin{align}\label{linear_concave}
        h^*_{\eta}(q) = \frac{h_{\eta}(q_2) - h_{\eta}(q_1)}{q_2 - q_1} (q - q_1) + h_{\eta}(q_1).
    \end{align}
    Let $z_1 = k^{-1}_{\eta}(q_1)$, $z_2 = k^{-1}_{\eta}(q_2)$, $\beta_1 = \frac{q_2 -\alpha }{2(q_2 - q_1)}$, $\beta_2 = \frac{\alpha - q_1}{2(q_2 - q_1)}$, and 
    \begin{align}
        g(z) = \beta_1 \delta(z+z_1) +\beta_2 \delta(z+z_2) \nonumber \\
        +\beta_1 \delta(z-z_1) +\beta_2 \delta(z-z_2).
    \end{align}
    Note that we can calculate $z_1$ and $z_2$ based on \eqref{inverse_of_CDF}. One can verify that
    \begin{align}
        2\beta_1 + 2\beta_2 &= 1,  \\
        2\beta_1 q_1 + 2\beta_2 q_2 &= \alpha \label{q_plus_beta}.
    \end{align}
    Based on \eqref{prob_acc_to_noise_general}, we have
    \begin{align}
       \mathsf{PA} \left( g(.), \pare \right) &=  2\int_{(\eta-1)\Delta}^{(\eta+1)\Delta}  k_{\eta}(z) g(z) \,dz \nonumber \\
       & = 2\beta_1k_{\eta}(z_1) + 2\beta_2 k_{\eta}(z_2) \nonumber \\
       & = 2\beta_1 q_1 + 2\beta_2 q_2 \nonumber \\
       & \overset{(a)}{=} \alpha,
    \end{align}
    where (a) follows from \ref{q_plus_beta}.
    Additionally, based on \eqref{mse_general_for_bounded_noise}
    \begin{align}
        \mathsf{MSE}&\big(\mathsf{mean}(.,.), \gdot, \pare \big) \nonumber \\
        & =  \frac{1}{2\alpha}\int_{(\eta-1)\Delta}^{(\eta+1)\Delta} \nu_{\eta}(z)g(z) \,dz  \nonumber \\
        & = \frac{1}{2\alpha} \left( \beta_1\nu_{\eta}(z_1) + \beta_2\nu_{\eta}(z_2)\right)\nonumber \\
        &\overset{(a)}{=} \frac{1}{2\alpha} \left( \beta_1\nu_{\eta}(k^{-1}_{\eta}(q_1)) + \beta_2\nu_{\eta}(k^{-1}_{\eta}(q_2))\right)\nonumber\\
        & = \frac{1}{4\alpha} \left( 2\beta_1 h_{\eta}(q_1) + 2\beta_2 h_{\eta}(q_2)\right) \nonumber \\
        & \overset{(b)}{=} \frac{1}{4\alpha} \left( 2\beta_1 h^*_{\eta}(q_1) + 2\beta_2 h^*_{\eta}(q_2)\right) \nonumber \\
        &\overset{(c)}{=} \frac{1}{4\alpha} h^*_{\eta}\left( 2\beta_1 q_1 + 2\beta_2 q_2\right) \nonumber \\
        & \overset{(d)}{=} \frac{h^*_{\eta}(\alpha)}{4\alpha},
    \end{align}
    where (a) follows from the definition of $h_{\eta}(.)$ , and (b) follows from \eqref{concave_envelop_boundries}, (c) follows from the fact that based on \ref{linear_concave}, $h^*(.)$ is a linear function in $[q_1, q_2]$ and $2\beta_1 + 2\beta_2 = 1$, and (d) follows from \eqref{q_plus_beta}. This completes the proof of Claim \ref{lemma:lower_bound_of_mean}.
\end{enumerate}
\end{proof}

\section{Conclusion}\label{sec:conclusion}

In this paper, we introduced a novel game-theoretic framework for coding theory, called \emph{game of coding}, which extends the conventional coding theory beyond its fundamental limitations, makes it applicable for emerging decentralized applications. Our framework acknowledges the unique dynamics of decentralized platforms, where adversaries have an incentive to ensure data recoverability rather than merely causing disruption.

Key contributions of this work include:

\begin{itemize}
    \item \textbf{Game-Theoretic Approach}: We formulated the interaction between the data collector and adversaries as a game, where both parties optimize their utility functions. This perspective allows us to navigate the limitations of traditional trust assumptions by leveraging the rational behavior of adversaries who benefit from the system's liveness.

    \item \textbf{Equilibrium Characterization}: Focusing on repetition codes with a factor of two, we characterized the equilibrium strategies for both the data collector and the adversary. This equilibrium ensures that the system remains live and operational, even in scenarios with no honest majority.

    \item \textbf{Utility Functions and Practical Relevance}: Our approach accommodates a broad class of utility functions with minimal assumptions, making it applicable to various practical applications, such as decentralized machine learning (DeML) and blockchain oracles. By ensuring the system's liveness and robustness against adversarial behaviors, our framework enhances the reliability and applicability of decentralized platforms.

    \item \textbf{Implications for Decentralized Systems}: The game of coding framework demonstrates that, under rational adversary behavior, it is possible to improve the liveness and functionality of decentralized systems without relying on stringent trust assumptions. This opens new avenues for secure and efficient coding in decentralized applications.
\end{itemize}

In conclusion, our work provides a foundational step towards integrating game-theoretic principles into coding theory, offering a promising solution to the transcend the fundamental limitations in conventional coding theory.

Future work can expand to evaluate the effect of having more nodes and power for the adversary and also explore more advanced coding techniques.


\bibliographystyle{ieeetr}
\bibliography{Abbr,Ref}

@STRING{STOC = "Proc. ACM STOC"}

@STRING{ISIT = "Proc. IEEE ISIT"}

@STRING{ITW = "Proc. IEEE ITW"}

@STRING{STOC = "Proceedings of the ACM Symposium on Theory of Computing"}

@STRING{ISIT = "Proceedings of the IEEE International Symposium on Information Theory"}

@STRING{ITW = "Proceedings of the Information Theory Workshop"}

@article{akbari2024game,
  title={Game of Coding: Sybil Resistant Decentralized Machine Learning with Minimal Trust Assumption},
  author={Akbari Nodehi, Hanzaleh and Cadambe, Viveck R and Maddah-Al, Mohammad Ali},
  journal={arXiv e-prints},
  pages={arXiv--2410},
  year={2024}
}

@article{akbarinodehi2025game,
  title={Game of Coding With an Unknown Adversary},
  author={Akbarinodehi, Hanzaleh and Moradi, Parsa and Maddah-Ali, Mohammad Ali},
  journal={arXiv preprint arXiv:2502.07109},
  year={2025},
  booktitle={arXiv preprint}
}

@article{yu2017polynomial,
  title={Polynomial codes: an optimal design for high-dimensional coded matrix multiplication},
  author={Yu, Qian and Maddah-Ali, Mohammad and Avestimehr, Salman},
  journal={Advances in Neural Information Processing Systems},
  volume={30},
  year={2017}
}

@article{akbari2021secure,
  title={Secure coded multi-party computation for massive matrix operations},
  author={Akbari-Nodehi, Hanzaleh and Maddah-Ali, Mohammad Ali},
  journal={IEEE Transactions on Information Theory},
  volume={67},
  number={4},
  pages={2379--2398},
  year={2021},
  publisher={IEEE}
}

@article{sun2017capacity,
  title={The capacity of private information retrieval},
  author={Sun, Hua and Jafar, Syed Ali},
  journal={IEEE Transactions on Information Theory},
  volume={63},
  number={7},
  pages={4075--4088},
  year={2017},
  publisher={IEEE}
}

@article{roth2020analog,
  title={Analog error-correcting codes},
  author={Roth, Ron M},
  journal={IEEE Transactions on Information Theory},
  volume={66},
  number={7},
  pages={4075--4088},
  year={2020},
  publisher={IEEE}
}

@inproceedings{han2021fact,
  title={Fact and fiction: Challenging the honest majority assumption of permissionless blockchains},
  author={Han, Runchao and Sui, Zhimei and Yu, Jiangshan and Liu, Joseph and Chen, Shiping},
  booktitle={Proceedings of the 2021 ACM Asia Conference on Computer and Communications Security},
  pages={817--831},
  year={2021}
}

@techreport{gans2023zero,
  title={{"Zero Cost"} Majority Attacks on Permissionless Blockchains},
  author={Gans, Joshua S and Halaburda, Hanna},
  year={2023},
  institution={National Bureau of Economic Research}
}

@inproceedings{nodehi2018entangled,
  title={Entangled polynomial coding in limited-sharing multi-party computation},
  author={Nodehi, Hanzaleh Akbari and Najarkolaei, Seyed Reza Hoseini and Maddah-Ali, Mohammad Ali},
  booktitle={2018 IEEE Information Theory Workshop (ITW)},
  pages={1--5},
  year={2018},
  organization={IEEE}
}

@misc{bitcoin2008bitcoin,
  title={Bitcoin: A peer-to-peer electronic cash system},
  author={Bitcoin, Nakamoto S},
  year={2008}
}

@article{buterin2013ethereum,
  title={Ethereum white paper},
  author={Buterin, Vitalik and others},
  journal={GitHub repository},
  volume={1},
  pages={22--23},
  year={2013}
}

@article{ruoti2019sok,
  title={{SoK}: Blockchain technology and its potential use cases},
  author={Ruoti, Scott and Kaiser, Ben and Yerukhimovich, Arkady and Clark, Jeremy and Cunningham, Robert},
  journal={arXiv preprint arXiv:1909.12454},
  year={2019}
}

@article{shafay2023blockchain,
  title={Blockchain for deep learning: review and open challenges},
  author={Shafay, Muhammad and Ahmad, Raja Wasim and Salah, Khaled and Yaqoob, Ibrar and Jayaraman, Raja and Omar, Mohammed},
  journal={Cluster Computing},
  volume={26},
  number={1},
  pages={197--221},
  year={2023},
  publisher={Springer}
}

@inproceedings{ding2022survey,
  title={Survey on the Convergence of Machine Learning and Blockchain},
  author={Ding, Shengwen and Hu, Chenhui},
  booktitle={Proceedings of SAI Intelligent Systems Conference},
  pages={170--189},
  year={2022},
  organization={Springer}
}

@article{kayikci2024blockchain,
  title={Blockchain meets machine learning: a survey},
  author={Kayikci, Safak and Khoshgoftaar, Taghi M},
  journal={Journal of Big Data},
  volume={11},
  number={1},
  pages={1--29},
  year={2024},
  publisher={SpringerOpen}
}

@article{taherdoost2023blockchain,
  title={Blockchain and Machine Learning: A Critical Review on Security},
  author={Taherdoost, Hamed},
  journal={Information},
  volume={14},
  number={5},
  pages={295},
  year={2023},
  publisher={MDPI}
}

@article{taherdoost2022blockchain,
  title={Blockchain technology and artificial intelligence together: a critical review on applications},
  author={Taherdoost, Hamed},
  journal={Applied Sciences},
  volume={12},
  number={24},
  pages={12948},
  year={2022},
  publisher={MDPI}
}

@article{bhat2023sakshi,
  title={SAKSHI: Decentralized AI Platforms},
  author={Bhat, Suma and Chen, Canhui and Cheng, Zerui and Fang, Zhixuan and Hebbar, Ashwin and Kannan, Sreeram and Rana, Ranvir and Sheng, Peiyao and Tyagi, Himanshu and Viswanath, Pramod and others},
  journal={arXiv preprint arXiv:2307.16562},
  year={2023}
}

@inproceedings{tian2022blockchain,
  title={Blockchain for ai: A disruptive integration},
  author={Tian, Ruijiao and Kong, Lanju and Min, Xinping and Qu, Yunhao},
  booktitle={2022 IEEE 25th International Conference on Computer Supported Cooperative Work in Design (CSCWD)},
  pages={938--943},
  year={2022},
  organization={IEEE}
}

@article{salah2019blockchain,
  title={Blockchain for AI: Review and open research challenges},
  author={Salah, Khaled and Rehman, M Habib Ur and Nizamuddin, Nishara and Al-Fuqaha, Ala},
  journal={IEEE Access},
  volume={7},
  pages={10127--10149},
  year={2019},
  publisher={IEEE}
}

@inproceedings{liu2021zkcnn,
  title={{ZkCNN}: Zero knowledge proofs for convolutional neural network predictions and accuracy},
  author={Liu, Tianyi and Xie, Xiang and Zhang, Yupeng},
  booktitle={Proceedings of the 2021 ACM SIGSAC Conference on Computer and Communications Security},
  pages={2968--2985},
  year={2021}
}

@article{xing2023zero,
  title={Zero-knowledge Proof Meets Machine Learning in Verifiability: A Survey},
  author={Xing, Zhibo and Zhang, Zijian and Liu, Jiamou and Zhang, Ziang and Li, Meng and Zhu, Liehuang and Russello, Giovanni},
  journal={arXiv preprint arXiv:2310.14848},
  year={2023}
}

@inproceedings{weng2021mystique,
  title={Mystique: Efficient conversions for {Zero-Knowledge} proofs with applications to machine learning},
  author={Weng, Chenkai and Yang, Kang and Xie, Xiang and Katz, Jonathan and Wang, Xiao},
  booktitle={30th USENIX Security Symposium (USENIX Security 21)},
  pages={501--518},
  year={2021}
}

@inproceedings{mohassel2017secureml,
  title={{SecureML}: A system for scalable privacy-preserving machine learning},
  author={Mohassel, Payman and Zhang, Yupeng},
  booktitle={2017 IEEE symposium on security and privacy (SP)},
  pages={19--38},
  year={2017},
  organization={IEEE}
}

@article{lee2024vcnn,
  title={{vCNN}: Verifiable convolutional neural network based on zk-snarks},
  author={Lee, Seunghwa and Ko, Hankyung and Kim, Jihye and Oh, Hyunok},
  journal={IEEE Transactions on Dependable and Secure Computing},
  year={2024},
  publisher={IEEE}
}

@inproceedings{eskandari2021sok,
  title={{SoK}: Oracles from the ground truth to market manipulation},
  author={Eskandari, Shayan and Salehi, Mehdi and Gu, Wanyun Catherine and Clark, Jeremy},
  booktitle={Proceedings of the 3rd ACM Conference on Advances in Financial Technologies},
  pages={127--141},
  year={2021}
}

@article{breidenbach2021chainlink,
  title={Chainlink 2.0: Next steps in the evolution of decentralized oracle networks},
  author={Breidenbach, Lorenz and Cachin, Christian and Chan, Benedict and Coventry, Alex and Ellis, Steve and Juels, Ari and Koushanfar, Farinaz and Miller, Andrew and Magauran, Brendan and Moroz, Daniel and others},
  journal={Chainlink Labs},
  volume={1},
  pages={1--136},
  year={2021}
}

@article{benligiray2020decentralized,
  title={Decentralized {API}s for web 3.0},
  author={Benligiray, Burak and Milic, Sa{\v{s}}a and V{\"a}nttinen, Heikki},
  journal={API3 Foundation Whitepaper},
  year={2020}
}

@article{lee2017speeding,
  title={Speeding up distributed machine learning using codes},
  author={Lee, Kangwook and Lam, Maximilian and Pedarsani, Ramtin and Papailiopoulos, Dimitris and Ramchandran, Kannan},
  journal={IEEE Transactions on Information Theory},
  volume={64},
  number={3},
  pages={1514--1529},
  year={2017},
  publisher={IEEE}
}

@article{reisizadeh2019coded,
  title={Coded computation over heterogeneous clusters},
  author={Reisizadeh, Amirhossein and Prakash, Saurav and Pedarsani, Ramtin and Avestimehr, Amir Salman},
  journal={IEEE Transactions on Information Theory},
  volume={65},
  number={7},
  pages={4227--4242},
  year={2019},
  publisher={IEEE}
}

@inproceedings{tandon2017gradient,
  title={Gradient coding: Avoiding stragglers in distributed learning},
  author={Tandon, Rashish and Lei, Qi and Dimakis, Alexandros G and Karampatziakis, Nikos},
  booktitle={International Conference on Machine Learning},
  pages={3368--3376},
  year={2017},
  organization={PMLR}
}

@inproceedings{dutta2017coded,
  title={Coded convolution for parallel and distributed computing within a deadline},
  author={Dutta, Sanghamitra and Cadambe, Viveck and Grover, Pulkit},
  booktitle={2017 IEEE International Symposium on Information Theory (ISIT)},
  pages={2403--2407},
  year={2017},
  organization={IEEE}
}

@inproceedings{lee2017high,
  title={High-dimensional coded matrix multiplication},
  author={Lee, Kangwook and Suh, Changho and Ramchandran, Kannan},
  booktitle={2017 IEEE International Symposium on Information Theory (ISIT)},
  pages={2418--2422},
  year={2017},
  organization={IEEE}
}

@inproceedings{cadambe2023differentially,
  title={Differentially Private Secure Multiplication: Hiding Information in the Rubble of Noise},
  author={Cadambe, Viveck R and Jeong, Haewon and Calmon, Flavio P},
  booktitle={2023 IEEE International Symposium on Information Theory (ISIT)},
  pages={2207--2212},
  year={2023},
  organization={IEEE}
}

@inproceedings{yu2019lagrange,
  title={Lagrange coded computing: Optimal design for resiliency, security, and privacy},
  author={Yu, Qian and Li, Songze and Raviv, Netanel and Kalan, Seyed Mohammadreza Mousavi and Soltanolkotabi, Mahdi and Avestimehr, Salman A},
  booktitle={The 22nd International Conference on Artificial Intelligence and Statistics},
  pages={1215--1225},
  year={2019},
  organization={PMLR}
}

@inproceedings{chang2018capacity,
  title={On the capacity of secure distributed matrix multiplication},
  author={Chang, Wei-Ting and Tandon, Ravi},
  booktitle={2018 IEEE Global Communications Conference (GLOBECOM)},
  pages={1--6},
  year={2018},
  organization={IEEE}
}

@article{d2020gasp,
  title={GASP codes for secure distributed matrix multiplication},
  author={D’Oliveira, Rafael GL and El Rouayheb, Salim and Karpuk, David},
  journal={IEEE Transactions on Information Theory},
  volume={66},
  number={7},
  pages={4038--4050},
  year={2020},
  publisher={IEEE}
}

@article{jia2021capacity,
  title={On the capacity of secure distributed batch matrix multiplication},
  author={Jia, Zhuqing and Jafar, Syed Ali},
  journal={IEEE Transactions on Information Theory},
  volume={67},
  number={11},
  pages={7420--7437},
  year={2021},
  publisher={IEEE}
}

@article{freij2017private,
  title={Private information retrieval from coded databases with colluding servers},
  author={Freij-Hollanti, Ragnar and Gnilke, Oliver W and Hollanti, Camilla and Karpuk, David A},
  journal={SIAM Journal on Applied Algebra and Geometry},
  volume={1},
  number={1},
  pages={647--664},
  year={2017},
  publisher={SIAM}
}

@article{tajeddine2018private,
  title={Private information retrieval from MDS coded data in distributed storage systems},
  author={Tajeddine, Razane and Gnilke, Oliver W and El Rouayheb, Salim},
  journal={IEEE Transactions on Information Theory},
  volume={64},
  number={11},
  pages={7081--7093},
  year={2018},
  publisher={IEEE}
}

@article{zhang2019private,
  title={On private information retrieval array codes},
  author={Zhang, Yiwei and Wang, Xin and Wei, Hengjia and Ge, Gennian},
  journal={IEEE Transactions on Information Theory},
  volume={65},
  number={9},
  pages={5565--5573},
  year={2019},
  publisher={IEEE}
}

@book{SudanBook,
	title={Essential Coding Theory},
	author={Venkatesan Guruswami and Atri Rudra and  Madhu Sudan},
		publisher={Draft is Available},
	year={2022}
}

@INPROCEEDINGS{ZamirCoded,  author={Yosibash, Royee and Zamir, Ram},  booktitle={2021 11th International Symposium on Topics in Coding (ISTC)},   title={Frame Codes For Distributed Coded Computation},   year={2021},  volume={},  number={},  pages={1-5}}

@article{jahani2018codedsketch,
 author={Jahani-Nezhad, Tayyebeh and Maddah-Ali, Mohammad Ali},
journal={IEEE Transactions on Information Theory}, 
title={CodedSketch: A Coding Scheme for Distributed Computation of Approximated Matrix Multiplication}, 
year={2021},
volume={67},
number={6},
pages={4185-4196}
}

@ARTICLE{BACC,  author={Jahani-Nezhad, Tayyebeh and Maddah-Ali, Mohammad Ali},  journal={IEEE Transactions on Pattern Analysis and Machine Intelligence},   title={Berrut Approximated Coded Computing: Straggler Resistance Beyond Polynomial Computing},   year={2023},  volume={45},  number={1},  pages={111-122}}

@article{garg2023experimenting,
  title={Experimenting with Zero-Knowledge Proofs of Training},
  author={Garg, Sanjam and Goel, Aarushi and Jha, Somesh and Mahloujifar, Saeed and Mahmoody, Mohammad and Policharla, Guru-Vamsi and Wang, Mingyuan},
  journal={Cryptology ePrint Archive},
  year={2023}
}

@article{chen2022interactive,
  title={Interactive Proofs for Rounding Arithmetic},
  author={Chen, Shuo and Cheon, Jung Hee and Kim, Dongwoo and Park, Daejun},
  journal={IEEE Access},
  volume={10},
  pages={122706--122725},
  year={2022},
  publisher={IEEE}
}

@inproceedings{garg2022succinct,
  title={Succinct Zero Knowledge for Floating Point Computations},
  author={Garg, Sanjam and Jain, Abhishek and Jin, Zhengzhong and Zhang, Yinuo},
  booktitle={Proceedings of the 2022 ACM SIGSAC Conference on Computer and Communications Security},
  pages={1203--1216},
  year={2022}
}

@inproceedings{setty2012taking,
  title={Taking Proof-Based verified computation a few steps closer to practicality},
  author={Setty, Srinath and Vu, Victor and Panpalia, Nikhil and Braun, Benjamin and Blumberg, Andrew J and Walfish, Michael},
  booktitle={21st USENIX Security Symposium (USENIX Security 12)},
  pages={253--268},
  year={2012}
}

@inproceedings{sliwinski2019blockchains,
  title={Blockchains cannot rely on honesty},
  author={Sliwinski, Jakub and Wattenhofer, Roger},
  booktitle={The 19th International Conference on Autonomous Agents and Multiagent Systems (AAMAS 2020)},
  year={2019}
}

@book{von2010market,
  title={Market structure and equilibrium},
  author={Von Stackelberg, Heinrich},
  year={2010},
  publisher={Springer Science \& Business Media}
}

@article{zhao2021veriml,
  title={{VeriML}: Enabling integrity assurances and fair payments for machine learning as a service},
  author={Zhao, Lingchen and Wang, Qian and Wang, Cong and Li, Qi and Shen, Chao and Feng, Bo},
  journal={IEEE Transactions on Parallel and Distributed Systems},
  volume={32},
  number={10},
  pages={2524--2540},
  year={2021},
  publisher={IEEE}
}

@article{thaler2022proofs,
  title={Proofs, arguments, and zero-knowledge},
  author={Thaler, Justin and others},
  journal={Foundations and Trends{\textregistered} in Privacy and Security},
  volume={4},
  number={2--4},
  pages={117--660},
  year={2022},
  publisher={Now Publishers, Inc.}
}

@inproceedings{ben2019completeness,
title = {Completeness theorems for non-cryptographic fault-tolerant distributed computation},
author = {Ben-Or, Michael and Goldwasser, Shafi and Wigderson, Avi},
year = {1988},
isbn = {0897912640},
publisher = {Association for Computing Machinery},
address = {New York, NY, USA},
booktitle = {Proceedings of the Twentieth Annual ACM Symposium on Theory of Computing},
pages = {1–10},
numpages = {10},
location = {Chicago, Illinois, USA},
series = {STOC '88}
}

@article{han2023hyperattention,
  title={Hyperattention: Long-context attention in near-linear time},
  author={Han, Insu and Jayaram, Rajesh and Karbasi, Amin and Mirrokni, Vahab and Woodruff, David P and Zandieh, Amir},
  journal={arXiv preprint arXiv:2310.05869},
  year={2023}
}

@book{kay1993fundamentals,
  title={Fundamentals of statistical signal processing: estimation theory},
  author={Kay, Steven M},
  year={1993},
  publisher={Prentice-Hall, Inc.}
}

@article{han2010physical,
  title={Physical layer security game: interaction between source, eavesdropper, and friendly jammer},
  author={Han, Zhu and Marina, Ninoslav and Debbah, M{\'e}rouane and Hj{\o}rungnes, Are},
  journal={EURASIP Journal on Wireless Communications and Networking},
  volume={2009},
  pages={1--10},
  year={2010},
  publisher={Springer}
}

@article{saraydar2002efficient,
  title={Efficient power control via pricing in wireless data networks},
  author={Saraydar, Cem U and Mandayam, Narayan B and Goodman, David J},
  journal={IEEE transactions on Communications},
  volume={50},
  number={2},
  pages={291--303},
  year={2002},
  publisher={IEEE}
}

@article{wang2008distributed,
  title={Distributed relay selection and power control for multiuser cooperative communication networks using Stackelberg game},
  author={Wang, Beibei and Han, Zhu and Liu, KJ Ray},
  journal={IEEE Transactions on Mobile Computing},
  volume={8},
  number={7},
  pages={975--990},
  year={2008},
  publisher={IEEE}
}

@article{bonneau2008non,
  title={Non-atomic games for multi-user systems},
  author={Bonneau, Nicolas and Debbah, M{\'e}rouane and Altman, Eitan and Hj{\o}rungnes, Are},
  journal={IEEE Journal on Selected Areas in Communications},
  volume={26},
  number={7},
  pages={1047--1058},
  year={2008},
  publisher={IEEE}
}
\section{Biographies}

    \begin{IEEEbiographynophoto}{Hanzaleh Akbari Nodehi} is a Ph.D. candidate in the Department of Electrical and Computer Engineering at the University of Minnesota Twin Cities. He received his B.Sc. and M.Sc. degrees from the Department of Electrical Engineering, Sharif University of Technology, Tehran, Iran, in 2016 and 2018, respectively. His research interests include information theory and coding, blockchain technology, game theory, and decentralized systems.
    \end{IEEEbiographynophoto}
    
 \begin{IEEEbiographynophoto}{Viveck R. Cadambe} (S'07-- M'11 -- SM '24) is an Associate Professor in the School of Electrical and Computer Engineering at Georgia Tech. He received his Ph.D. from the University of California Irvine in 2011, and was a postdoctoral researcher jointly at the Massachusetts Institute of Technology (MIT)and Boston University between 2011 and 2014. He was an Associate Professor in the Department of Electrical Engineering at Pennsylvania State University until 2024. His core expertise is in information theory, coding theory, communication theory, and distributed algorithms. In his research, he studies applications to distributed machine learning, data privacy, cloud computing, and wireless communications. 

Dr. Cadambe has received the 2009 Information Theory Society Best Paper Award, the 2014 IEEE Network Computing and Applications Best Paper Award, an NSF CRII Award in 2015, an NSF Career Award in 2016, a Google Faculty Award in 2019 and he was a finalist for the 2016 Bell Labs Prize. He has served as an Associate Editor for IEEE Transactions on Wireless Communications, an issue of the IEEE Journal on Special Areas in Information Theory, and the IEEE Transactions on Communications.
\end{IEEEbiographynophoto}

\begin{IEEEbiographynophoto}{Mohammad Ali Maddah-Ali} (IEEE Fellow, 2023) is an Associate Professor at the University of Minnesota Twin Cities. He received his B.Sc. degree in Electrical Engineering from Isfahan University of Technology, his M.A.Sc. degree from the University of Tehran, and his Ph.D. in Electrical and Computer Engineering from the University of Waterloo, Canada, in 2007.

From 2007 to 2008, he was with the Wireless Technology Laboratories at Nortel Networks, Ottawa, ON, Canada. He then held a Postdoctoral Fellowship at the Department of Electrical Engineering and Computer Sciences, University of California, Berkeley, from 2008 to 2010. From September 2010 to September 2020, he served as a Communication Research Scientist at Nokia Bell Labs, NJ, USA.

Dr. Maddah-Ali is the recipient of several honors, including the NSERC Postdoctoral Fellowship (2007), the Best Paper Award at the IEEE International Conference on Communications (ICC) in 2014, the IEEE Communications Society and IEEE Information Theory Society Joint Paper Award in 2015, and the IEEE Information Theory Society Paper Award in 2016. He served as an Associate Editor for the IEEE Transactions on Information Theory (2019–2022) and as Lead Editor for the IEEE Journal on Selected Areas in Information Theory. He is currently a distinguished lecturer of the IEEE Information Theory Society.
\end{IEEEbiographynophoto}

\appendices
\section{Proof of Lemma \ref{lemma:making_symmetry}}\label{proof:lemma:making_symmetry}
For the formal proof, first consider the following lemma.
\begin{lemma}\label{lemma:flipping_noise}
    For any $\pare \in \Lambda_{\mathsf{DC}} $, $g(z) \in \Lambda_{\mathsf{AD}}$, we have
    \begin{align*}
        \PA (g_{\textrm{ref}}(.), \pare) &= \PA (g(.), \pare),  \\
        \mathsf{MSE}\big(\mathsf{mean}(.,.), g_{\textrm{ref}}(.), \pare\big) &= \mathsf{MSE}\big(\mathsf{mean}(.,.), \gdot, \pare\big),
    \end{align*}
    where $g_{\textrm{ref}}(z) \triangleq g(-z)$.
\end{lemma}
\begin{proof}
Let us consider two scenarios. In scenario $1$ the adversary employs $\gdot$ as the noise distribution, whereas 
in scenario 2 the adversary employs $g_{\textrm{ref}}(.)$ as the noise distribution. We note that that for the joint probability density functions  $f^{(1)}_{\mathbf{n}_a, \mathbf{n}_h}$ and $f^{(2)}_{\mathbf{n}_a, \mathbf{n}_h}$, for scenarios $1$ and $2$ respectively, we have
\begin{align}\label{scenario1_and_scenario2}
    f^{(1)}_{\mathbf{n}_a, \mathbf{n}_h} (n_a,n_h) = f^{(2)}_{\mathbf{n}_a, \mathbf{n}_h} (-n_a, -n_h).
\end{align}

Recall that the acceptance rule is a function of $|\mathbf{y}_1 - \mathbf{y}_2| = |\mathbf{n}_a - \mathbf{n}_h|$.
This implies that for any realization of the noises in scenario $1$, i.e., $(n_a, n_h)$  where the computation is accepted, the corresponding realization $(-n_a, -n_h)$ in scenario $2$ is also accepted, and vice versa. This is because the absolute value of the difference between noises, is the same for both of these realizations. Thus, we have
\begin{align}\label{relation_between_gref_g_acc}
    &\Pr \big(\mathcal{A}_{\pare}; g_{\textrm{ref}}(.)~|~ \bn_a=-n_a,\bn_h=-n_h\big) \nonumber \\
    &= \Pr \big(\mathcal{A}_{\pare}; \gdot~|~\bn_a=n_a,\bn_h=n_h\big).
\end{align}

This implies that
\begin{align}\label{flipped_noise_Acc}
  \PA&\bigl(g_{\mathrm{ref}}(\cdot),\eta\bigr)
  \nonumber \\
  &=
  \int_{-\infty}^{\infty}\!\!\int_{-\Delta}^{\Delta}
    \Pr\bigl(\mathcal{A}_{\eta};\,g_{\mathrm{ref}}(\cdot)\mid \bn_a=n_a,\bn_h=n_h\bigr)
  \nonumber\\
  &\quad\times
    f^{(2)}_{\bn_a,\bn_h}(n_a,n_h)\;dn_h\,dn_a
  \nonumber\\
  &\overset{(a)}{=}
  \int_{-\infty}^{\infty}\!\!\int_{-\Delta}^{\Delta}
    \Pr\bigl(\mathcal{A}_{\eta};\,g_{\mathrm{ref}}(\cdot)\mid \bn_a=-n_a',\bn_h=-n_h'\bigr)
  \nonumber\\
  &\quad\times
    f^{(2)}_{\bn_a,\bn_h}(-n_a',-n_h')\;dn_h'\,dn_a'
  \nonumber\\
  &\overset{(b)}{=}
  \int_{-\infty}^{\infty}\!\!\int_{-\Delta}^{\Delta}
    \Pr\bigl(\mathcal{A}_{\eta};\,g(\cdot)\mid \bn_a=n_a',\bn_h=n_h'\bigr)
  \nonumber\\
  &\quad\times
    f^{(1)}_{\bn_a,\bn_h}(n_a',n_h')\;dn_h'\,dn_a'
  \nonumber\\
  &=
  \PA\bigl(g(\cdot),\eta\bigr).
\end{align}
where (a) follows from changing the parameters of the integral by $n^{\prime}_h = -n_h$, $n^{\prime}_a = -n_a$, (b) follows from \eqref{scenario1_and_scenario2} and \eqref{relation_between_gref_g_acc}.

Additionally, one can verify that
\begin{align}\label{flipped_noise_Mse}
  &\mathsf{MSE}\bigl(\mathsf{mean}(\cdot,\cdot),\,g_{\mathrm{ref}}(\cdot),\,\eta\bigr)
  \nonumber\\
  &= 
  \int_{-\infty}^{\infty}\!\!\int_{-\Delta}^{\Delta}
    \mathbb{E}\!\Bigl[\bigl(\mathbf{u}-\tfrac{\mathbf{y}_1+\mathbf{y}_2}{2}\bigr)^2 
      \mid \mathcal{A}_{\eta},\,n_a,n_h;\,g_{\mathrm{ref}}(\cdot)\Bigr]
  \nonumber\\
  &\quad\times 
    f^{(2)}_{\mathbf{n}_a,\mathbf{n}_h\mid\mathcal{A}_{\eta}}
      \!(n_a,n_h \mid \mathcal{A}_{\eta};\,g_{\mathrm{ref}}(\cdot))
    \,dn_h\,dn_a
  \nonumber\\
  &= 
  \int_{-\infty}^{\infty}\!\!\int_{-\Delta}^{\Delta}
    \frac{(n_a + n_h)^2}{4}\, \nonumber \\
    &\quad \times
    \frac{\Pr\bigl(\mathcal{A}_{\eta};g_{\mathrm{ref}}(\cdot)\mid n_a,n_h\bigr)\,
          f^{(2)}_{\mathbf{n}_a,\mathbf{n}_h}(n_a,n_h)}
         {\Pr\bigl(\mathcal{A}_{\eta};g_{\mathrm{ref}}(\cdot)\bigr)} 
    \,dn_h\,dn_a
  \nonumber\\
  &\overset{(a)}{=}
  \int_{-\infty}^{\infty}\!\!\int_{-\Delta}^{\Delta}
    \frac{(-n'_a - n'_h)^2}{4}\, \nonumber \\
    &\quad\times \frac{\Pr\bigl(\mathcal{A}_{\eta};g_{\mathrm{ref}}(\cdot)\mid -n'_a,-n'_h\bigr)\,
          f^{(2)}_{\mathbf{n}_a,\mathbf{n}_h}(-n'_a,-n'_h)}
         {\Pr\bigl(\mathcal{A}_{\eta};g_{\mathrm{ref}}(\cdot)\bigr)}
    \,dn'_h\,dn'_a
  \nonumber\\
  &\overset{(b)}{=}
  \int_{-\infty}^{\infty}\!\!\int_{-\Delta}^{\Delta}
    \frac{(n'_a + n'_h)^2}{4}\, \nonumber \\
    &\quad \times
    \frac{\Pr\bigl(\mathcal{A}_{\eta};g(\cdot)\mid n'_a,n'_h\bigr)\,
          f^{(1)}_{\mathbf{n}_a,\mathbf{n}_h}(n'_a,n'_h)}
         {\Pr\bigl(\mathcal{A}_{\eta};g(\cdot)\bigr)}
    \,dn'_h\,dn'_a
  \nonumber\\
  &= 
  \int_{-\infty}^{\infty}\!\!\int_{-\Delta}^{\Delta}
    \mathbb{E}\!\Bigl[\bigl(\mathbf{u}-\tfrac{\mathbf{y}_1+\mathbf{y}_2}{2}\bigr)^2 
      \mid \mathcal{A}_{\eta},\,n'_a,n'_h;\,g(\cdot)\Bigr]
  \nonumber\\
  &\quad\times
    f^{(1)}_{\mathbf{n}_a,\mathbf{n}_h\mid\mathcal{A}_{\eta}}
      \!(n'_a,n'_h \mid \mathcal{A}_{\eta};\,g(\cdot))
    \,dn'_h\,dn'_a
  \nonumber\\
  &=
  \mathsf{MSE}\bigl(\mathsf{mean}(\cdot,\cdot),\,g(\cdot),\,\eta\bigr).
\end{align}
where (a) follows from changing the parameters of the integral by $n^{\prime}_h = -n_h$, $n^{\prime}_a = -n_a$, (b) follows from \eqref{scenario1_and_scenario2}, \eqref{relation_between_gref_g_acc}, and \eqref{flipped_noise_Acc}.

\end{proof}
Now we prove Lemma \ref{lemma:making_symmetry}.
  Consider a scenario where the adversary uses a $\mathsf{Bernoulli}(\frac{1}{2})$ random variable $\mathbf{c}$ to generate its noise distribution. Specifically, the adversary chooses $\gdot$ as its noise distribution when $\mathbf{c}=0$ and employs $g_{\textrm{ref}}(.)$ otherwise. The random variable $\mathbf{c}$ is independent of all other random variables in the system. We denote the noise distribution of the adversary in this scenario as $g_c(.)$. It can be easily verified that we have $g_c(z) = g_{\textrm{sym}}(z)$. Therefore, we have
\begin{align}\label{bernoli_noise_
acc}
    \Pr \left(\mathcal{A}_{\pare}; g_{\textrm{sym}}(.)\right) &=\Pr \left(\mathcal{A}_{\pare}; g_c(.)\right) \nonumber \\
    &= \Pr \left(\mathcal{A}_{\pare}; g_c(.) | \mathbf{c} = 0\right)\Pr(\mathbf{c} = 0) \nonumber \\
    &\quad+ 
    \Pr \left(\mathcal{A}_{\pare}; g_c(.) | \mathbf{c} = 1\right)\Pr(\mathbf{c} = 1)
    \nonumber\\
    &=\Pr \left(\mathcal{A}_{\pare}; \gdot\right)\frac{1}{2} + \Pr \left(\mathcal{A}_{\pare}; g_{\textrm{ref}}(.)\right)\frac{1}{2} \nonumber \\
    &=\frac{\Pr \left(\mathcal{A}_{\pare}; \gdot\right) + \Pr \left(\mathcal{A}_{\pare}; g_{\textrm{ref}}(.)\right)}{2} \nonumber \\
    &\overset{(a)}{=} \Pr \left(\mathcal{A}_{\pare}; \gdot\right),
\end{align}
where (a) follows from Lemma \ref{lemma:flipping_noise}. 

Additionally, we have
\begin{align}
    &\mathsf{MSE}\big(\mathsf{mean}(.,.), g_{\textrm{sym}}(.), \pare\big) \nonumber \\
    &= \mathsf{MSE}\big(\mathsf{mean}(.,.), g_c(.), \pare\big) \\ \nonumber 
    &=  \mathbb{E}\big[\big(\mathbf{u} - \frac{\by_1+\by_2}{2}\big)^2 | \mathcal{A}_{\pare}; g_c(.)\big] \nonumber\\
    &=  \mathbb{E}\big[\big(\mathbf{u} - \frac{\by_1+\by_2}{2}\big)^2 | \mathcal{A}_{\pare}, \mathbf{c}=0; g_c(.)\big]\nonumber \\
    &\quad \times\Pr \big( \mathbf{c}=0|\mathcal{A}_{\pare}; g_c(.) \big) \nonumber \\
    &\quad+ 
    \mathbb{E}\big[\big(\mathbf{u} - \frac{\by_1+\by_2}{2}\big)^2 | \mathcal{A}_{\pare}, \mathbf{c}=1; g_c(.)\big] \nonumber \\
    &\quad \times\Pr \big( \mathbf{c}=1|\mathcal{A}_{\pare}; g_c(.) \big) \nonumber \\
    &= \mathbb{E}\big[\big(\mathbf{u} - \frac{\by_1+\by_2}{2}\big)^2 | \mathcal{A}_{\pare}; \gdot\big]
    \nonumber \\
    &\quad \times\frac{\Pr \big( \mathcal{A}_{\pare}; g_c(.)|\mathbf{c} = 0 \big)\Pr(\mathbf{c}=0)}{\Pr \big( \mathcal{A}_{\pare}; g_c(.) \big)} \nonumber\\
    &+\mathbb{E}\big[\big(\mathbf{u} - \frac{\by_1+\by_2}{2}\big)^2 | \mathcal{A}_{\pare}; g_{\textrm{ref}}(.)\big]
    \nonumber \\
    &\quad \times\frac{\Pr \big( \mathcal{A}_{\pare}; g_c(.)|\mathbf{c} = 1 \big)\Pr(\mathbf{c}=1)}{\Pr \big( \mathcal{A}_{\pare}; g_c(.) \big)}\nonumber \\
    = &\mathsf{MSE}\big(\mathsf{mean}(.,.), \gdot, \pare\big)
    \frac{\frac{1}{2}\Pr \big( \mathcal{A}_{\pare}; \gdot \big)}{\Pr \big( \mathcal{A}_{\pare}; g_c(.) \big)} \nonumber \\
    &+ \mathsf{MSE}\big(\mathsf{mean}(.,.), g_{\textrm{ref}}(.), \pare\big)
    \frac{\frac{1}{2}\Pr \big( \mathcal{A}_{\pare}; g_{\textrm{ref}}(.) \big)}{\Pr \big( \mathcal{A}_{\pare}; g_c(.) \big)} \nonumber \\
    \overset{(a)}{=} & \frac{1}{2}\mathsf{MSE}\big(\mathsf{mean}(.,.), \gdot, \pare\big) + \frac{1}{2}\mathsf{MSE}\big(\mathsf{mean}(.,.), g_{\textrm{ref}}(.), \pare\big)
    \nonumber \\
     \overset{(b)}{=} & \mathsf{MSE}\big(\mathsf{mean}(.,.), \gdot, \pare\big),
\end{align}
where (a) follows from \eqref{bernoli_noise_
acc}, and (b) follows from Lemma \ref{lemma:flipping_noise}.
This completes the proof of Lemma \ref{lemma:making_symmetry}.

\section{Proof of Lemma 
\ref{MMSE_for_Internal_Part}}\label{proof:MMSE_for_Internal_Part}
Recall that based on the definition  we have
\begin{align}
 \est^*_{\pare,g} = \underset{\est: \mathbb{R}^2 \to \mathbb{R}}{\arg\min} ~\mathsf{MSE}\big(\est(.,.), \gdot, \pare\big).
 \end{align}
To prove this lemma, based on the orthogonality principle in minimum mean square estimation\cite{kay1993fundamentals}, it is necessary and sufficient to show that for any function $\nu: \mathbb{R}^2 \to \mathbb{R}$, we have
\begin{align}\label{best_mmse_v1}
    \mathbb{E}&\bigg[ (\bu - \frac{\by_1 + \by_2}{2}) \nu(\by_1, \by_2) \big|\nonumber \\
    &\quad\acce, |\by_1+\by_2| \leq 2M - (\eta+2)\Delta\bigg] = 0.
\end{align}
Let $\mu(y_1,y_2) \triangleq \nu(\frac{y_1+y_2}{2}, \frac{y_1-y_2}{2})$. Thus $\mu(y_1+y_2, y_1-y_2) = \nu(y_1, y_2)$. Therefore, we can reformulate \eqref{best_mmse_v1} to
\begin{align}\label{best_mmse_v2}
    \mathbb{E}&\bigg[ (\bu - \frac{\by_1 + \by_2}{2}) \mu(\by_1+\by_2, \by_1-\by_2) \big| \nonumber \\
    &\quad\acce, |\by_1+\by_2| \leq 2M - (\eta+2)\Delta\bigg] = 0.
\end{align}
Note that $\bu - \frac{\by_1 + \by_2}{2} = \frac{\bn_1+\bn_2}{2}$.
On the other hand, the event of $\acce$ is equivalent to the event of $|\by_1 - \by_2| = |\bn_1 - \bn_2| \leq \eta \Delta$. Moreover, the event of $|\by_1+\by_2| \leq 2M - (\eta+2)\Delta$ is equivalent to the event of $\{\bu \leq \Gamma - \frac{\bn_1 + \bn_2}{2}, \bu \geq -\Gamma - \frac{\bn_1 + \bn_2}{2} \}$, where $$\Gamma \triangleq M - \frac{(\eta + 2)\Delta}{2}.$$ Therefore, we can reformulate \eqref{best_mmse_v2} to 
\begin{align}\label{best_mmse_v3}
    &\mathbb{E}\bigg[ (\bn_1+\bn_2) \mu(\by_1+\by_2, \bn_1-\bn_2) \bigg|\nonumber \\
    &\quad|\bn_1 - \bn_2| \leq \eta \Delta, \bu \leq \Gamma - \frac{\bn_1 + \bn_2}{2}, \bu \geq -\Gamma - \frac{\bn_1 + \bn_2}{2}\bigg] \nonumber \\
    &= 0.
\end{align}
To prove Lemma \ref{MMSE_for_Internal_Part}, we show \eqref{best_mmse_v3}. Note that
\begin{align}\label{best_mmse_v4}
    &\mathbb{E}\bigg[ (\bn_1+\bn_2) \mu(\by_1+\by_2, \bn_1-\bn_2) \bigg| \nonumber\\
    &\quad|\bn_1 - \bn_2| \leq \eta \Delta, \bu \leq \Gamma - \frac{\bn_1 + \bn_2}{2}, \bu \geq -\Gamma - \frac{\bn_1 + \bn_2}{2}\bigg] \nonumber \\
    &= \mathbb{E}_{\bn_1,\bn_2}\bigg[ \mathbb{E}\bigg[ (n_1+n_2) \mu(\by_1+\by_2, n_1-n_2) \bigg| \nonumber \\
    &\quad|\bn_1 - \bn_2| \leq \eta \Delta, \bu \leq \Gamma - \frac{\bn_1 + \bn_2}{2},\nonumber\\
    &\quad \bu \geq -\Gamma - \frac{\bn_1 + \bn_2}{2}, \bn_1 =n_1, \bn_2 =n_2\bigg]\bigg].
\end{align}
In addition, given the condition of $|\bn_1 - \bn_2| \leq \eta \Delta$, we have
\begin{align}
    |\bn_1 + \bn_2| &= |\bn_h + \bn_a| \leq 2|\bn_h| + |\bn_a - \bn_h| 
 \nonumber \\
 &= 2|\bn_h| + |\bn_1 - \bn_2| \leq 2\Delta + \eta\Delta = (\eta + 2)\Delta.
\end{align}
This implies that $\Gamma - \frac{\bn_1 + \bn_2}{2} \leq \Gamma + \frac{(\eta + 2)\Delta}{2} \leq M$, and $-\Gamma - \frac{\bn_1 + \bn_2}{2} \geq -\Gamma - \frac{(\eta + 2)\Delta}{2} \geq -M$. Therefore, since $\mathbf{u} \sim \text{unif}[-M,M]$, given $n_1,n_2$ and the fact that $|n_1 - n_2| \leq \eta \Delta$, the event of $\{ \bu \leq \Gamma - \frac{n_1 + n_2}{2}, \bu \geq -\Gamma - \frac{n_1 + n_2}{2}\}$ implies that 
$\mathbf{u}  \sim \text{unif}[-\Gamma - \frac{n_1 + n_2}{2},\Gamma - \frac{n_1 + n_2}{2}]$ and thus $\by_1 + \by_2 \sim \text{unif}[-2\Gamma ,2\Gamma ]$. Let $\mathbf{v} \triangleq \by_1 + \by_2 $ be a random variable in $[-2\Gamma ,2\Gamma ]$. We can reformulate \eqref{best_mmse_v4} as
\begin{align}\label{best_mmse_v5}
    &\mathbb{E}_{\bn_1,\bn_2}\bigg[ \mathbb{E}\bigg[ (n_1+n_2) \mu(\by_1+\by_2, n_1-n_2) \bigg| \nonumber \\
    &\quad |\bn_1 - \bn_2| \leq \eta \Delta, \bu \leq \Gamma - \frac{\bn_1 + \bn_2}{2}, \nonumber \\
    &\quad \bu \geq -\Gamma - \frac{\bn_1 + \bn_2}{2}, \bn_1 =n_1, \bn_2 =n_2\bigg]\bigg] \nonumber \\
    &=\mathbb{E}_{\bn_1,\bn_2}\bigg[ \mathbb{E}\bigg[ (n_1+n_2) \mu(\mathbf{v}, n_1-n_2) \bigg| \nonumber \\
    &\quad |\bn_1 - \bn_2| \leq \eta \Delta,  \bn_1 =n_1, \bn_2 =n_2\bigg]\bigg] \nonumber \\
    & = \mathbb{E}\bigg[ (\bn_1+\bn_2) \mu(\mathbf{v}, \bn_1-\bn_2) \bigg| |\bn_1 - \bn_2| \leq \eta \Delta \bigg].
\end{align}
 Let $\mathbf{t} \triangleq \bn_1 - \bn_2$, and $\mathbf{s} \triangleq \bn_1 + \bn_2$.  Recall that both $\mathbf{n}_h$ and $\mathbf{n}_a$ have symmetric distributions. Additionally, with a probability of $\frac{1}{2}$, node $1$ is honest and the second node is adversarial, and with a probability of $\frac{1}{2}$, it is vice versa. Therefore, each of the pairs of $(n_1, n_2)$ and $(-n_2, -n_1)$ are equally probable. This implies that for any value of $t$, $f_{\mathbf{s}|\mathbf{t}}(s|t)$ is symmetric with respect to $s$.
Also, $f_{\mathbf{s}|\mathbf{t}, |\mathbf{t}| \leq \eta\Delta}(s|t, |\mathbf{t}| \leq \eta\Delta)$ is symmetric with respect to $s$.
 Using this, we reformulate \eqref{best_mmse_v5} as
\begin{align}\label{best_mmse_v6}
     \mathbb{E}\bigg[& (\bn_1+\bn_2) \mu(\mathbf{v}, \bn_1-\bn_2) \bigg| |\bn_1 - \bn_2| \leq \eta \Delta \bigg] \nonumber \\
     &= \mathbb{E}\bigg[ \mathbf{s} \mu(\mathbf{v}, \mathbf{t}) \bigg| |\mathbf{t}| \leq \eta \Delta \bigg] \nonumber \\
     &=\mathbb{E}_{\mathbf{v},\mathbf{t}} \bigg[ 
     \mathbb{E} \big[
     \mathbf{s} \mu(v, t) \big| |\mathbf{t}| \leq \eta \Delta, v,t
     \big]
     \bigg] \nonumber \\
     & = \mathbb{E}_{\mathbf{v},\mathbf{t}} \bigg[ \mu(v, t)
     \mathbb{E} \big[
     \mathbf{s}  \big| |\mathbf{t}| \leq \eta \Delta, v,t
     \big]
     \bigg] \nonumber \\
     &\overset{(a)}{=}
     \mathbb{E}_{\mathbf{v},\mathbf{t}} \bigg[ \mu(v, t)
     \mathbb{E} \big[
     \mathbf{s}  \big| |\mathbf{t}| \leq \eta \Delta, t
     \big]
     \bigg] \nonumber \\
     & \overset{(b)}{=}
     0,
\end{align}
where (a) follows from the fact that $\mathbf{v} \sim \text{unif}[-2\Gamma ,2\Gamma ]$ and independent to $\mathbf{s}$, and (b) follows from the fact that $f_{\mathbf{s}|\mathbf{t}, |\mathbf{t}| \leq \eta\Delta}(s|t, |\mathbf{t}| \leq \eta\Delta)$ is symmetric with respect to $s$, resulting in $\mathbb{E} \big[
     \mathbf{s}  \big| |\mathbf{t}| \leq \eta \Delta, t
     \big] = 0$.
     This completes the proof of Lemma \ref{MMSE_for_Internal_Part}.

\section{Proof of Lemma \ref{bound_mse_symmetric_mean}}\label{proof:bound_mse_symmetric_mean}

    Note that $\mathsf{MMSE}\big( \gdot, \pare\big) = \mathsf{MSE}\big(\est^*_{\pare,g}(.,.), \gdot, \pare\big)$. Based on definition \eqref{MMSE_definition}, we have 
    \begin{align}
        \mathsf{MMSE}\big( \gdot, \pare\big) - \mathsf{MSE}\big(\mathsf{mean}(.,.), \gdot, \pare\big) \leq 0.
    \end{align}
 To complete the proof of this lemma, we prove the lower bound in \eqref{statement_of_lemma_for_mmse_and_mean}. Based on Lemma \ref{MMSE_for_Internal_Part}, for the case of $|y_1+y_2| \leq 2M - (\eta+2)\Delta$, we have
\begin{align}
    \est^*_{\pare,g}\big( \underline{y} \big) = \mathbb{E}\big[\mathbf{u} | \mathcal{A}_{\pare}, \underline{y}\big] = \frac{y_1+y_2}{2}.
\end{align}
 One can verify
\begin{align}\label{MMSE_lower_bound_general}
     &\mathsf{MMSE}\big( \gdot, \pare\big) \nonumber \\
     &=  \mathbb{E}\big[\big(\mathbf{u} - \est^*_{\pare,g}(\mathbf{\underline{y}})\big)^2 | \mathcal{A}_{\pare}\big] \nonumber \\
     &=\mathbb{E}_{\mathbf{\underline{y}}|\acce}\Big[ \mathbb{E}\big[\big(\mathbf{u} - \est^*_{\pare,g}(\underline{y})\big)^2 | \mathcal{A}_{\pare} ,\mathbf{\underline{y}} = \underline{y}\big] \Big] \nonumber \\
     &\overset{(a)}{=} \mathbb{E}_{\mathbf{\underline{y}}|\acce}\Big [ \mathbb{E}\big[\big(\mathbf{u} - \frac{y_1+y_2}{2}\big)^2 | \mathcal{A}_{\pare} ,\mathbf{\underline{y}} = \underline{y}\big] \nonumber \\
     &\quad \mathbbm{1}(|y_1+y_2| \leq 2M - (\eta+2)\Delta) \nonumber \\ \left.
     \right. &\quad+ \mathbb{E}\big[\big(\mathbf{u} - \est^*_{\pare,g}(\underline{y})\big)^2 | \mathcal{A}_{\pare} ,\mathbf{\underline{y}} = \underline{y}\big] \nonumber \\
     &\quad \mathbbm{1}(|y_1+y_2| > 2M - (\eta+2)\Delta)\Big ] \nonumber \\
      &\geq \mathbb{E}_{\mathbf{\underline{y}}|\acce}\Big [ \mathbb{E}\big[\big(\mathbf{u} - \frac{y_1+y_2}{2}\big)^2 | \mathcal{A}_{\pare} ,\mathbf{\underline{y}} = \underline{y}\big] \nonumber \\
      &\quad \mathbbm{1}(|y_1+y_2| \leq 2M - (\eta+2)\Delta)  \Big],
 \end{align}
 where (a) follows from Lemma \ref{MMSE_for_Internal_Part}.

 On the other hand, one can verify that
 \begin{align}\label{expansion_mean_estimator}
     &\mathsf{MSE}\big(\mathsf{mean}(.,.), \gdot, \pare\big)  \nonumber \\
     &=  \mathbb{E}\big[\big(\mathbf{u} - \frac{\by_1+\by_2}{2}\big)^2 | \mathcal{A}_{\pare}\big] \nonumber \\
     &=\mathbb{E}_{\mathbf{\underline{y}}|\acce}\Big[ \mathbb{E}\big[\big(\mathbf{u} - \frac{\by_1+\by_2}{2}\big)^2 | \mathcal{A}_{\pare} ,\mathbf{\underline{y}} = \underline{y}\big] \Big] \nonumber \\
     &=\mathbb{E}_{\mathbf{\underline{y}}|\acce}\Big [ \mathbb{E}\big[\big(\mathbf{u} - \frac{y_1+y_2}{2}\big)^2 | \mathcal{A}_{\pare} ,\mathbf{\underline{y}} = \underline{y}\big] \nonumber \\
     &\quad\mathbbm{1}(|y_1+y_2| \leq 2M - (\eta+2)\Delta) \nonumber \\ \left.
     \right. &\quad+ \mathbb{E}\big[\big(\mathbf{u} - \frac{y_1+y_2}{2}\big)^2 | \mathcal{A}_{\pare} ,\mathbf{\underline{y}} = \underline{y}\big] \nonumber \\
     &\quad\mathbbm{1}(|y_1+y_2| > 2M - (\eta+2)\Delta)\Big ].
 \end{align}

 Based on \eqref{MMSE_lower_bound_general} and \eqref{expansion_mean_estimator}, we have
 \begin{align}\label{lowerbound_difference_mmse_mean}
     &\mathsf{MMSE}\big( \gdot, \pare\big) - \mathsf{MSE}\big(\mathsf{mean}(.,.), \gdot, \pare\big) \nonumber \\ 
     &\geq - \mathbb{E}_{\mathbf{\underline{y}}|\acce}\Big [\mathbb{E}\big[\big(\mathbf{u} - \frac{y_1+y_2}{2}\big)^2 | \mathcal{A}_{\pare} ,\mathbf{\underline{y}} = \underline{y}\big] \nonumber \\
     &\quad\mathbbm{1}(|y_1+y_2| > 2M - (\eta+2)\Delta)\Big ].
 \end{align}
Note that
\begin{align}\label{lower_bound_mean_mase_outside_internal}
    &\mathbb{E}\big[\big(\mathbf{u} - \frac{y_1+y_2}{2}\big)^2 | \mathcal{A}_{\pare} ,\mathbf{\underline{y}} = \underline{y}\big] \nonumber \\
    &= \frac{1}{4}\mathbb{E}\big[ (\bn_a + \bn_h)^2 | ~|\bn_a - \bn_h| \leq \eta \Delta,~\mathbf{\underline{y}} = \underline{y}\big] \nonumber \\
    & \overset{(a)}{\leq} \frac{1}{4}\mathbb{E}\big[ (\bn_a - \bn_h)^2 | ~|\bn_a - \bn_h| \leq \eta \Delta,~\mathbf{\underline{y}} = \underline{y}\big]\nonumber \\
    &\quad + \frac{1}{4}\mathbb{E}\big[ (2\bn_h)^2 | ~|\bn_a - \bn_h| \leq \eta \Delta,~\mathbf{\underline{y}} = \underline{y}\big] \nonumber \\
    & \leq \frac{(\eta \Delta)^2}{4} + \frac{(2\Delta)^2}{4} \nonumber \\
    & = \frac{(\eta^2+4)\Delta^2}{4},
\end{align}
where (a) follows from triangle inequality.

Based on \eqref{lowerbound_difference_mmse_mean} and \eqref{lower_bound_mean_mase_outside_internal}, we have
\begin{align}\label{difference_mean_mmse_new}
    &\mathsf{MMSE}\big( \gdot, \pare\big) - \mathsf{MSE}\big(\mathsf{mean}(.,.), \gdot, \pare\big) \nonumber \\ 
     &\geq \frac{-(\eta^2+4)\Delta^2}{4} \mathbb{E}_{\mathbf{\underline{y}}|\acce}\Big [ \mathbbm{1}(|y_1+y_2| > 2M - (\eta+2)\Delta)\Big ] \nonumber \\
     & = \frac{-(\eta^2+4)\Delta^2}{4} \Pr \big( |\by_1 + \by
_2 | > 2M - (\eta+2)\Delta ~\big|  \acce\big).
\end{align}

Note that, given the acceptance of the inputs, for the case of $|u| < M - (\eta+2)\Delta$ we have
\begin{align}
    |y_1 + y_2| &= |2u + n_a + n_h| \leq |2u| + |n_a - n_h| + |2n_h| \nonumber \\
    & \leq \big( 2M - 2(\eta+2)\Delta \big) + \eta \Delta + 2\Delta \nonumber \\
    & = 2M - \big( \eta + 2\big)\Delta.
\end{align}
This implies that
\begin{align}
    \Pr &\big( |\by_1 + \by
_2 | \leq 2M - (\eta+2)\Delta ~\big|  \acce\big)\nonumber \\
&\geq \Pr (|\bu| < M - (\eta+2)\Delta) \nonumber \\
&= 1 - \frac{(\eta+2)\Delta}{M},
\end{align}
and thus
\begin{align}\label{bound_of_being_internal}
    \Pr \big( |\by_1 + \by
_2 | > 2M - (\eta+2)\Delta ~\big|  \acce\big) \leq \frac{(\eta+2)\Delta}{M}.
\end{align}

Based on \eqref{difference_mean_mmse_new} and \eqref{bound_of_being_internal}, we have
\begin{align}
    \mathsf{MMSE}\big( \gdot, \pare\big) &- \mathsf{MSE}\big(\mathsf{mean}(.,.), \gdot, \pare\big)\nonumber \\
    &\geq 
    \frac{-(\eta^2+4)(\eta+2)\Delta^3}{M}.
\end{align}
This completes the proof of Lemma \ref{bound_mse_symmetric_mean}.

\section{Proof of Lemma \ref{general_format_symmetric}}\label{proof:general_format_symmetric}
    First we show \eqref{general_symmetric_acc}. For any $\gdot \in \Lambda_{\mathsf{AD}}$, note that
\begin{align}\label{acc_v1}
    \mathsf{PA} \left( \gdot, \pare \right) = \int_{-\infty}^{\infty}\Pr \big( \acce| \bn_a = z \big) g(z) \,dz.
\end{align}
Recall that the acceptance rule is $|\mathbf{y}_1 - \mathbf{y}_2| \leq \eta \Delta $, or equivalently $|\mathbf{n}_a - \mathbf{n}_h| \leq \eta\Delta$. Also, we assume that $\Pr (|\bn_h| > \Delta) = 0$. Thus, for the case of $|z| < (\eta-1)\Delta$, $\Pr \big( \acce| \bn_a = z \big) = 1$. Also, for the case of $|z| > (\eta+1)\Delta$, $\Pr \big( \acce| \bn_a = z \big) = 0$. Additionally, for the case of $(\eta-1)\Delta \leq z \leq (\eta+1)\Delta$, we have
\begin{align}\label{conditional_acc_positive_bounded}
    \Pr \big( \acce| \bn_a = z \big) &= 
    \Pr \big( \bn_h \in [z-\eta\Delta, \Delta] \big)
    \nonumber \\
    &=\int_{z-\eta\Delta}^{\Delta} f_{\bn_h}(x) \,dx.
\end{align}
Also, in the case of 
$-(\eta+1)\Delta \leq z \leq -(\eta-1)\Delta$, we have
\begin{align}
    \Pr \big( \acce| \bn_a = z \big) &= 
    \Pr \big( \bn_h \in [-\Delta, z+\eta\Delta] \big)
    \nonumber \\
    &=\int_{-\Delta}^{z+\eta\Delta} f_{\bn_h}(x) \,dx \nonumber \\
    &= 
    \int_{-z-\eta\Delta}^{\Delta} f_{\bn_h}(x) \,dx = \Pr \big( \acce| \bn_a = -z \big).
\end{align}
This implies that $\Pr \big( \acce| \bn_a = z \big)$ is a symmetric function with respect to $z$. Therefore, we can rewrite \eqref{acc_v1} as 
\begin{align}
    \mathsf{PA} \left( \gdot, \pare \right) &= \int_{-\infty}^{\infty}\Pr \big( \acce| \bn_a = z \big) g(z) \,dz \nonumber \\
    &= 2\int_{0}^{\infty}\Pr \big( \acce| \bn_a = z \big) g(z) \,dz \nonumber \\
    & \overset{(a)}{=} 2\int_{0}^{(\eta+1)\Delta}\Pr \big( \acce| \bn_a = z \big) g(z) \,dz,
    \label{general_acc_for_symmetric} \\
    & \overset{(b)}{=} 2 \int_{0}^{(\eta-1)\Delta} g(z) \,dz \nonumber \\
    &\quad+ 2\int_{(\eta-1)\Delta}^{(\eta+1)\Delta}\Pr \big( \acce| \bn_a = z \big) g(z) \,dz, \label{semi_general_acc_symmetric} \\
    & \overset{(c)}{=} 2 \int_{0}^{(\eta-1)\Delta} g(z) \,dz \nonumber \\
    &\quad + 2\int_{(\eta-1)\Delta}^{(\eta+1)\Delta} \left( 
  \int_{z-\eta\Delta}^{\Delta} f_{\bn_h}(x) \,dx\right) g(z) \,dz,
\end{align}
where (a) follows from the fact that for the case of $|z| > (\eta+1)\Delta$, $\Pr \big( \acce| \bn_a = z \big) = 0$, and (b) follows from the fact that 
for the case of $|z| < (\eta-1)\Delta$, $\Pr \big( \acce| \bn_a = z \big) = 1$, and (c) follows from \eqref{conditional_acc_positive_bounded}. This completes the proof of \eqref{general_symmetric_acc}.

Now we prove \eqref{general_symmetric_mse}. For any $\gdot \in \Lambda_{\mathsf{AD}}$, note that
\begin{align}\label{costV1}
    &\mathsf{MSE}\big(\mathsf{mean}(.,.), \gdot, \pare \big) \nonumber \\
    &=\mathbb{E}\big[\big(\mathbf{u} - \frac{\by_1+\by_2}{2}\big)^2 | \mathcal{A}_{\pare}\big] =\frac{1}{4}\mathbb{E}\big[\big(\mathbf{n}_h + \mathbf{n}_a\big)^2 | \mathcal{A}_{\pare}\big]\nonumber \\
    &=\frac{1}{4}\int_{-\infty}^{\infty} \mathbb{E}[(\mathbf{n}_h + z)^2 \mid \mathcal{A}_{\pare}, \mathbf{n}_a = z]f_{\mathbf{n}_a|\mathcal{A}_{\pare}}(z|\acce) \,dz \nonumber \\
    &=\frac{1}{4}\int_{-\infty}^{\infty} \mathbb{E}[(\mathbf{n}_h + z)^2 \mid \mathcal{A}_{\pare}, \mathbf{n}_a = z]\frac{\Pr(\mathcal{A}_{\pare}|\mathbf{n}_a =z)g(z)}{\mathsf{PA} \left( \gdot, \pare \right)} \,dz \nonumber \\
    &=\frac{1}{4\mathsf{PA} \left( \gdot, \pare \right)}\int_{-\infty}^{\infty} \mathbb{E}[(\mathbf{n}_h + z)^2 \mid \mathcal{A}_{\pare}, \mathbf{n}_a = z]\nonumber \\
    &\quad \times \Pr(\mathcal{A}_{\pare}|\mathbf{n}_a =z)g(z) \,dz \nonumber \\
    &\overset{(a)}{=}\frac{1}{4\mathsf{PA} \left( \gdot, \pare \right)}\int_{-(\eta+1)\Delta}^{(\eta+1)\Delta} \mathbb{E}[(\mathbf{n}_h + z)^2 \mid \mathcal{A}_{\pare}, \mathbf{n}_a = z]\nonumber \\
    &\quad \times \Pr(\mathcal{A}_{\pare}|\mathbf{n}_a =z)g(z) \,dz,
    \end{align}
    where (a) follows from the fact that for $|z| > (\eta+1)\Delta$, $\Pr \big( \acce| \bn_a = z \big) = 0$. Note that 
    \begin{align}
        &\mathbb{E}[(\mathbf{n}_h + z)^2 \mid \mathcal{A}_{\pare}, \mathbf{n}_a = z]\Pr \big( \acce| \bn_a = z \big) \nonumber \\
        &= \int_{-\Delta}^{\Delta} (x+z)^2f_{\bn_h|\bn_a,\acce}(x|z,\acce) \Pr \big( \acce| \bn_a = z \big) \,dx  \ \nonumber \\
        &= \int_{-\Delta}^{\Delta} (x+z)^2 \Pr \big( \acce| \bn_a = z \big)\nonumber \\
        &\quad \times \frac{\Pr \big( \acce| \bn_a = z , \bn_h = x\big)f_{\bn_a|\bn_h}(z|x)f_{\bn_h}(x)}{\Pr \big( \acce| \bn_a = z \big)g(z)}\,dx \nonumber \\
        &= \int_{-\Delta}^{\Delta} (x+z)^2 \nonumber \\
        &\quad \times \frac{\Pr \big( \acce| \bn_a = z , \bn_h = x\big)f_{\bn_a|\bn_h}(z|x)f_{\bn_h}(x)}{g(z)}\,dx \nonumber \\
        & \overset{(a)}{=} \int_{-\Delta}^{\Delta} (x+z)^2\Pr \big( \acce| \bn_a = z , \bn_h = x\big)f_{\bn_h}(x) \,dx
    \end{align}
    where (a) follows from the fact that $\bn_a$ and $\bn_h$ are independent and thus $f_{\bn_a|\bn_h}(z|x)  = g(z)$. Therefore, we can rewrite \eqref{costV1} as  
    \begin{align}\label{costV2}
        &\mathsf{MSE}\big(\mathsf{mean}(.,.), \gdot, \pare \big) \nonumber \\
        &= \frac{1}{4\mathsf{PA} \left( \gdot, \pare \right)}\int_{-(\eta+1)\Delta}^{(\eta+1)\Delta} \int_{-\Delta}^{\Delta} (x+z)^2 \nonumber \\
        &\quad \times \Pr \big( \acce| \bn_a = z , \bn_h = x)f_{\bn_h}(x)g(z)\,dx  \,dz.
    \end{align}
    Note that for the case of $|z| \leq (\eta-1)\Delta$, $\Pr \big( \acce| \bn_a = z , \bn_h = x) = 1$, thus,
    \begin{align}\label{mse_fixed_noisev1}
        &\int_{-\Delta}^{\Delta} (x+z)^2\Pr \big( \acce| \bn_a = z , \bn_h = x)f_{\bn_h}(x)\,dx \nonumber \\
        &= \int_{-\Delta}^{\Delta} (x+z)^2f_{\bn_h}(x)\,dx = z^2 + \sigma^2_{\bn_h},
    \end{align}
    which is a symmetric function with respect to $z$. Additionally, for the case of $(\eta-1)\Delta \leq z \leq (\eta+1)\Delta$, we have
    \begin{align}\label{mse_fixed_noisev2}
        &\int_{-\Delta}^{\Delta} (x+z)^2\Pr \big( \acce| \bn_a = z , \bn_h = x)f_{\bn_h}(x)\,dx \nonumber \\
        &= \int_{z-\eta\Delta}^{\Delta} (x+z)^2f_{\bn_h}(x)\,dx.
    \end{align}
    Also, for the case of $-(\eta+1)\Delta \leq z \leq -(\eta-1)\Delta$, we have
    \begin{align}\label{mse_fixed_noisev3}
        &\int_{-\Delta}^{\Delta} (x+z)^2\Pr \big( \acce| \bn_a = z , \bn_h = x)f_{\bn_h}(x)\,dx \nonumber \\
        &= \int_{-\Delta}^{z+\eta\Delta} (x+z)^2f_{\bn_h}(x)\,dx.
    \end{align}

            Combining \eqref{mse_fixed_noisev1}, \eqref{mse_fixed_noisev2}, and \eqref{mse_fixed_noisev3}, one can verify that $\int_{-\Delta}^{\Delta} (x+z)^2\Pr \big( \acce| \bn_a = z , \bn_h = x)f_{\bn_h}(x)\,dx$ is a symmetric function with respect to $z$, for the case of $|z| \leq (\eta+1)\Delta$.
    Based on this observation, we can rewrite \eqref{costV2} as 
    \begin{align}
        &\mathsf{MSE}\big(\mathsf{mean}(.,.), \gdot, \pare \big) \nonumber \\
        &=  \frac{1}{2\mathsf{PA} \left( \gdot, \pare \right)}\int_{0}^{(\eta+1)\Delta}\int_{-\Delta}^{\Delta} (x+z)^2\nonumber \\
        &\quad \times \Pr \big( \acce| \bn_a = z , \bn_h = x)f_{\bn_h}(x)g(z)\,dx  \,dz \label{mse_mean_general_proxy} \\
         &\overset{(a)}{=} \frac{1}{2\mathsf{PA} \left( \gdot, \pare \right)}\int_{0}^{(\eta-1)\Delta} (z^2 + \sigma^2_{\bn_h})g(z) \,dz \nonumber \\
        &+\frac{1}{2\mathsf{PA} \left( \gdot, \pare \right)}\int_{(\eta-1)\Delta}^{(\eta+1)\Delta} \int_{z-\eta\Delta}^{\Delta} (x+z)^2f_{\bn_h}(x)g(z)\,dx  \,dz
    \end{align}
    where (a) follows from \eqref{mse_fixed_noisev1} and \eqref{mse_fixed_noisev2}. This completes the proof of \eqref{general_symmetric_mse}.

\section{Proof of Lemma \ref{lemma:bounded_noise_existence}}\label{proof:lemma:bounded_noise_existence}
    To prove this lemma we first show that there exist a symmetric noise distribution $g_0(.)$, such that $g_0(.)=0$ for $|z| > (\eta+1)\Delta$, and 
    \begin{align*}
        \PA(g_0(.), \eta) &\geq \PA(g_1(.), \eta),  \\
        \mathsf{MSE}(\mathsf{mean}(.,.), g_0(.), \eta) &= \mathsf{MSE}(\mathsf{mean}(.,.), g_1(.), \eta) .
    \end{align*}
    To prove this, note that
    since $g_1(.)$ is symmetric, we have $\int_{0}^{\infty} g_1(z) \,dz = \frac{1}{2}$.
    Let  $b \triangleq \int_{(\eta+1)\Delta}^{\infty} g_1(z) \,dz$. Since $\PA(g_1(.), \eta) > 0$, based on \eqref{general_symmetric_acc}, it implies that $b < \frac{1}{2}$. Define  a symmetric noise distribution $g_0(.)$, where $g_0(z) = \frac{g_1(z)}{1 - 2b}$, for all $|z| \leq (\eta+1)\Delta$, and otherwise $g_0(z) = 0$.
Based on \eqref{general_symmetric_acc}, we have
\begin{align}\label{new_noise_bounded_acc}
    \mathsf{PA} \left( g_0(.), \pare \right) = \frac{\mathsf{PA} \left( g_1(.), \pare \right)}{1-2b} \geq \mathsf{PA} \left( g_1(.), \pare \right).
\end{align}
On the other hand, Based on \eqref{general_symmetric_mse}, we have
\begin{align}\label{new_noise_bounded_mse}
    \mathsf{MSE}\big(\mathsf{mean}(.,.), g_0(.), \pare \big) = \mathsf{MSE}\big(\mathsf{mean}(.,.), g_1(.), \pare \big).
\end{align}

     For any $x \in \mathbb{R}$ define the function $u(z,x): \mathbb{R} \to \mathbb{R}$ such that $u(z,x) = 1$, if and only if $z\geq x$, and otherwise we have   $u(z,x) = 0$. Define  a symmetric noise distribution $g_2(.)$, where for $z \geq 0$, we have
     \begin{align}\label{def_of_g_0}
         g_2(z) &= \delta(z - (\eta-1)\Delta)  \times \int_{0}^{(\eta-1)\Delta} g_0(z) \,dz \nonumber \\
         &\quad + g_0(z)  \big(u(z,(\eta-1)\Delta) - u(z,(\eta+1)\Delta)\big),
     \end{align}
     where $\delta(z)$ is a Dirac's delta function. Note that $g_2(.)$ is a symmetric noise distribution and $g_2(.)=0$ for
     $|z| < (\eta-1)\Delta$,   
    $|z| > (\eta+1)\Delta$. We show that it satisfies \eqref{betternoise_acc} and \eqref{betternoise_mse}.
    Based on \eqref{general_acc_for_symmetric}, we have \begin{align}
        &\mathsf{PA} \left( g_2(.), \pare \right) \nonumber \\
        &= 2\int_{0}^{(\eta+1)\Delta}\Pr \big( \acce| \bn_a = z \big) g_2(z) \,dz, \nonumber \\
        &\overset{(a)}{=}2\int_{0}^{(\eta+1)\Delta}\Pr \big( \acce| \bn_a = z \big) \nonumber \\
        &\quad \times \bigg( \delta(z - (\eta-1)\Delta) \times \int_{0}^{(\eta-1)\Delta} g_0(z) \,dz\bigg) \,dz \nonumber\\
        &\quad +2\int_{0}^{(\eta+1)\Delta}\Pr \big( \acce| \bn_a = z \big) \nonumber \\
        &\quad \times \bigg( g_0(z)  \big(u(z,(\eta-1)\Delta) - u(z,(\eta+1)\Delta)\big)\bigg) \,dz
        \nonumber\\
        &= 2 \Pr \big( \acce| \bn_a = (\eta-1)\Delta \big) \int_{0}^{(\eta-1)\Delta} g_0(z) \,dz \nonumber \\
        &\quad+ 2\int_{(\eta-1)\Delta}^{(\eta+1)\Delta}\Pr \big( \acce| \bn_a = z \big) g_0(z) \,dz
        \nonumber \\
        &\overset{(b)}{=} 2 \int_{0}^{(\eta-1)\Delta} g_0(z) \,dz \nonumber \\
        &\quad+ 2\int_{(\eta-1)\Delta}^{(\eta+1)\Delta}\Pr \big( \acce| \bn_a = z \big) g_0(z) \,dz \nonumber \\
        & \overset{(c)}{=} \mathsf{PA} \left( g_0(.), \pare \right)\label{g_0_relation_to_g_2} \\
        &\overset{(d)}{\geq} \mathsf{PA} \left( g_1(.), \pare \right) \label{new_noise_same_acc}
    \end{align}
    where (a) follows from  \eqref{def_of_g_0}, and (b) follows from the fact that acceptance rule is $|\mathbf{y}_1 - \mathbf{y}_2| \leq \eta \Delta $, or equivalently $|\mathbf{n}_a - \mathbf{n}_h| \leq \eta\Delta$, and we assume that $\Pr (|\bn_h| > \Delta) = 0$, which implies that $\Pr \big( \acce| \bn_a = (\eta-1)\Delta \big) = 1$, and (c)
    follows from \eqref{semi_general_acc_symmetric}, and (d) follows from \eqref{new_noise_bounded_acc}.
    
    On the other hand, based on \eqref{mse_mean_general_proxy}, we have
 \begin{align}\label{new_noise_same_mse}
        &\mathsf{MSE}\big(\mathsf{mean}(.,.), g_2(.), \pare \big) \nonumber\\
        &=\frac{1}{2\mathsf{PA} \left( g_2(.), \pare \right)}\int_{0}^{(\eta+1)\Delta} \int_{-\Delta}^{\Delta} (x+z)^2\nonumber \\
        &\quad \times \Pr \big( \acce| \bn_a = z , \bn_h = x)f_{\bn_h}(x)g_2(z)\,dx  \,dz 
        \nonumber\\
        &\overset{(a)}{=} \frac{1}{2\mathsf{PA} \left( g_2(.), \pare \right)} \int_{0}^{(\eta+1)\Delta} 
        \int_{-\Delta}^{\Delta} (x+z)^2 \nonumber \\
        &\quad \times \Pr \big( \acce| \bn_a = z , \bn_h = x)f_{\bn_h}(x)\,dx \nonumber \\
        &\quad \times \left( \delta(z - (\eta-1)\Delta) \int_{0}^{(\eta-1)\Delta} g_0(z) \,dz \right) \,dz
        \nonumber \\
        &+\frac{1}{2\mathsf{PA} \left( g_2(.), \pare \right)} \int_{0}^{(\eta+1)\Delta} 
        \int_{-\Delta}^{\Delta} (x+z)^2 \nonumber\\
        &\quad \times \Pr \big( \acce| \bn_a = z , \bn_h = x)f_{\bn_h}(x)\,dx \nonumber\\
        &\quad \times \left( g_0(z)\left(u(z,(\eta-1)\Delta) - u(z,(\eta+1)\Delta)\right) \right) \,dz
        \nonumber\\
        &\overset{(b)}{=} \frac{1}{2\mathsf{PA} \left( g_2(.), \pare \right)} 
        \int_{-\Delta}^{\Delta} (x+(\eta-1)\Delta)^2 \nonumber\\
        &\quad \times \Pr \big( \acce| \bn_a = (\eta-1)\Delta , \bn_h = x)
        f_{\bn_h}(x)\,dx \nonumber\\
        &\quad \times \left( \int_{0}^{(\eta-1)\Delta} g_0(z) \,dz \right) \nonumber \\
        &+\frac{1}{2\mathsf{PA} \left( g_2(.), \pare \right)}\int_{(\eta-1)\Delta}^{(\eta+1)\Delta} \int_{z-\eta\Delta}^{\Delta} (x+z)^2f_{\bn_h}(x)g_0(z)\,dx  \,dz \nonumber \\
        &\overset{(c)}{=}\frac{1}{2\mathsf{PA} \left( g_2(.), \pare \right)}((\eta-1)^2 + \sigma^2_{\bn_h}) \int_{0}^{(\eta-1)\Delta} g_0(z) \,dz \nonumber \\
        &+\frac{1}{2\mathsf{PA} \left( g_2(.), \pare \right)}\int_{(\eta-1)\Delta}^{(\eta+1)\Delta} \int_{z-\eta\Delta}^{\Delta} (x+z)^2f_{\bn_h}(x)g_0(z)\,dx  \,dz \nonumber \\
        &\geq \frac{1}{2\mathsf{PA} \left( g_2(.), \pare \right)} \int_{0}^{(\eta-1)\Delta} (z^2 + \sigma^2_{\bn_h})g_0(z) \,dz \nonumber \\
        &+\frac{1}{2\mathsf{PA} \left( g_2(.), \pare \right)}\int_{(\eta-1)\Delta}^{(\eta+1)\Delta} \int_{z-\eta\Delta}^{\Delta} (x+z)^2f_{\bn_h}(x)g_0(z)\,dx  \,dz 
        \nonumber \\
        &\overset{(d)}{=}\frac{1}{2\mathsf{PA} \left( g_0(.), \pare \right)} \int_{0}^{(\eta-1)\Delta} (z^2 + \sigma^2_{\bn_h})g_0(z) \,dz \nonumber \\
        &+\frac{1}{2\mathsf{PA} \left( g_0(.), \pare \right)}\int_{(\eta-1)\Delta}^{(\eta+1)\Delta} \int_{z-\eta\Delta}^{\Delta} (x+z)^2f_{\bn_h}(x)g_0(z)\,dx  \,dz 
        \nonumber \\
        &\overset{(e)}{=}  \mathsf{MSE}\big(\mathsf{mean}(.,.), g_0(.), \pare \big), \nonumber \\
        & \overset{(f)}{=}  \mathsf{MSE}\big(\mathsf{mean}(.,.), g_1(.), \pare \big)
\end{align}
where (a) follows from \eqref{def_of_g_0}, (b) follows from \eqref{mse_fixed_noisev2}, (c) follows from \eqref{mse_fixed_noisev1}, (d) follows from \eqref{g_0_relation_to_g_2}, (e) follows from \eqref{general_symmetric_mse}, (f) follows from \eqref{new_noise_bounded_mse}. 
    
Combining \eqref{new_noise_same_acc} and \eqref{new_noise_same_mse} completes the proof of Lemma \ref{lemma:bounded_noise_existence}.

\section{Proof of Lemma \ref{lemma:exact_acc_noise_existence}}\label{proof:lemma:exact_acc_noise_existence}
Let $g_1(.)$ be a satisfying noise of Lemma \ref{lemma:bounded_noise_existence} and $\mathsf{PA} \left( g_1(.), \pare \right) = \alpha_1 > \alpha$.
    Define   satisfying noise of Lemma \ref{lemma:bounded_noise_existence} $g_2(.)$, where $g_2(z) = \frac{\alpha}{\alpha_1}g_1(z)$, for all $|z| \leq (\eta+1)\Delta$, $g_2((\eta+2)\Delta) = \frac{\alpha_1-\alpha}{2\alpha_1}\delta(z - (\eta+2)\Delta)$, and otherwise, $g_2(z) = 0$. Based on \eqref{general_symmetric_acc}, one can verify that
    \begin{align}\label{new_noise_exactprob_acc}
        \mathsf{PA} \left( g_2(.), \pare \right) = \frac{\alpha}{\alpha_1}\mathsf{PA} \left( g_1(.), \pare \right) = \alpha.
    \end{align}
    Additionally,
    based on \eqref{general_symmetric_mse}, we have
\begin{align}\label{new_noise_exactprob_mse}
    \mathsf{MSE}\big(\mathsf{mean}(.,.), g_2(.), \pare \big) = \mathsf{MSE}\big(\mathsf{mean}(.,.), g_1(.), \pare \big).
\end{align}
    Combining \eqref{new_noise_exactprob_acc} and \eqref{new_noise_exactprob_mse} completes the proof of Lemma \ref{lemma:exact_acc_noise_existence}.

\section{\texorpdfstring{Characterizing $h_{\eta}(z)$: An special Case}{X}}\label{Characterizing_h_eta}
In this appendix we characterize $h_{\eta}(z)$ and its concave envelop, for the specific case of $\mathbf{n}_h \sim \text{unif}[-\Delta, \Delta]$. Recall that $k_{\eta}(z) \triangleq \int_{z-\eta\Delta}^{\Delta} f_{\bn_h}(x)  \,dx$, where $z \in [(\eta-1)\Delta, (\eta+1)\Delta]$. For the case of $\mathbf{n}_h \sim \text{unif}[-\Delta, \Delta]$, it implies that
\begin{align}
    k_{\eta}(z) = \frac{(\eta+1)\Delta-z}{2\Delta}.
\end{align}
Thus, for $0 \leq q \leq 1$ we have
\begin{align}
    k^{-1}_{\eta}(q) = \Delta(\eta+1-2q).
\end{align}
On the other hand, we have
\begin{align}
    \nu_{\eta}(z) &= \int_{z-\eta\Delta}^{\Delta} (x+z)^2f_{\bn_h}(x)\,dx \nonumber \\
    &= \frac{(\Delta+z)^3 - (2z-\eta \Delta)^3}{6\Delta}.
\end{align}
Therefore, we have
\begin{align}
    h_{\eta}(q) &= \nu_{\eta}(k_{\eta}^{-1} (q)) \nonumber \\
    &= \Delta^2 \frac{(\eta+2-2q)^3 - (\eta+2-4q)^3}{6}.
\end{align}
This implies that
\begin{align}
    \frac{d}{dq}h_{\eta}(q) = \Delta^2(28q^2 - 12q(\eta+2) + (\eta+2)^2),
\end{align}
and 
\begin{align}\label{second_derivitive}
    \frac{d^2}{dq^2}h_{\eta}(q) = \Delta^2(56q - 12\eta - 24).
\end{align}
Based on \eqref{second_derivitive}, for the case of $\eta \geq \frac{8}{3}$, and $0 \leq q \leq 1$, we have
\begin{align}
    \frac{d^2}{dq^2}h_{\eta}(q) &= \Delta^2(56q - 12\eta - 24) \nonumber \\
    &\leq \Delta^2(56 - 12\frac{8}{3} - 24) = 0.
\end{align}
Therefore, for the case of $\eta \geq \frac{8}{3}$, and $0 \leq q \leq 1$, $h_{\eta}(q)$ is a concave function, leading to  $h^*_{\eta}(q) = h_{\eta}(q)$. On the other hand, for the case of $2 \leq \eta < \frac{8}{3}$, based on \eqref{second_derivitive}, one can easily verify that for the case of $\frac{12\eta+24}{56} < q \leq 1$,  $h_{\eta}(q)$ is convex, and for the case of $0 \leq q \leq \frac{12\eta+24}{56}$, it is concave. In order to find the concave envelop, i.e., $h_{\eta}^*(q)$ for this case, we draw a line passing through the point $\big(1,  h_{\eta}(1) \big) $  and find the point of tangency with $h_{\eta}(q)$. More precisely, we need to solve the following equation for $q$:
\begin{align}\label{tengency_equation}
  h_{\eta}(q) - h_{\eta}(1) = \big(\frac{d}{dq}h_{\eta}(q)\big)  (q-1).
\end{align}
Note that
\begin{align}\label{touching_point_equation}
    &\big( h_{\eta}(q) - h_{\eta}(1)\big)  - \big(\frac{d}{dq}h_{\eta}(q)\big)  (q-1) \nonumber \\
    &=  \Delta^2(q-1)\frac{28q^2 - 2q(9\eta+4) + 3\eta^2 - 6\eta + 4}{3} \nonumber \\
    &\quad -\big(\frac{d}{dq}h_{\eta}(q)\big)  (q-1) \nonumber \\
    & = \Delta^2(q-1) \big( \frac{28q^2 - 2q(9\eta+4) + 3\eta^2 - 6\eta + 4}{3} \nonumber \\
    &\quad - 28q^2 - 12q(\eta+2) + (\eta+2)^2 \big) \nonumber \\
    & = \frac{2\Delta^2(q-1)^2}{3}(-28q + 9\eta+4).
\end{align}
Based on \eqref{touching_point_equation} the point of tangency is $\frac{9\eta+4}{28}$. Using this, one can easily calculate the concave envelop using the details of the proof of Lemma \ref{lemma:lower_bound_of_mean}.  In particular, if $\eta \geq \frac{8}{3}$ or $0 \leq q \leq \frac{9\eta+4}{28}$, we have $h^*(q) = h(q)$. Otherwise, we have 
    \begin{align}
        h^*_{\eta}(q) = \frac{h_{\eta}(1) - h_{\eta}(\frac{9\eta+4}{28})}{1 - \frac{9\eta+4}{28}} (q - \frac{9\eta+4}{28}) + h_{\eta}(\frac{9\eta+4}{28}).
    \end{align}

\section{Proof of Lemma \ref{lemma:more_on_remark:mean_is_good_for_all_best_response}}\label{appendix:lemma:more_on_remark:mean_is_good_for_all_best_response}
Here in this appendix, we prove that for all of the choices of the adversary's noise distributions that achieve an equilibrium, the $\mathsf{mean}$ estimator is optimum, within a vanishing gap. More precisely, for any $\eta \in \Lambda_{\DC}$ and $g^*(.) \in \mathcal{B}^{\eta}_{\mathsf{AD}}$ we have 
\begin{align}
     &\frac{-(\eta^2+4)(\eta+2)\Delta^3}{M} \nonumber \\
     &\leq \mathsf{MMSE}\big( g^*(.) , \pare\big) - \mathsf{MSE}\big(\mathsf{mean}(.,.), g^*(.) , \pare\big)\leq 0.
    \end{align}
    To prove this, note that if $g^*(.) \in \Bar{\Lambda}_{\mathsf{AD}}$, i.e., it is symmetric, based on Lemma \ref{bound_mse_symmetric_mean} this claim is valid.

    Now consider the case that $g^*(.)$ is not symmetric.  Based on definition \eqref{J_definition}
    \begin{align}\label{belonging_to_j}
        \bigl( 
    \mathsf{MMSE}\big(g^*(.),  \eta \big), 
    \mathsf{PA} \left( g^*(.), \pare \right)\big) \in \mathcal{J}_{\eta}.
    \end{align}
    Based on Lemma \ref{lemmaJSC}, \eqref{belonging_to_j} implies that 
    \begin{align}\label{belonging_to_c}
        \bigl( 
    \mathsf{MMSE}\big(g^*(.),  \eta \big), 
    \mathsf{PA} \left( g^*(.), \pare \right)\big) \in \mathcal{C}_{\eta}.
    \end{align}
    Let $\alpha \triangleq \PA (g^*(.), \eta)$. Based on \eqref{defining_the_Set_C_eta}, \eqref{C_definition}, and \eqref{belonging_to_c}  we have
    \begin{align}\label{best_noise_mmse}
       \mathsf{MMSE}\big(g^*(.),  \eta \big) =  \underset{\gdot \in \Lambda_{\mathsf{AD}}}{\max} ~ \underset{\mathsf{PA} \left( \gdot, \pare \right) \geq \alpha}{\mathsf{MMSE}\left(\gdot, \pare \right)}
    \end{align}
    
    Let\footnote{For the case that   \eqref{best_symmetric_noise_definition_appendix} has more than one solution, we just pick one them randomly.}
    \begin{align}
        g^*_s(.) &= \underset{\gdot \in \Bar{\Lambda}_{\mathsf{AD}}}{\arg \max} ~ \underset{\mathsf{PA} \left( \gdot, \pare \right) \geq \alpha}{\mathsf{MMSE}\big( \gdot, \pare \big)}, \label{best_symmetric_noise_definition_appendix} \\
        g_{\textrm{sym}}(z) &\triangleq \frac{g^*(z) + g^*(-z)}{2}.
    \end{align}
    
To complete the proof, note that
    \begin{align}\label{non_symmetic_mean_is_best}
         \mathsf{MMSE}\big(g^*(.),  \pare \big) 
    &\leq \mathsf{MSE}\big(\mathsf{mean}(.,.), g^*(.),  \pare \big) \nonumber \\
    &\overset{(a)}{=}\mathsf{MSE}\big(\mathsf{mean}(.,.), g^*_{\textrm{sym}}(.), \pare \big) \nonumber \\
    &\overset{(b)}{\leq}\underset{\gdot \in \Bar{\Lambda}_{\mathsf{AD}}}{\max} ~ \underset{\mathsf{PA} \left( \gdot, \pare \right) \geq \alpha}{\mathsf{MSE}\big(\mathsf{mean}(.,.), \gdot, \pare \big)} \nonumber \\
    &\overset{(c)}{=} \mathsf{MSE}\big(\mathsf{mean}(.,.), g^*_s(.), \pare \big) \nonumber \\
    & \overset{(d)}{\leq} \mathsf{MMSE}\big( g^*_s(.), \pare \big) + \frac{(\eta^2+4)(\eta+2)\Delta^3}{M} \nonumber \\
    & \overset{(e)}{\leq} \mathsf{MMSE}\big(g^*(.),  \pare \big) + \frac{(\eta^2+4)(\eta+2)\Delta^3}{M} 
    \end{align}

    where (a) follows from Lemma \ref{lemma:making_symmetry}, (b) follows from the fact that $g^*_{\textrm{sym}}(.) \in \Bar{\Lambda}_{\mathsf{AD}}$, (c) follows from \eqref{best_symmetric_noise_definition_appendix}, (d) follows from Lemma \ref{bound_mse_symmetric_mean}, (e) follows from \eqref{best_noise_mmse} and \eqref{best_symmetric_noise_definition_appendix}.

    Note that \eqref{non_symmetic_mean_is_best} implies that
    \begin{align}
        \mathsf{MMSE}\big(g^*(.),  \pare \big) &\leq \mathsf{MSE}\big(\mathsf{mean}(.,.), g^*_s(.), \pare \big) \nonumber \\ &\leq 
        \mathsf{MMSE}\big(g^*(.),  \pare \big) + \frac{(\eta^2+4)(\eta+2)\Delta^3}{M},
    \end{align}
    which completes the proof.

\end{document}